\documentclass[preprint,12pt,3p]{elsarticle}

\usepackage{makeidx}
\usepackage{amsmath}
\usepackage{amsthm}
\usepackage{amssymb}
\usepackage{amsfonts}
\usepackage{dsfont}
\usepackage{float}
\usepackage[plainpages=false,pdfpagelabels=true,colorlinks=true,citecolor=blue,hypertexnames=false]{hyperref}
\usepackage{mathtools}
\usepackage{extarrows}
\usepackage{url}
\usepackage{balance}
\usepackage{multirow}
\usepackage{color}
\usepackage{diagbox}
\usepackage{booktabs}
\usepackage{makecell}
\usepackage[switch]{lineno}
\newcolumntype{V}{!{\vrule width 1pt}}

\usepackage{caption}
\usepackage{amsfonts}

\newcommand{\bA}{{\mathbb{A}}}
\newcommand{\bE}{{\mathbb{E}}}
\newcommand{\bF}{{\mathbb{F}}}
\newcommand{\bZ}{{\mathbb{Z}}}
\newcommand{\bQ}{{\mathbb{Q}}}

\newcommand{\bK}{{\mathbb{K}}}
\newcommand{\bN}{{\mathbb{N}}}

\newcommand{\bI}{{\mathbb{I}}}
\newcommand{\gG}{{G}}
\newcommand{\bH}{{\mathbb{H}}}
\newcommand{\bT}{{\mathbb{T}}}

\newcommand{\cP}{{\mathcal{P}}}

\newcommand{\coeff}{\operatorname{coeff}}
\newcommand{\lc}{\operatorname{lc}}

\newcommand{\im}{\operatorname{im}}

\newcommand{\spa}{\operatorname{span}}

\newcommand{\den}{\operatorname{den}}
\newcommand{\num}{\operatorname{num}}

\newcommand{\tdeg}{\operatorname{tdeg}}
\newcommand{\hdeg}{\operatorname{hdeg}}

\newcommand{\const}{\operatorname{const}}
\newcommand{\ord}{\operatorname{ord}}
\newcommand{\directsum}{\sum}
\newcommand{\directsumideals}{\bigoplus}

\newtheorem{thm}{Theorem}[section]
\newtheorem{prop}[thm]{Proposition}
\newtheorem{cor}[thm]{Corollary}
\newtheorem{lemma}[thm]{Lemma}
\newtheorem{convention}[thm]{Convention}
\newtheorem{remark}[thm]{Remark}
\newtheorem{define}[thm]{Definition}
\newtheorem{example}[thm]{Example}

\newtheorem{hyp}[thm]{Hypothesis}
\newtheorem{alg}[thm]{Algorithm}

\numberwithin{equation}{section}

\journal{*****}

\begin{document}

\begin{frontmatter}



\title{Complete Reductions and Idempotent Representations for $R\Pi\Sigma^*$-towers}

\author[jku]{Yiman Gao}
\ead{ymgao@risc.jku.at}

\author[jku]{Jakob Obrovsky}
\ead{Jakob.Obrovsky@risc.jku.at}

\author[jku]{Carsten Schneider}
\ead{Carsten.Schneider@risc.jku.at}

\begin{flushleft}
\vspace*{-2cm}
		RISC Report 26--14 
\vspace*{1cm}
\end{flushleft}

\address[jku]{Johannes Kepler University Linz, Research Institute for Symbolic Computation (RISC), \\ Altenberger Stra\ss e 69, 4040, Linz, Austria}

\begin{abstract}
 $R\Pi\Sigma^*$-extensions form a rich class of difference rings that provide a unified algebraic framework for modeling indefinite nested sums, transcendental products, and nested products over roots of unity structures that frequently appear
 in combinatorics, number theory, and particle physics. For a large subclass of these extensions whose ring of constants is a field, we introduce a complete reduction approach to resolve the telescoping problem without solving any difference equations. More precisely, we explicitly construct a complement to the subspace of differences over the constant field and develop an  algorithm that decomposes any element of the extension into the sum of a difference and a component lying in this complement. Consequently, summability holds
 if and only if this complementary component is zero. This structural approach yields significant speed-ups for parameterized telescoping and, notably, creative telescoping for deriving linear recurrences of definite sums. Finally, we compute an explicit idempotent representation that extends existing telescoping algorithms and our complete reduction framework to the  general class of $R\Pi\Sigma^*$-extensions, opening up previously untreatable classes of sums and products.
\end{abstract}


\end{frontmatter}

\section{Introduction}\label{SECT:Intro}
The symbolic simplification of discrete sums is a foundational challenge across diverse branches of mathematics and theoretical physics, ranging from enumerative combinatorics and number theory to the analysis of algorithms and quantum field theory. Central to symbolic summation is the telescoping problem. 
Given an expression $f(k)$ that belongs to some domain of sequences, decide constructively  whether there exists a sequence $g(k)$ in the same domain such that
\begin{equation}\label{EQ:telescoping}
f(k)=g(k+1)-g(k).
\end{equation}
If such a solution can be found, one obtains the identity
$$\sum_{k=a}^{b}f(k)=g(b+1)-g(a)$$
by choosing appropriate non-negative integers $a$ and $b$.  
Algorithmic solutions to the telescoping problem have been developed for  rational functions \cite{Abra1971}, hypergeometric terms \cite{Gosp1978,PS1995}, $q$-hypergeometric terms \cite{Koor1993, PR1997} and  D-finite sequences~\cite{Zeil1990b,Chyz2000,Kout2009}. 
By developing a discrete analogue of Risch's  algorithm \cite{Risc1970} for
symbolic integration, Karr \cite{Karr1981,Karr1985} reformulated the above problem in terms of difference algebra, introduced  so-called $\Pi\Sigma^*$-fields covering wide classes of indefinite nested sums and products such as harmonic numbers and their generalized versions, and provided an algorithm for solving the telescoping problem in this setting. 

However, Karr's algorithm cannot handle algebraic products like $(-1)^n$ because of zero divisors, such as $1-(-1)^n$ and $1+(-1)^n$, which are incompatible with the underlying field and integral domain structures. Nevertheless, alternating signs and nested sums constructed with them~\cite{AB2013,ABS:2011,Verm1999} arise frequently as fundamental building blocks across combinatorics, number theory, and mathematical physics. Properly capturing and managing these objects is therefore essential for symbolic summation.
To overcome this drawback, $R\Pi\Sigma^*$-extensions were introduced in \cite{Schn2016, Schn2017} and the telescoping problem was solved for the subclass of simple $R\Pi\Sigma^*$-extensions (see Definition~\ref{DEF:SimpleRPiSigma} below). In this setting, one can model indefinite nested sums defined over indefinite nested products, where the products within the multiplicands may take the form of power products; see~\cite{Schn2021} and the literature therein. More generally, algorithms have been provided to solve the parameterized  telescoping problem, which encompasses Zeilberger's creative telescoping paradigm \cite{Zeil1990b} for producing linear recurrences. Furthermore, algorithms to find solutions of linear recurrences defined over such rings were introduced in~\cite{AS:2021,ABPS2021}.
Most of these
algorithms have been implemented in the Mathematica package \texttt{Sigma}, which has been applied to solving problems from combinatorics \cite{PSW2011}, number theory \cite{OS2009,BPRS2025} and elementary particle physics \cite{QCD20}. 

The telescoping algorithms mentioned above are limited in two aspects.
\begin{itemize}
    \item They rely on solving auxiliary parameterized linear difference equations, which often suffer from high computational costs.
     \item They do not apply to general $R\Pi\Sigma^*$-extensions which are not simple.
\end{itemize}

As an alternative, one can tackle the telescoping problem via additive decompositions: 
given a sequence $f(k)$ that belongs to some domain, compute $g(k)$ and $r(k)$ in the same domain such that\begin{equation}\label{Equ:RefinedTele}
f(k)=g(k+1)-g(k)+r(k),
\end{equation}
and $r(k)$ is minimal in some sense with the following extra property: $f(k)$ is summable if and only if $r(k)=0$. 
Additive decompositions provide us with an alternative and possibly more informative way to express sums. If one fails to find a solution $g(k)$ of~\eqref{EQ:telescoping}, summing~\eqref{Equ:RefinedTele} over $k$ from $a$ to $b$ yields
\begin{equation}\label{Equ:RefinedTeleSummed}
	\sum_{k=a}^bf(k)=g(b+1)-g(a)+\sum_{k=a}^b r(k),
\end{equation}
where the arising sum on the right hand side can be considered in many approaches to be simpler than the input sum on the left hand side.

Additive decompositions are well established for various classes of functions and sequences, including rational functions \cite{Abra1975,Paul1995}, ($q$)-hypergeometric terms  \cite{AbPe2002b,CHKL2015,DHL2018,CDGHL2025}, and, more generally, nested sums~\cite{Schn2007,CGHS2025}
and D-finite sequences \cite{BrSa2024, CDKW2025}.
While the underlying algorithms in~\cite{AbPe2002b,Schn2007} require explicitly solving difference equations within the underlying difference ring, the remaining approaches bypass these expensive computations entirely. 
Related reduction techniques have also been developed for symbolic summation in multivariate difference fields; see, e.g.,~\cite{CDFW2026,DW2026}.
Crucially, these newer constructions served as the key step toward refined versions --referred to as complete reductions -- which have been elaborated, e.g., for hypergeometric terms in~\cite[Theorem~5.78]{Kaue2023}. As it turns out, the additive decompositions for D-finite sequences in~\cite{BrSa2024,CDKW2025}, and for nested sums in~\cite{CGHS2025} also establish the refined construction of complete reductions.
Notably, this framework originated in the differential context of symbolic integration; we refer the reader to~\cite{Herm1872,Ostr1845, BLLRR2016, CKK2016,vdHJ2021, DGLL2025}  and the references therein for further details.

 Below, to facilitate our subsequent constructions, we use the following definition of complete reductions, which is equivalent to the one given in \cite[Definition 5.67]{Kaue2023}. 

\begin{define} \label{DEF:cr}
	Let $C$ be a field, $U$ be a $C$-linear space and $\rho$ be a $C$-linear operator on $U$.
	\begin{itemize}
		\item[(i)] Another  $C$-linear operator $\phi$ on $U$ is called a {\em complete reduction for $(U,\, \rho)$} if $U = \im(\rho) \oplus \im(\phi)$ and $\phi$ is the projection from $U$ to $\im(\phi)$ with respect to the direct sum\footnote{The first condition $U=\im(\rho) \oplus \im(\phi)$ means that every element $u$ in $U$ can be uniquely written as $q+w$ for $q\in \im(\rho)$ and $w\in \im(\phi)$; the second condition that $\phi$ is the projection from $U$ to $\im(\phi)$ with respect to the direct sum means that $\phi$ itself establishes this direct sum representation by mapping $u$ to the component $w$.}.
		
		\item[(ii)] Let $\phi$ be a complete reduction for $(U,\rho)$ and $u \in U$. If there is  $v \in U$ such that $\rho(v)=u-\phi(u)$, then $(v, \phi(u))$ is called a {\em $\Sigma$-pair with respect to $\phi$}.

		\item[(iii)] 	We say that {\em there is an algorithm for constructing a complete reduction
			for $(U, \rho)$} if we can fix a subspace $W$ such that $U=\im(\rho)\oplus W$, and there is an algorithm that, for every $u \in U$, computes 
		$v \in U$ and $w \in W$ such that $u = \rho(v) +w$. 
Note that such an algorithm induces a complete reduction for $(U, \rho)$ by mapping $u$ to $w$.
	\end{itemize} 
\end{define}

\begin{remark} \label{RE:cr}
	An easy linear algebra argument reveals that a $C$-linear operator $\phi$ is a complete reduction for $(U, \rho)$ 
	if and only if it is an idempotent whose kernel is equal to $\im(\rho)$.
\end{remark}

Reduction-based methods can be broadly defined in the context of difference rings as follows. Let $(\bA,\, \sigma)$ be a \emph{difference ring} (resp.\ \emph{difference field}), that is, a ring (resp.\ field) $\bA$ equipped with an automorphism $\sigma: \bA \to \bA$. Here the summation objects are represented formally in the ring $\bA$, and the shift operator is modeled by $\sigma$. An element $c \in \bA$ is called a \emph{constant} if $\sigma(c) = c$, and the set of all constants in $\bA$ is denoted by $C = \const(\bA,\, \sigma)$. In particular, if $C$ is a field (which is always the case when $\bA$ is a field), $C$ is referred to as the \emph{constant field} of $(\bA,\, \sigma)$. In this setting, the forward difference operator $\Delta: \bA \to \bA$ is defined by
\begin{equation}\label{EQ:differenceoperator} 
	\begin{array}{cccc}
		\Delta: & \bA & \longrightarrow & \bA \\
		&  f   & \mapsto         &  \sigma(f)-f.
	\end{array}
\end{equation}
The operator $\Delta$ is a $C$-linear map, and its image $\Delta(\bA)$ forms a $C$-subspace of $\bA$, which is called the \emph{summable space} of $\bA$. 
A complete reduction $\phi$ for $(\bA,\, \Delta)$ then satisfies the following key properties which naturally emerge from the definition; see also the existing literature. \\ 
(i) \textit{Decidability of summability:} If $f\in\bA$ is decomposed as $f=\Delta(g)+\phi(f)$ for some $g\in\bA$, then $f\in\Delta(\bA)$ if and only if $\phi(f)=0$.\\
(ii) \textit{Additivity:} If $f_1,f_2\in\bA$ are decomposed as $f_i=\Delta(g_i)+\phi(f_i)$ with $g_i\in\bA$, then $f_1+f_2=\Delta(g_1+g_2)+\phi(f_1+f_2)$, where $\phi(f_1+f_2)=\phi(f_1)+\phi(f_2)$.\\
(iii)\label{prop:complete_reduction} \textit{Parameterized telescoping:} 
 Parameterized telescoping, including its special case of creative telescoping~\cite{Zeil1991} (see Section~\ref{SECT:ParaTele}) can be tackled in an elegant manner.
Given nonzero elements $ f_0, \ldots, f_m \in \bA$, the goal is to find $g \in \bA$ and constants $c_0, \ldots, c_m \in C$ (not all zero) such that
\begin{equation}\label{Equ:ParaTele}
\Delta(g)= c_0 f_0 + \cdots + c_m f_m.
\end{equation}
Let $(v_i, \phi(f_i))$ be $\Sigma$-pairs of $f_i$ with respect to $\phi$ for $i=0,1,\ldots,m$. Then 
$0=\phi(c_0\,f_0+\dots+c_m\,f_m-\Delta(g))=c_0\,\phi(f_0)+\dots+c_m\,\phi(f_m)$
by Property (ii).
So the $c_i$ can be obtained by solving a linear system over $C$ (provided a suitable $C$-basis of the image space under $\phi$ is available),
after which $g$ can be taken as $c_0\,v_0+\dots+c_m\,v_m$.\\
(iv) \textit{Relation to additive decompositions:} In general, a complete reduction for $(\bA,\rho)$ decomposes an element $f \in \bA$ as $\rho(g) + r$, where $r \in \bA$ is minimal with respect to a partial ordering that has been introduced in  \cite[Section 1]{GaoPhD2024} in the vector space setting. 
Consequently, $\phi$ in the special case $\rho = \Delta$ constitutes an additive decomposition with respect to this minimality notion.  

\medskip

We remark that additive decomposition approaches that do not implement a complete reduction (such as~\cite{Abra1975,Paul1995,AbPe2002b,CHKL2015,DHL2018,CDGHL2025,Schn2007}) satisfy property~(i) but fail to satisfy property~(ii). Consequently, the consise approach to solving the parameterized telescoping problem as described in~(iii) cannot be applied directly, and further non-trivial calculations are required to pursue this tactic.

\subsection*{Contribution} We are interested in  constructing complete reductions for $R\Pi\Sigma^*$-towers~(see Definition \ref{DEF:RPiSigma}). More precisely, we start with a difference ring $(\bA, \, \sigma)$ with constant field $C$. Then we extend it to $(\bE,\, \sigma)$ by a tower of ring extensions  of the following form:
\begin{itemize}
    \item algebraic ring extensions which introduce objects of the type $\alpha^n$ with $\alpha$ being a root of unity;
    \item product extensions which introduce nested products in terms of Laurent polynomials where the multiplicands are invertible;
    \item sum extensions which introduce nested sums in terms of polynomials. 
\end{itemize}
Finally, we require that $\const(\bE,\,\sigma)=C$.

This paper derives two main results. The first is summarized in Theorem~\ref{THM:SimpleRPiSigmaExt}. Here we
develop an algorithm for constructing a complete reduction for $(\bE,\, \Delta)$, where $(\bE,\, \sigma)$ is a simple $R\Pi\Sigma^*$-tower~(see Definition \ref{DEF:SimpleRPiSigma})  over a ground ring $(\bA, \sigma)$ that satisfies certain properties, including the ability to compute complete reductions in $(\bA,\,\sigma)$. 
The crucial step is to construct a complete reduction for each of the three types of extensions within a slightly generalized framework and subsequently combine them inductively. Specifically, we utilize the framework for sum extensions from~\cite{CGHS2025} and derive new constructions for algebraic and product extensions. Finally, by specializing to the rational difference field $(\bA, \sigma) = (C(x),\, \sigma)$ with $\sigma(x) = x + 1$ over a factorizable constant field $C$ (see Definition~\ref{DEF:computable}), we employ techniques from~\cite{CHKL2015, CDGHL2025} and \cite[Theorem~5.78]{Kaue2023} to obtain an algorithm that computes complete reductions for this ground ring. As a consequence, we obtain a fully algorithmic procedure for constructing complete reductions for simple $R\Pi\Sigma^*$-towers over $(C(x),\,\sigma)$.

The second main result is Theorem~\ref{THM:CR+RPiSigma}, in which we construct a complete reduction for $(\bE,\, \Delta)$, where $(\bE,\, \sigma)$ is an $R\Pi\Sigma^*$-tower over an integral and constant-stable difference ring $(\bA,\, \sigma)$ (see Definition~\ref{DEF:stable}), where $(\bA,\, \sigma)$ is subject to the same assumptions as in Theorem~\ref{THM:SimpleRPiSigmaExt}.
Since the generators in the general case lack a regular pattern, constructing a complete reduction directly is challenging. Inspired by \cite{vdSi1997,HS2008,Schn2017,AS:2021}, we introduce the concept of idempotent representations for difference rings (see Definition~\ref{DEF:idempotent}) and generalize the construction of idempotent representations for basic $R\Pi\Sigma^*$-extensions --a special class of simple $R\Pi\Sigma^*$-extensions --from \cite[Section 4]{Schn2017} to general $R\Pi\Sigma^*$-towers. We prove that each component in the resulting decomposition is a $\Pi\Sigma^*$-tower, which is a special type of simple $R\Pi\Sigma^*$-tower. Furthermore, we reduce the problem of constructing complete reductions for $(\bE,\, \Delta)$ to constructing a complete reduction on a single component, enabling us to complete the construction using our first main result. As a byproduct, we obtain a fully algorithmic procedure for $R\Pi\Sigma^*$-towers over $(C(x),\,\sigma)$. 
 Remarkably, this provides the first algorithm to solve the telescoping problem for the general class of $R\Pi\Sigma^*$-towers.  In particular, the algorithm does not require solving any difference equations.
Specifically, one is now in the position to directly model products such as $\lfloor \frac {n}{2}\rfloor!$, which can not be treated in a simple $R\Pi\Sigma^*$-tower. For instance, our algorithm finds
\[\sum_{j=0}^n \frac{((-1)^j(2-j)+j)j!}{2^j \lfloor \frac j2\rfloor!}=\frac{2(n+1)!}{2^n \lfloor \frac n2\rfloor!},\]
compare Example~\ref{EX:newidentity} below. 

Finally, we obtain a new efficient framework for the creative telescoping problem in this general setting based on property~(iii) stated above. As a case study, we compute recurrences for the following family of sums 
\begin{equation}\label{EQ:recurrence}
S_{n, \ell}= \sum_{j=0}^{n}\left(1-\ell jH_j+\ell (n-j)H_j\right)\binom{n}{j}^{\ell}, \quad \ell\in \bZ^{+}
\end{equation}
in terms of the harmonic numbers $H_j=\sum_{i=1}^j\frac1i$ considered in~\cite{PS2003,KrattenthalerRivoal2007}. These sums are connected to irrationality proofs of zeta-values and and are related to supercongruences using the $p$-adic gamma function,  see \cite{Mort2003a,Mort2003b,AO2000}.
We benchmark our method against the corresponding functions in \texttt{Sigma}; the results are summarized in Tables \ref{tab:CT1}--\ref{tab:CT3}.

All algorithms presented in this paper are implemented in Mathematica. The code is available at
\begin{center}
\url{https://risc.jku.at/sw/corrps/}.
\end{center}

\subsection*{Structure of the paper}

In Section~\ref{SECT:pre}, we present basic notions and properties of $R\Pi\Sigma^*$-towers. In Section~\ref{SECT:simplecase}, we algorithmically construct complete reductions for simple $R\Pi\Sigma^*$-towers and present computational benchmarks. Utilizing idempotent representations, we obtain telescoping algorithms for the general class of $R\Pi\Sigma^*$-towers over $\Pi\Sigma^*$-fields in Section~\ref{SECT:IdemotentRepresentation}. Furthermore, we derive a general mechanism for constructing complete reductions for $R\Pi\Sigma^*$-towers over the rational difference field. Finally, Section~\ref{SECT:ParaTele} illustrates how to solve the creative telescoping problem in $R\Pi\Sigma^*$-towers via complete reductions, followed by concluding remarks in Section~\ref{SECT:Conclusion}.

\section{Preliminaries}\label{SECT:pre}

Basic terminology is reviewed in Section~\ref{SUBSECT:term}, and $R\Pi\Sigma^*$-extensions are introduced in Section~\ref{Sec:RPS-extensions}, summarizing the properties from~\cite{Schn2016a} required to construct complete reductions. Here, the condition that the ground ring forms a field is partially relaxed to an integral domain, thereby highlighting and extending relevant properties.

\subsection{Basic terminologies} \label{SUBSECT:term}

Let $\bN$, $\bZ^{+}$ and $\bZ^{-}$ be the sets of nonnegative, positive and negative integers, respectively. Set $[n]:=\{1, \ldots, n\}$ for $n \in \bZ^{+}$, and $[n]_0:=\{0,1,\ldots,n\}$ for $n\in \bN$. For $a\in \bZ$ and $\lambda \in \bZ^{+}$, 
$a \bmod \lambda$ is the unique integer $r$ such that $a = q\lambda + r$ for some integer $q$, where $0 \le r < \lambda$.

For a ring $\bA$, $\bA^*$ stands for the multiplicative group of units in $\bA$. For a root of unity $\alpha\in \bA$, we define the \emph{order of $\alpha$} to be $\min\{n \in\bZ^+\mid \alpha^n=1\}$, denoted by $\ord(\alpha)$. Let $t$ be an indeterminate over $\bA$ and consider the polynomial ring $\bA[t]$. For  $p \in \bA[t]$, denote its leading coefficient and degree by $\lc_t(p)$ and $\deg_t(p)$, respectively. In addition,  $\lc_t(0):=0$ and $\deg_t(0):=-\infty$.
The ring of Laurent polynomials in $t$ over a ring $\bA$ is denoted by $\bA[t,t^{-1}]$. Let $q =q_k t^k + q_{k-1} t^{k-1} + \cdots + q_{\ell+1} t^{\ell+1} + q_{\ell} t^{\ell} \in \bA[t, t^{-1}]$,
where $k, \ell \in \bZ$, $k \ge \ell$, $q_k, q_{k-1}, \ldots,  q_{\ell+1}, q_{\ell} \in \bA$, and $q_k, q_{\ell} \neq 0$.
We define $k$ and $\ell$  to be {\em head degree} and {\em tail degree} of $q$, which are
denoted by  $\hdeg_t(q)$ and $\tdeg_t(q)$, respectively. 
In addition, $\hdeg_t(0):= -\infty$ and $\tdeg_t(0) := \infty$. 
For a ring $\bE=\bA[z_1,z_2,\ldots,z_n]$ ( where $z_i$ is
transcendental or algebraic over $\bA[z_1,\ldots,z_{i-1}]$
for $i\in[n]$), $f\in\bE$, and $a_1,\ldots,a_n\in\bA$, whenever
the substitution $z_i\mapsto a_i$ is well-defined on $\bE$, we write
$
f(a_1,\ldots,a_n)
=f|_{\{z_i\to a_i\mid i\in[n]\}}.$

 Let $(\bA, \, \sigma)$ be a difference ring. For any element $a$  in $(\bA, \, \sigma)$ and any $i \in \bN$, we define
\[a^{\sigma,i}:=\prod_{j=0}^{i-1} \sigma^j(a)\quad\text{and}\quad a^{\sigma,\sigma,i}:=\prod_{j=0}^{i-1} a^{\sigma,j}, \]
where $a^{\sigma, 0}=a^{\sigma, \sigma, 0}=1$. 
\begin{remark}\label{RE:sigma+notation}
$a\sigma(a^{\sigma,i})=a^{\sigma,i+1}$ and $a^i\sigma(a^{\sigma,\sigma,i})=a^{\sigma,\sigma,i+1}$.
\end{remark}

A difference ring $(\tilde{\bA},\,\tilde \sigma)$ is called a {\em difference ring extension} of $(\bA,\,\sigma)$ if $\bA$ is a subring of $\tilde{\bA}$ and $\tilde{\sigma}|_{\bA}=\sigma$. For simplicity, we denote the automorphism $\tilde{\sigma}$ by $\sigma$ whenever the context is clear. Let $(\bA_1,\, \sigma_1)$ and $(\bA_2, \, \sigma_2)$ be two difference rings. A ring isomorphism $\tau: \bA_1 \to \bA_2$ is called a {\em difference ring isomorphism} if $\tau \sigma_1= \sigma_2 \tau$. For a comprehensive introduction to these concepts, we refer to \cite{Cohn1965, Levi2008}.
 
Throughout this article, when algorithms come into play, we assume  that $\bA$ and its constant field $C$ satisfy one of the following properties, compare \cite{Karr1981, KS:06}.

\begin{define}\label{DEF:computable}
A ring/field $\bA$ is called {\em computable} if all ring/field operations are computable and the zero-recognition problem is solvable. A field $\bF$ is said to be {\em factorizable} if it is computable and one can additionally factorize multivariate polynomials over $\bF$. Moreover, a field $\bF$ is said to be {\em orbit-computable} if it is factorizable, and for any $c_1,\dots,c_n \in \bF$, one can compute  a $\bZ$-basis of
$$M(c_1,\dots,c_n; \bF)= \{(e_1,\dots,e_n) \in \mathbb{Z}^n\mid 1=\prod_{i=1}^n c_i^{e_i}\}.$$
 A difference ring/field $(\bA,\,\sigma)$ is said to be \emph{computable} if both $\bA$ and the function $\sigma:\bA\to\bA$ are computable. 
\end{define}

\begin{remark}\label{Remark:AlgRatOribt}
Any  field of multivariate rational functions over an algebraic number field is orbit-computable~\cite{Ge:93a,Schneider:2005}.
\end{remark}

\begin{convention}
Throughout the paper,
\begin{itemize}
\item [(i)] all rings are commutative with 1 and have characteristic $0$,
 \item[(ii)] whenever $U$ as given in Definition \ref{DEF:cr} is a difference ring with constant 
 field $C$,  $U$ is regarded as a $C$-linear space and every operator on $U$ is understood to be $C$-linear, unless otherwise specified.
 \end{itemize}
 \end{convention}

 \subsection{$R\Pi\Sigma^*$-towers and basic properties}\label{Sec:RPS-extensions}

 Let $(\bA,\,\sigma)$ be a difference ring. We consider three types of difference ring extensions. First, given $a\in \bA$, let $\bA[t]$ be a polynomial ring over $\bA$. The automorphism $\sigma$ of $\bA$ extends uniquely to a ring
automorphism of $\bA[t]$ such that $\sigma(t)=t+a$. The resulting difference ring extension $(\bA[t],\,\sigma)$ of $(\bA,\,\sigma)$ is called a \emph{sum extension} (in short \emph{$S$-extension}), and $t$ is called an \emph{$S$-monomial (over $\bA$)}. Similarly, given $b\in\bA^*$, let $\bA[t,t^{-1}]$ be a Laurent polynomial ring over $\bA$. The automorphism $\sigma$ extends uniquely to a ring automorphism of $\bA[t,t^{-1}]$ such that $\sigma(t)=bt$. In particular, $\sigma(t^{-1})=b^{-1}t^{-1}$. The resulting difference ring extension $(\bA[t,t^{-1}],\,\sigma)$ of $(\bA,\,\sigma)$ is called a \emph{product extension} (in short \emph{$P$-extension}), and $t$ is called a \emph{$P$-monomial (over $\bA$)}. Finally, we introduce algebraic extensions in which one can model $\alpha^n$ for a root of unity $\alpha$. To be more precise, let $\lambda \in \bN$ with $\lambda>1$ and $\alpha \in \bA$ be a $\lambda$th root of unity, i.e., $\alpha^{\lambda}=1$. According to \cite[Lemma~2.6]{Schn2016}, there is (up to a difference ring isomorphism) a unique difference ring extension $(\bA[t],\, \sigma)$ of $(\bA, \, \sigma)$,  where $\bA[t]$ is generated by an element $t\notin \bA$ subject to the relation $t^{\lambda}=1$ and $\sigma(t)=\alpha t$. Such an extension is called an \emph{algebraic extension} (in short \emph{$A$-extension}) \emph{of order $\lambda$}, and the generator $t$ is also called an \emph{$A$-monomial (over $\bA$)}. In this case, 
$\ord(t)=\lambda$. By construction, $\bA[t]$ is isomorphic to the quotient ring $\bA[x]/\langle x^\lambda-1\rangle$, where $x$ is an indeterminate over $\bA$ and $t$ corresponds to the equivalence class of $x$. Since $x^\lambda-1$ is monic of degree $\lambda$, every element of $\bA[t]$ can be uniquely written as $\sum_{i=0}^{\lambda-1}f_it^i$, where $f_i \in \bA$; for further details, we also refer to the recent survey article~\cite{Schneider:26}.

For convenience, we introduce the following notion. Let $(\bA,\, \sigma)$ be a difference and let $t$ be one of the generators introduced above. Then $\bA\langle t \rangle$ denotes the polynomial ring $\bA[t]$ if $t$ is an $S$-monomial, $\bA\langle t \rangle$ stands for the ring of Laurent polynomials $\bA[t,t^{-1}]$ if $t$ is a $P$-monomial; and $\bA\langle t \rangle$ denotes the ring $\bA[t]$ with $t\notin \bA$ subject to the relation $t^{\lambda}=1$ if $t$ is an $A$-monomial with  $\ord(t)=\lambda$. We call a difference ring extension $(\bA \langle t \rangle, \sigma)$ of $(\bA, \, \sigma)$  an {\em $AP$-extension} if it is an $A$ or $P$-extension; a {\em $PS$-extension} if it is a $P$ or an $S$-extension; an {\em $APS$-extension} if it is an $A$-, $P$- or $S$-extension.

Note that $(\bA \langle t \rangle, \sigma)$ is computable if $(\bA, \, \sigma)$ is computable. Later, we will use the following lemma that can be verified by a straightforward calculation.

\begin{lemma}\label{LM:RPiSigma+term}
Let $(\bA, \, \sigma)$ be a difference ring and $\ell \in \bZ^{+}$.
\begin{itemize}
  \item [(i)] Assume that $t$ is an $S$-monomial over $\bA$ with $\sigma(t)=t+a$ for some $a \in \bA \backslash \{0\}$. Then $\sigma^{\ell}(t)=t+\sum_{i=0}^{\ell-1}\sigma^i(a)$.
  \item [(ii)] Assume that $t$ is an $AP$-monomial over $\bA$ with $\sigma(t)=bt$ for some $b \in \bA^*$. Then $\sigma^{\ell}(t)=b^{\sigma, \ell}t$.
\end{itemize}
\end{lemma}

In what follows, we often take the ground difference ring to be the rational difference field $(\bA,\,\sigma)=(C(x),\, \sigma)$. Here, $C(x)$ is the rational function field over $C$, equipped with the automorphism $\sigma:C(x)\to C(x)$ defined by $\sigma(x)=x+1$ and $\sigma(c)=c$ for all $c\in C$.  Clearly, $\const(C(x),\,\sigma)=C$.
More generally, we consider, e.g., the  broader class of Karr's $\Pi\Sigma^*$-fields~\cite{Karr1981}.

\begin{define}
	 A difference field $(\bF,\,\sigma)$ with constant field $C$ is called {\em a $\Pi\Sigma^*$-field } over $C$ if $(\bF,\sigma)$ is built by a tower of difference field extensions $\bF=C(t_1)\dots(t_e)$ with the same constant field
$C$, where, for every $i\in [e]$, we have 
	\begin{enumerate}
		\item [(i)] $\sigma(t_i)/t_i\in C(t_1)\dots(t_{i-1})^*$~($t_i$ is called a {\em $\Pi$-field monomial}), or
		\item [(ii)] $\sigma(t_i)-t_i\in C(t_1)\dots(t_{i-1})$~($t_i$ is called a {\em $\Sigma^*$-field monomial}).
	\end{enumerate}
\end{define}
Building upon  such a difference field or ring $(\bA,\,\sigma)$, we focus on  towers of $APS$-extensions and their refined versions, namely $R\Pi\Sigma^*$- extensions introduced in \cite[Definition 2.10]{Schn2016}. 

\begin{define}\label{DEF:RPiSigma}
Let $(\bA,\, \sigma)$ be a difference ring and $(\bE,\, \sigma)$ be a tower of difference ring extensions of $(\bA,\, \sigma)$,
\begin{equation} \label{EQ:tower}
	\begin{array}{cccccccc}
		\bA=\bA_0       & \leq &  \bA_1      & \leq & \cdots         & \leq     & \bA_n=\bE  \\
		&         & \shortparallel &         &            &             & \shortparallel \\
		&         &  \bA_0\langle t_1 \rangle &         &                   &  & \bA_{n-1}\langle t_n \rangle.
	\end{array}
\end{equation}
If $t_i$ is an $A$-/$P$-/$S$-/$AP$-/$PS$-/$APS$-monomial over $\bA_{i-1}$ for every $i\in [n]$, then $(\bE,\, \sigma)$ is called an {\em $A$-/$P$-/$S$-/$AP$-/$PS$-/$APS$-tower} over $(\bA,\,\sigma)$.
\noindent A difference ring extension $(\bA\langle t \rangle, \, \sigma)$ is called an {\em  $R$-/$\Pi$-/$\Sigma^*$-/$R\Pi$-/$\Pi\Sigma^*$-/$R\Pi\Sigma^*$- extension} of $(\bA, \, \sigma)$ if $t$ is an $A$-/$P$-/$S$-/$AP$-/$PS$-/$APS$-monomial and $\const(\bA\langle t \rangle, \sigma)=\const(\bA, \, \sigma)$. Moreover, the difference ring extension $\bE$ as given in \eqref{EQ:tower} is called an {\em  $R$-/$\Pi$-/$\Sigma^*$-/$R\Pi$-/$\Pi\Sigma^*$-/$R\Pi\Sigma^*$-tower} over $(\bA, \, \sigma)$~(or {\em $R$-/$\Pi$-/$\Sigma^*$-/$R\Pi$-/$\Pi\Sigma^*$-/$R\Pi\Sigma^*$-tower} over $(\bA, \, \sigma)$) if $\const(\bE,\, \sigma)=\const(\bA, \, \sigma)$. 
\end{define}

Next, we state the necessary and sufficient conditions for an $APS$-extension to be an $R\Pi\Sigma^*$-extension \cite[Theorem~2.2]{Schn2016}, referring to~\cite{Karr1981} for the field setup.

\begin{thm}\label{THM:testingRPiSigma-monomials}
Let $(\bA, \, \sigma)$ be a difference ring, where the set of constants forms a field. Then the following assertions hold.
\begin{itemize}
  \item [(i)] Let $(\bA[t],\, \sigma)$ be an $S$-extension of $(\bA, \, \sigma)$ with $\sigma(t)=t+a$ for $a\in \bA$. Then it is an $\Sigma^*$-extension if and only if there is
  no $g\in \bA$ such that $\sigma(g)=g+a$.
  \item [(ii)] Let $(\bA[t, t^{-1}],\, \sigma)$ be a $P$-extension of $(\bA, \, \sigma)$ with $\sigma(t)=b\,t$ for $b\in \bA^*$. Then it is an $\Pi$-extension if and only if there are no $g\in \bA\backslash \{0\}$ and $m \in \bZ\backslash \{0\}$ such that $\sigma(g)=b^mg$.
  \item [(iii)] Let $(\bA[t], \,\sigma)$ be an $A$-extension of $(\bA, \, \sigma)$ of order $\lambda>1$ with $\sigma(t)=\alpha t$ for $\alpha \in \bA^*$. Then it is an $R$-extension if and only if there are no $g\in \bA\backslash \{0\}$ and $m \in [\lambda-1]$ such that $\sigma(g)=\alpha^mg$. If it is an $R$-extension, then $\alpha$ is primitive, i.e., $\ord(\alpha)=\lambda$.
\end{itemize} 
\end{thm}
For algorithmic versions of these criteria, see \cite[\S 2.2]{Schn2016}.
   
   \begin{example}
 Let $\bF=\bQ(x)$ be the rational difference field with  $\sigma(x)=x+1$. 
 \begin{itemize}
   \item [(i)] Take the unimonomial ring extension $(\bF[t],\sigma)$ of $(\bF,\sigma)$ with $\sigma(t)=t+\frac{1}{x+1}$. Then $t$, which models the harmonic numbers $H_n:=\sum_{i=1}^{n}\frac{1}{i}$, is a $\Sigma^*$-monomial over $\bF$. 
   \item [(ii)] Let $(\bF[t, t^{-1}],\sigma)$ be a difference ring extension of $\bF$ with $\sigma(t)=(x+1)t$. Then $t$, which models $n!$, is a $\Pi$-monomial over $\bF$. 
   \item [(iii)] Take the algebraic extension $(\bF[t], \sigma)$ with $\sigma(t)=-t$ of order 2. Then it is an $R$-extension subject to the relation $t^2=1$. And $t$ represents  $(-1)^n$.
 \end{itemize}
 \end{example}

Later we will utilize the property that any $R\Pi\Sigma^*$-tower over an integral domain can be reordered as follows.

\begin{prop}\label{PROP:orderRPiSigma}
Let $(\bE,\, \sigma)$ be an $R\Pi\Sigma^*$-tower over a difference ring $(\bA, \, \sigma)$ which is an integral domain. Then it can be reordered to the form 
\begin{equation}\label{EQ:reoder+RPiSigma}
\bA\langle y_1 \rangle \ldots \langle y_{\ell} \rangle \langle p_1 \rangle \ldots \langle p_{m} \rangle  \langle s_1 \rangle \ldots \langle s_{n} \rangle,
\end{equation}
where the $y_i$ are $R$-monomials, the $t_i$ are $\Pi$-monomials and the $s_i$ are $\Sigma^*$-monomials.
\end{prop}
\begin{proof}
The proof for the case where $\bA$ is a field is given in~\cite[Lemma~4.13]{Schn2016}. Since this argument does not rely on the invertibility of elements, it carries over straightforwardly to the setting where $\bA$ is an integral domain by replacing \cite[Corollary~4.3]{Schn2016} with the following claim.

\noindent{\em Claim.} Let $(\bA\langle y_1\rangle \ldots \langle y_\ell \rangle,\sigma)$ be an $R$-tower over a difference ring $(\bA,\,\sigma)$ which is integral. Then it contains no nonzero nilpotent elements.

\noindent{\em Proof of the claim.}  Set $\lambda_i=\ord(y_i)$ for $i\in[\ell]$. Let $H$ be the polynomial ring $\bA[ y_1] \ldots [y_\ell]$ and $I\subset H$ be the ideal generated by $y_1^{\lambda_1}-1,\ldots,y_\ell^{\lambda_\ell}-1$. Then $\bA\langle y_1\rangle \ldots \langle y_\ell \rangle$ is isomorphic to the quotient ring $H/I$. It remains to show that $I$ is radical. 
Assume that $h\in H$ satisfies $h^n\in I$.
Let $Q(\bA)$ denote the quotient field of $\bA$ and let $I_Q\subset Q(\bA)[ y_1] \ldots [y_\ell]$ be the ideal generated by  $y_1^{\lambda_1}-1,\ldots,y_\ell^{\lambda_\ell}-1$. It is radical by \cite[Lemma~4.2]{Schn2016}. Note that $h^n\in I \subset I_Q$. So $h\in I_Q$.  Hence, there are $g_1,\ldots,g_r\in Q(\bA)[ y_1] \ldots [y_\ell]$ such that $h=g_1 (y_1^{\lambda_1}-1)+\cdots +g_\ell (y_\ell^{\lambda_\ell}-1)$. Since the generators of $I_Q$ form a Gröbner basis of $I_Q$ with respect to the lexicographic order and are monic in $H$, we get $g_1,\ldots,g_\ell\in H$, which implies that $h\in I$. Thus, $I$ is radical. 
\end{proof}

In addition, we will also rely on the fact that the ground difference ring or field is constant-stable.
 
 \begin{define}\label{DEF:stable}
  A difference ring $(\bA, \, \sigma)$ is said to be {\em constant-stable} if $\const(\bA,\, \sigma^k)=\const(\bA, \, \sigma)$ for every $k \in \bZ^{+}$. 
 \end{define}
 \begin{remark}\label{RE:constant-stable}
If the difference ring $(\bA, \, \sigma)$ is constant-stable, then  $(\bA, \sigma^{\lambda})$ is constant-stable for every $\lambda \in \bZ^{+}$. 
\end{remark}
 
 For instance, the rational difference field is constant-stable.
 More generally, Karr's $\Pi\Sigma^*$-fields over a constant field $C$ are constant-stable; for a proof of the field version, see, e.g., \cite[Proposition~32]{Schn2020}.
 Below we will rely in particular on the following ring version.
 
\begin{lemma}\label{Lemma:ConstantStable}
Let $(\bE, \,\sigma)$ be a $\Pi\Sigma^*$-tower over a difference ring $(\bA, \,\sigma)$. If $(\bA,\, \sigma)$ is a constant-stable integral domain, then so is $(\bE,\,\sigma)$.
\end{lemma}
\begin{proof} 
	We prove the lemma by induction on the number of $\Pi\Sigma^*$-monomials. If $\bE=\bA$, the statement clearly holds. Otherwise, suppose that the statement has been proven for a $\Pi\Sigma^*$-tower $(\bE,\,\sigma)$ over $(\bA,\,\sigma)$ and let $t$ be  $\Pi\Sigma^*$-monomial over $\bE$.	
	Clearly, $\bE\langle t\rangle$ is an integral domain. Hence
	it suffices to show that $\const(\bE\langle t\rangle, \sigma^k)=\const(\bE,\,\sigma)$ for all $k>1$. The conclusion that $\const(\bE,\,\sigma) \subset \const(\bE\langle t\rangle, \sigma^k)$ is immediate. Conversely, let $f \in \bE \langle t \rangle$ satisfy $\sigma^k(f)=f$ for some $k>1$. Suppose that $f \notin \bE$. Define $h:=f^{\sigma,k}=\prod_{i=0}^{k-1} \sigma^i(f)$. 
	Since $\deg_t(\sigma^i(f))=\deg_t(f)$ for every $i \in \bN$ by Lemma \ref{LM:RPiSigma+term}, it follows $h \notin\bE$. On the other hand,
	$f\sigma(h)=f^{\sigma, k+1}=h\sigma^{k}(f)=hf.$
	Since $\bE$ is an integral domain, we have $\sigma(h)=h$. Together with $\const(\bE\langle t \rangle, \sigma)=\const(\bE,\,\sigma)$, $h \in \const(\bE,\,\sigma)$. So $h \in \bE$, a contradiction. Thus, $f \in \bE$. Furthermore, $f \in \const(\bE,\,\sigma)$ because $\bE$ is constant-stable. Consequently,
	$\const(\bE\langle t\rangle, \sigma^k)=\const(\bE,\,\sigma)$, which completes the induction step.
\end{proof}

Assuming the ground difference ring is integral and constant-stable, we provide a simpler criterion for an $A$-extension to be an $R$-extension (cf. Theorem~\ref{THM:testingRPiSigma-monomials} (iii)).
 
 \begin{prop}\label{PROP:R-monomial+primitive}
 Let $(\bA, \, \sigma)$ be a difference ring and $t$ be an $A$-monomial over $\bA$ of order $\lambda$. Assume that $\bA$ is a constant-stable integral domain. Then $\alpha:=\sigma(t)/t \in \const(\bA, \, \sigma)$. In particular, $t$ is an $R$-monomial if and only if $\alpha$ is a $\lambda$th primitive root of unity in $\const(\bA, \, \sigma)$, i.e.,\ $\ord(\alpha)=\lambda$.
 \end{prop}
 \begin{proof} Since $\alpha$ is a root of unity in $\bA$, there is $n \in \bN$ such that $\alpha^n=1$. For every $i \in \bZ$, applying $\sigma^i$ to the above equation yields $(\sigma^i(\alpha))^n=1$. Let $x$ be an indeterminate over $\bA$. Then $\alpha$, $\sigma(\alpha)$, $\ldots$, $\sigma^i(\alpha)$, $\ldots$ are roots of $x^n-1$. Note that $\bA$ is an integral domain. So $x^n-1$ has at most $n$ roots in $\bA$. Hence, there are $i,j \in \bZ$ with $i >j$ such that $\sigma^i(\alpha)=\sigma^j(\alpha)$, which implies  $\alpha \in \const(\bA,\, \sigma^{i-j})$. Furthermore, $\alpha \in \const(\bA,\, \sigma$) by $\const(\bA, \sigma^k)=\const(\bA, \, \sigma)$ for $k\in \bZ^{+}$. Note that \cite[Proposition~2.20]{Schn2017} admits a mild generalization from fields to integral domains. So $t$ is an $R$-monomial if and only if $\alpha$ is a $\lambda$th primitive root of unity in $\const(\bA, \, \sigma)$.
 \end{proof}

\section{Complete reductions for simple $R\Pi\Sigma^*$-towers}\label{SECT:simplecase}

In this section, we restrict to the following subclass of $R\Pi\Sigma^*$-towers introduced by \cite[Definition 2.19]{Schn2016}.

\begin{define}\label{DEF:SimpleRPiSigma}
Let $(\bA,\, \sigma)$ be a difference ring with constant field $C$ and $\gG$ be a multiplicative subgroup of $\bA^*$.  Let $(\bE,\, \sigma)$ be an $R\Pi\Sigma^*$-tower over $(\bA,\,\sigma)$ with $\bE=\bA\langle t_1 \rangle \ldots \langle t_n\rangle$, and define the multiplicative group
$${\gG}_{\bA}^{\bE}:=\{ft_1^{k_1}\ldots t_n^{k_n}~|~f\in \gG \text{ and }k_i\in \bZ~\text{with}~k_i=0~\text{if $t_i$ is a $\Sigma^*$-monomial}\}.$$
We say that $\bE$ is {\em $\gG$-simple} if for any $R\Pi$-monomial $t_i$, we have $\sigma(t_i)/t_i \in {\gG}_{\bA}^{\bE}$.
Any such extension is said to be {\em simple} if it is $\bA^*$-simple.
\end{define}

\begin{remark}\label{RE:SimplePiSigma=PiSigma}
 A $\Pi\Sigma^*$-tower over a difference integral domain $(\bA, \, \sigma)$ is simple by \cite[Proposition~4.14]{Schn2016}. Thus, it is precisely the presence of $R$-monomials, or more generally the situation where $\bA$ is not a field (or not an integral domain), that gives rise to broader classes of $R\Pi\Sigma^*$-towers which are not simple,
  for a concrete example, we refer to Example~\ref{EX:CheckRmonomials} below. Complete reductions for this general case will be treated in Section~\ref{SECT:IdemotentRepresentation} below.
\end{remark}

We now construct a complete reduction in a simple $R\Pi\Sigma^*$-tower. For the remainder of this section, let $(\bA,\, \sigma)$ be a difference ring with constant field $C$ and $t$ be a simple $R\Pi\Sigma^*$-monomial over $\bA$. A natural first attempt is to construct a complete reduction for $(\bA \langle t \rangle, \Delta)$ given one for $(\bA, \Delta)$. However, this approach falls short when handling $R\Pi$-monomials.

\begin{example}\label{EX:assumptionPi}
Let $(\bA[t,t^{-1}],\,\sigma)$ be a $\Pi$-extension of a difference ring $(\bA,\,\sigma)$ with $\sigma(t)=ft$ for some $f\in \bA^*$. We aim to determine whether $ut$ with $u\in \bA$ is summable in $\bA[t,t^{-1}]$. 
Let $q=\sum_nq_nt^n\in \bA[t,t^{-1}]$ be such that $ut=\Delta(q)$. A direct computation shows that $\Delta(q_nt^n)=(f^n\sigma(q_n)-q_n)t^n$
for every $n\in \bZ$. In particular, $f^n\sigma(q_n)-q_n=0$, or equivalently $\sigma(q_n)=f^{-n}q_n$ for all $n\neq1$. Since $t$ is a $\Pi$-monomial, {we have $q_n=0$ by Theorem~\ref{THM:testingRPiSigma-monomials}~(ii)}.  Thus, if there is a solution $ut=\Delta(q)$ with $q\in\bA[t,t^{-1}]$, then $q=yt$ for some $y\in\bA$ satisfying
\begin{equation}\label{Equ:FirstOrderEq}
u=f\sigma(y)-y.
\end{equation} 
Conversely, if $y\in\bA$ is a solution of~\eqref{Equ:FirstOrderEq},
then $\Delta(y\,t)=v\,t$. Note that~\eqref{Equ:FirstOrderEq} is a typical first-order difference equation that also appears in Karr's algorithm~\cite{Karr1981} or is related to the so-called Gosper equation~\cite{Gosp1978}. 
\end{example}

The above example motivates us to introduce a more general operator. Let $f\in \bA$. We define the $C$-linear operator $\Delta_f$ associated to $f$ by
\begin{equation} \label{EQ:karrequation}
 \begin{array}{cccc}
\Delta_f: & \bA  & \longrightarrow & \bA  \\
      &  y   & \mapsto         & f\sigma(y)-y. 
\end{array} 
\end{equation}
Then $f\sigma(y)-y=g$ has a solution in $\bA$ if and only if $g \in \im(\Delta_f)$. Note that $\Delta_1=\Delta$ introduced in~\eqref{EQ:differenceoperator}.

Dealing with $R$-monomials, it turns out that the so far introduced operator is still not general enough.

\begin{example}
Let $\bA[t]$ be an $R$-extension of a difference ring $(\bA,\,\sigma)$ of order $\lambda$ with $\sigma(t)=at$. For $f=st\in \bA[t]$ and $g\in \bA$, we determine whether there is a $q\in \bA[t]$ such that $g=f\sigma(q)-q$. If such a $q$ exists, write it as $\sum_{i=0}^{\lambda-1}q_it^i$ with $q_i \in \bA$. Then 
$$g=(sa^{\lambda-1}\sigma(q_{\lambda-1})-q_0)+
\sum_{j=1}^{\lambda-1}(s\sigma(q_{\lambda-j-1})a^{\lambda-j-1}-q_{\lambda-j})t^{\lambda-j}.$$
Since $g \in \bA$, we have  
$q_{\lambda-j}=s\sigma(q_{\lambda-j-1})a^{\lambda-j-1}$ for every $j \in [\lambda-1]$. By iterative substitution, it follows that 
$g=s^{\sigma, \lambda}a^{\sigma,\sigma, \lambda}\sigma^{\lambda}(q_0)-q_0.$
So $g \in \im(\Delta_f)$ if and only if $p\sigma^{\lambda}(y)-y=g$ with $p=s^{\sigma, \lambda}a^{\sigma,\sigma, \lambda}$ has a solution in $\bA$.
\end{example}

To this end, we generalize \eqref{EQ:karrequation} as follows.

\begin{define} \label{DEF:Generalizedkarrequation}
    Let $(\bA, \, \sigma)$ be a difference ring with constant field $C$. For $f \in \bA$ and $i \in \bZ^{+}$, define
    \begin{equation} \label{EQ:Generalizedkarrequation}
 \begin{array}{cccc}
\Delta_f^{(i)}: & \bA & \longrightarrow & \bA  \\
      &  y   & \mapsto         & f\sigma^i(y)-y. 
\end{array} 
\end{equation}
\end{define}
For simplicity, we henceforth write $\Delta_f$ for $\Delta_f^{(1)}$.
Note that one can solve certain classes of difference equations related to \eqref{EQ:Generalizedkarrequation} by Karr's algorithm~\cite{Karr1981} and the algorithms introduced in~\cite[Theorem~2.26]{Schn2016}, which can be summarized as follows.

\begin{thm} \cite{Karr1981, Schn2016}\label{Thm:PiSiOverSimpleRPiSiDESolver}
Let $(\bE,\,\sigma)$ be a simple $R\Pi\Sigma^*$-tower over a $\Pi\Sigma^*$-field $(\bF,\,\sigma)$ with orbit-computable\footnote{If the $\Pi\Sigma^*$-field $(\bF,\sigma)$ over $C$ is built only by $\Sigma^*$-field monomials, it suffices to require that $C$ is factorizable.} constant field $C$. Then, for $f\in {(\bF^*)}_{\bF}^{\bE}$, $i\in\bZ^+$ and $g\in\bE$, one can decide algorithmically if there is $q\in\bE$  such that $\Delta^{(i)}_f(q)=g$.
\end{thm}

Prior to presenting the results of this section, we first establish the following lemma, which describes the relationship between complete reductions in two difference-isomorphic rings. It will be used in Sections \ref{SUBSECT:RationalCase} and~\ref{SUBSECT:crforgeneralcase}.

\begin{lemma}\label{LM:diffiso+CR}
 Let $(\bA, \, \sigma)$ and $(\tilde{\bA},\, \tilde{\sigma})$ be difference rings. Assume that the set of constants of $\bA$ and that of $\tilde{\bA}$ are fields, and there is a difference ring isomorphism $\tau$ from $(\tilde{\bA},\, \tilde{\sigma})$ onto $(\bA, \, \sigma)$. Let $h \in\tilde{\bA}^*$.  Assume further that $\phi_{\tau(h)}$ is a complete reduction  for $(\bA, \,\Delta_{\tau(h)})$. Then the following holds.
\begin{itemize}
\item[(i)] For $\tilde{\Delta}_h:=h\tilde{\sigma}-\bf{1}$, we have $\tilde{\Delta}_h=\tau^{-1} \circ \Delta_{\tau(h)} \circ \tau$.
  \item [(ii)]~$\psi_h:=\tau^{-1} \circ \phi_{\tau(h)} \circ \tau$ is a complete reduction for $(\tilde{\bA},\, \tilde{\Delta}_h)$, and for every $g\in \tilde{\bA}$, $(\tau^{-1}(u), \tau^{-1}(v))$ is a $\Sigma$-pair of $g$ with respect to $\psi_h$ if  $(u, v)$ is a $\Sigma$-pair of $\tau(g)$ with respect to $\phi_{\tau(h)}$.
  \end{itemize}
\end{lemma}
\begin{proof}~(i)  For all $s \in \tilde{\bA}$, $(\tau^{-1} \circ \Delta_{\tau(h)} \circ \tau)(s)=h \tilde{\sigma}(s)-s$ by $\sigma\circ\tau=\tau\circ\tilde{\sigma}$. So $\tilde{\Delta}_h=\tau^{-1} \circ \Delta_{\tau(h)} \circ \tau$.

(ii) First we show that $\psi_h$ is a complete reduction for $(\tilde{\bA},\, \tilde{\Delta}_h)$.
Since $\phi_{\tau(h)}$ is idempotent, so is $\psi_h$. 
By Remark \ref{RE:cr}, it suffices  to verify that $\ker(\psi_h)=\im(\tilde{\Delta}_h)$.  
Since $\psi_h = \tau^{-1} \circ \phi_{\tau(h)} \circ \tau$, we have $\ker(\psi_h) = \tau^{-1}(\ker(\phi_{\tau(h)})).$
By (i), 
$\im(\tilde{\Delta}_h) =\tau^{-1} \im(\Delta_{\tau(h)}) $.
Thus, $\ker(\psi_h)=\im(\tilde{\Delta}_h)$
by $\ker(\phi_{\tau(h)})=\im(\Delta_{\tau(h)})$.

 Note that $\tau(g)=\Delta_{\tau(h)}(u)+v$. Applying $\tau^{-1}$ to the above equation, we obtain $g=\tilde{\Delta}_h (\tau^{-1}(u))+\tau^{-1}(v)$ by (i).
Moreover, $\tau^{-1}(v) \in \im(\psi_h)$ by $v \in \im(\phi_{\tau(h)})$. So  $(\tau^{-1}(u), \tau^{-1}(v))$ is a $\Sigma$-pair of $g$ with respect to $\psi_h$. 
\end{proof}

We aim to develop tools that enhance Theorem~\ref{Thm:PiSiOverSimpleRPiSiDESolver} toward a constructive version of complete reductions (see Theorem~\ref{THM:SimpleRPiSigmaExt} below).
More precisely, we begin with a simple $R\Pi\Sigma^*$-tower $(\bE,\, \sigma)$ over a difference ring $(\bA,\,\sigma)$ with constant field $C$.
By \cite[Lemma 4.10]{Schn2016}, we can reorder it to the form
 \begin{equation}\label{Equ:ReorderedSimpleExt}
\bE=\bA \langle y_1\rangle \ldots \langle y_{\ell_1} \rangle \langle p_1\rangle \ldots \langle p_{\ell_2}\rangle  \langle s_1\rangle \ldots \langle s_{\ell_3} \rangle,
\end{equation}
where the $y_i$ are $R$-monomials, the $p_i$ are $\Pi$-monomials, and the $s_i$ are $\Sigma^*$-mono\-mials\footnote{Compared with Proposition~\ref{PROP:orderRPiSigma},  the condition on the ground ring is weaker; it is required merely to be a general ring, whereas the generators are simple.}.

Given such a tower, we develop a general strategy for computing a complete reduction in $(\bE, \Delta)$, provided it can be computed in $(\bA, \Delta_f^{(i)})$ for every $f \in \bA^*$ and $i \in \mathbb{Z}^+$. In addition, we require that $\bA$ admits an effective $C$-basis, as formalized later in Definition~\ref{Def:effectiveBasis}.

As shown in Section~\ref{SUBSECT:RationalCase}, this required reduction is explicitly constructed when $\bA$ is the rational difference field (see Proposition~\ref{PROP:Generalrational+highorder}). Building on this result, Section~\ref{SUBSECT:RCase} constructs a complete reduction for $(\bE_1, \Delta_f^{(i)})$, where $\bE_1 := \bA \langle y_1\rangle \dots \langle y_{\ell_1} \rangle$, $f \in (\bA^*)_{\bA}^{\bE_1}$, and $i \in \mathbb{Z}^+$ (see Proposition~\ref{Prop:NestRExt}). Section~\ref{SUBSECT:PiCase} then establishes a complete reduction for $(\bE_2, \Delta_f)$ with $\bE_2 := \bE_1 \langle p_1\rangle \dots \langle p_{\ell_2}\rangle$ and $f \in (\bA^*)_{\bA}^{\bE_2}$ (see Proposition~\ref{PROP:NestedRPiExt}). Finally, by combining these results with those from~\cite{CGHS2025}, Section~\ref{SUBSECT:RPiSigma-case} presents a complete reduction for $(\bE, \Delta)$ (see Theorem~\ref{THM:SimpleRPiSigmaExt}) and compares our algorithm against existing methods implemented in the summation package \texttt{Sigma}.

\subsection{The rational case}\label{SUBSECT:RationalCase}

For the base case of our construction we use the rational difference field $(C(x),\,\sigma)$ of a factorizable  field $C$ with characteristic 0. For every $f \in C(x)^*$ and $i\in \bZ^{+}$, we are going to construct a complete reduction for $(C(x), \Delta_f^{(i)})$ in this section.
To achieve this, we simplify the construction to the case $(C(x), \Delta_f)$ for arbitrary $f \in C(x)^*$ as follows.

\begin{lemma}\label{LM:GeneralRational}
 Let $(C(x),\,\sigma)$ be the rational difference field with $\sigma(x)=x+1$. Assume that there is an algorithm which,  for every $f\in C(x)^*$,  constructs a complete reduction $\phi_f$ for $(C(x), \Delta_f)$. Then there is an algorithm that, for every $f \in C(x)^*$ and $i \in \bZ^{+}$, constructs a complete reduction $\phi_f^{(i)}$ for $(C(x), \Delta_f^{(i)})$.
 \end{lemma}
\begin{proof} Let $i\in \bZ^{+}$ and $f\in C(x)^*$. Define
\[
\begin{array}{cccc}
\tau: & (C(x), \sigma^i) & \rightarrow & (C(x), \sigma) \\
      & w(x) & \mapsto     & w(ix).
\end{array}
\]
Since $\sigma\circ\tau=\tau\circ\sigma^i$, $(C(x),\, \sigma)$ and $(C(x), \sigma^{i})$ are isomorphic. By the assumption, there is an algorithm for constructing a complete reduction for $(C(x),\, \Delta_{\tau(f)})$. This yields an algorithm for constructing a complete reduction for $(C(x),\, \Delta_{f}^{(i)})$ by replacing $(\tilde{\bA},\,\tilde{\sigma})$ and  $(\bA, \, \sigma)$ with $(C(x),\, \sigma^i)$ and $(C(x),\, \sigma)$ in Lemma \ref{LM:diffiso+CR}, respectively. 
\end{proof}

An element $g$ in $C(x)$ is said to be {\em $\sigma$-reduced} if $\gcd(\num(g), \sigma^{i}(\den(g)))=1$ for all $i\in \bZ.$
Let $f \in C(x)$. We call $(\xi, \eta) \in C(x) \times C(x)^*$ a {\em normal form of $f$} 
if $f = \xi\sigma(\eta)/\eta$ and $\xi$ is $\sigma$-reduced. Such a pair is computed by Algorithm {\tt RNF} in \cite{AbPe2002b}.  Note that normal forms are not unique. For example, both $(\frac{1}{x-1},x)$ and $(\frac{1}{x},x(x-1))$ are normal forms of $\frac{x+1}{x(x-1)}$.

The following lemma reduces the problem to constructing a complete reduction for $(C(x),\, \Delta_\xi)$, where $\xi\in C(x)$ is $\sigma$-reduced.

\begin{lemma} \label{LM:ks}
Let $f\in C(x)$ with a normal form $(\xi, \eta)$. Then the following assertions hold.
\begin{itemize}
\item[(i)] $\Delta_f = \eta^{-1} \circ \Delta_\xi \circ \eta$, 
where $\eta$ is understood as a $C$-linear operator on $C(x)$  with $g \mapsto \eta g$ for all $g \in C(x)$.
\item[(ii)] If $\phi_\xi$ is a complete reduction for $(C(x), \Delta_\xi)$, then $\phi_f:=\eta^{-1} \circ \phi_\xi  \circ \eta$
is a complete reduction for $(C(x), \Delta_f)$. Moreover, let $g \in C(x)$ and $(u, r)$ be a $\Sigma$-pair of $\eta g$ with respect to $\phi_{\xi}$, then
$(\eta^{-1}u, \eta^{-1}r)$ is a $\Sigma$-pair of $g$ with respect to $\phi_f$.
\end{itemize}
\end{lemma}
\begin{proof}
(i) For every $g\in C(x)$, we have $(\eta^{-1} \circ \Delta_\xi \circ \eta)(g) = (\xi\sigma(\eta)/\eta)\sigma(g) - g$
by a straightforward calculation. So $(\eta^{-1} \circ \Delta_\xi \circ \eta)(g)=\Delta_f(g)$ due to $f = \xi\sigma(\eta)/\eta$.
  
(ii) Since $\phi^2_{\xi}=\phi_{\xi}$, we have $\phi^2_f=\phi_f$.  To show $\im(\Delta_f) = \ker(\phi_f)$,
we note that $\im(\Delta_f) = \eta^{-1} \im(\Delta_\xi)$ by (i).
Furthermore, $\ker(\phi_f) = \eta^{-1} \ker(\phi_\xi)$  by the definitions of $\phi_\xi$ and $\phi_f$. Hence, $\im(\Delta_f) = \ker(\phi_f)$ because $\im(\Delta_\xi)=\ker(\phi_\xi)$.  Consequently, $\phi_f$ is a complete reduction for $(C(x), \Delta_f)$. 

Let $g\in C(x)$. Assume that $(u, r)$ is a $\Sigma$-pair of $\eta g$ with respect to $\phi_{\xi}$. Then $\eta g=\Delta_{\xi}(u)+r$. Dividing both sides of this equation by $\eta$ and using (i), we obtain
$$g=\eta^{-1}\Delta_{\xi}(u)+\eta^{-1}r=(\eta^{-1}\circ\Delta_\xi \circ \eta\circ\eta^{-1})(u)+\eta^{-1}r=\Delta_f(\eta^{-1}u)+\eta^{-1}r.$$ 
Moreover, $\eta^{-1}r \in \im(\phi_f)$ by $\eta^{-1}\im(\phi_{\xi})=\im(\phi_f)$. Thus, $(\eta^{-1}u, \eta^{-1}r)$ is a $\Sigma$-pair of $g$ with respect to $\phi_f$.
\end{proof}

In the rest of this section, we fix $(\xi, \eta)$ to be a normal form of $f$ and provide a detailed construction of a complete reduction for $(C(x), \, \Delta_{\xi})$ in Proposition \ref{PROP:CR+Rational}. We emphasize that this construction rephrased in the setting of hypergeometric terms is strongly related to the additive decomposition given in~\cite{AP2001, CHKL2015, CDGHL2025} and the refined representation of complete reductions elaborated in~\cite[Theorem 5.78]{Kaue2023}.

According to \cite[Definition~11]{Karr1981} and \cite[Definition~1]{AbPe2002b}, two nonzero polynomials $p, q$ in $C[x]$ are said to be {\em $\sigma$-equivalent} if $p=c\sigma^{\ell}(q)$ for some $\ell \in \bZ$ and $c\in C$. In this case, we write $p\stackrel{\sigma}{\sim} q$. If $p, q \in C[x]\backslash C$ are $\sigma$-equivalent, then $\ell$ is unique because there is no nonzero integer $k$ such that $\sigma^k(q)=q$. Obviously, $\stackrel{\sigma}{\sim}$ forms an equivalence relation. 
Moreover, the $\sigma$-equivalence of $p$ and $q$ can be recognized by comparing the coefficients of $p$ and $q$. 
We say that $p, q\in C[x]$ are {\em $\sigma$-coprime} if $\gcd(p, \sigma^{\ell}(q))=1$ for all nonzero integer~$\ell$.

Let $p_1^{n_1}p_2^{n_2}\dots p_r^{n_r}$ be a complete factorization of a monic polynomial $p$ in $C[x]$, where the $p_i \in C[x]$ are distinct monic irreducible polynomials, and the $n_i$ are positive integers. By using the above equivalence relation, one may pick among $p_1,\dots,p_r$ a subset $S=\{p_1,\dots,p_s\}$ which are pairwise $\sigma$-coprime and where all other irreducible factors in $p$ are $\sigma$-equivalent to $S$. Then  $p$ can be rewritten in the form 
\begin{equation}\label{Equ:SigmaFacFull}
p=\prod_{j=k_1}^{\ell_1}\sigma^{j}(p_1)^{m_{1,j}}\dots\prod_{j=k_s}^{\ell_s}\sigma^{j}(p_s)^{m_{s,j}}
\end{equation}
with $m_{i,j} \in \bN$. Note that such a representation can be computed if $C$ is factorizable.

A nonzero polynomial $p$ in  $C[x]$ is said to be {\em strongly $\sigma$-coprime} with $\xi$ if 
$$\gcd(\num(\xi),{\sigma}^\ell(p)) = \gcd(\den(\xi),{\sigma}^{-\ell}(p)) = 1$$ for all $\ell\in \bN$.
Given an irreducible polynomial $q$, it is straightforward to compute $r\in\bZ$ such that the irreducible polynomial $\sigma^r(q)$ is strongly coprime with $\xi$; compare by \cite[Lemma 5.3]{Huan2016}. So we can assume that the $p_i$ in \eqref{Equ:SigmaFacFull} are strongly coprime with $\xi$ (by replacing $p_i$ with $\sigma^{r_i}(p_i)$ for some $r_i\in\bZ$, and $k_i$, $\ell_i$ by $k_i-r_i$, $\ell_i-r_i$, respectively). In this case, we call \eqref{Equ:SigmaFacFull} a {\em strong $\sigma$-factorization of $p$ with respect to $\xi$}.
 
 Let $M_x$ be the set consisting of all monic irreducible polynomials with positive degrees in $C[x]$. And fix a subset $S \subset M_x$ of representatives of the equivalence classes induced by $\stackrel{\sigma}{\sim}$. In addition, we require that these representatives are strongly $\sigma$-coprime with respect to $\xi$. Then any finite product of elements in $S$ is also strongly $\sigma$-coprime with $\xi$. 

Set 
\begin{equation}\label{EQ:rationalcomplement}
U_{S,\xi} := \left\{\frac{a}{q_1^{m_1}\cdots q_{\ell}^{m_\ell}}~|~q_i \in S, \, m_i\in\bN \, \text{and}\,\deg_x(a)<\deg_x(q_1^{m_1}\cdots q_{\ell}^{m_\ell}) \right\},
\end{equation}
which is  a $C$-linear subspace.

Now we define the $C$-linear map
\begin{equation} \label{EQ:PR}
 \begin{array}{cccc}
\cP_{\xi}: & C[x] & \longrightarrow & C[x] \\
      &  p   & \mapsto         &\num(\xi)\sigma(p)-\den(\xi)p. 
\end{array} 
\end{equation}
Let 
\begin{equation}\label{EQ:polynomialComplement}
V_{\xi}:=\spa_C\{x^d~|~d\in \bN~\text{and}~d \neq \deg_x(p)~\text{for all}~p\in \im(\cP_{\xi})\}.
\end{equation}
Then $C[x]=\im(\cP_{\xi}) \oplus V_{\xi}$  and $V_{\xi}$ is of finite dimension by \cite[Section 4.1]{CHKL2015}. 
Furthermore, a $C$-basis of $V_{\xi}$ can be constructed in \cite[\S 4.2]{CHKL2015}.

The following proposition collects relevant results in \cite{AP2001, CHKL2015, CDGHL2025} and \cite[Theorem~5.78]{Kaue2023}.

\begin{prop}\label{PROP:CR+Rational}
Let $\xi\in C(x)$ be $\sigma$-reduced and the elements of $S$ are strongly $\sigma$-coprime with $\xi$. Then
$$C(x)=\im(\Delta_{\xi}) \oplus U_{S,\xi} \oplus \frac{V_{\xi}}{\den(\xi)},$$
where $U_{S,\xi}, V_{\xi}$ are given by \eqref{EQ:rationalcomplement} and \eqref{EQ:polynomialComplement}, respectively. Furthermore, the projection from $C(x)$ to $U_{S,\xi} \oplus \frac{V_{\xi}}{\den(\xi)}$ is a complete reduction $\phi_{\xi}$ for $(C(x), \Delta_{\xi})$. If $C$ is factorizable, $\Sigma$-pairs with respect to $\phi_{\xi}$ can be computed explicitly.
\end{prop}
\begin{proof} First we prove $C(x)=\im(\Delta_{\xi}) + U_{S,\xi} + \frac{V_{\xi}}{\den(\xi)}$. Let $g\in C(x)$. 
Write it as $\tilde{g}+h$, where $\tilde{g} \in C[x]$ and $h \in C(x)$ with $\deg_x(\num(h))<\deg_x(\den(h))$.
 Let \eqref{Equ:SigmaFacFull} be a strong $\sigma$-factorization of $\den(g)$ with respect to $\xi$. Then by \cite[Lemma~3]{AP2001} or Algorithm {\tt ShellReduction} in \cite[\S 3]{CHKL2015}, one can compute $k_1,\ldots, k_s \in \bN$, $a, h\in C(x)$ with $\deg_x(h)<\deg_x(\prod_{i=1}^{s}p_i^{k_i})$, and $v\in C[x]$ such that
 $g=\Delta_{\xi}(a)+h/\prod_{i=1}^{s}p_i^{k_i}+v/\den(\xi)$; note that here one requires a strong $\sigma$-factorization, which one can compute if $C$ is factorizable. 
 Set $u=h/\prod_{i=1}^{s}p_i^{k_i}$. Then $u \in U_{S,\xi}$. 
 Applying \cite[Lemma 4.2]{CHKL2015} to $v$, we obtain $b\in C[x]$ and $\tilde{v} \in V_{\xi}$ such that $v=\cP_{\xi}(b)+\tilde{v}$. 
 It follows that $g=\Delta_{\xi}(a+b)+u+\tilde{v}/\den(\xi)$.  On the other hand, $\im(\Delta_{\xi}) \cap (U_{S,\xi}+V_{\xi}/\den(\xi))=\{0\}$ by \cite[Proposition~3.5]{CDGHL2025}. Hence, $C(x)=\im(\Delta_{\xi}) \oplus (U_{S,\xi}+V_{\xi}/\den(\xi))$.\\
 Let $w \in U_{S,\xi}\cap V_{\xi}/\den(\xi)$. Suppose that $w \neq 0$. Then $\deg_x(\den(w))>0$. Since $\den(w)$ is strongly coprime with $\den(\xi)$, $\gcd(\den(w), \den(\xi))=1$. So $\den(\xi)w \notin C[x]$, which contradicts to $\den(\xi)w \in V_{\xi}$. Consequently, $C(x)=\im(\Delta_{\xi}) \oplus U_{S,\xi} \oplus V_{\xi}/\den(\xi)$. In particular, the map $\phi_{\xi}:C(x)\to  U_{S,\xi} \oplus V_{\xi}/\den(\xi)$ defined by $\phi_{\xi}(g)=u+\tilde{v}/\den(\xi)$ forms a complete reduction for $(C(x), \Delta_{\xi})$ and $(a+b,u+\tilde{v}/\den(\xi))$ is a $\Sigma$-pair of $g$ with respect to  $\phi_{\xi}$.  By construction, such a $\Sigma$-pair is constructed if $C$ is factorizable.
\end{proof}

 \begin{prop}\label{PROP:Generalrational+highorder}
 Let $(C(x),\,\sigma)$ be the rational difference field with $\sigma(x)=x+1$, where $C$ is factorizable. Then, there is an algorithm that, for every $f\in C(x)^*$ and $i\in\bZ^{+}$, constructs a complete reduction $\phi^{(i)}_f$ for $(C(x), \Delta^{(i)}_f)$. 
 \end{prop}
\begin{proof} By Algorithm {\tt RNF} in \cite{AbPe2002b} (which is executable if $C$ is factorizable), we obtain a normal form $(\xi, \eta)$ of $f$. Note that $\xi$ is $\sigma$-reduced. So there is an algorithm for constructing a complete reduction $\phi_{\xi}$ for $(C(x), \Delta_{\xi})$ by Proposition~\ref{PROP:CR+Rational}. In particular, let
  $g\in C(x)$. One can compute a $\Sigma$-pair  $(u, r)$ of $\eta g$ with respect to $\phi_{\xi}$. 
 It follows from Lemma \ref{LM:ks} that $\phi_f:= \eta^{-1} \circ \phi_{\xi} \circ \eta$ is a complete reduction for $(C(x), \Delta_f)$ and $(\eta^{-1}u, \eta^{-1}r)$ is a $\Sigma$-pair of $g$ under $\phi_f$.  Finally, we can apply Lemma \ref{LM:GeneralRational} to obtain the desired result.
\end{proof}

\begin{alg}{\tt CRForRationalFunctions}$(C(x), g, f,i)$ \label{ALG:CRForRationalFunctions}

 \noindent
{\tt Input:} $g, f \in C(x)$ with $f\neq 0$, $i \in \bZ^{+}$. 

\noindent From outside accessible: a table $\bT$ indexed by every element $w$ in $C(x)$, the value of $w$ is $(\xi, \eta, S)$, where $(\xi, \eta)$ is a normal form of $w$ and $S\subseteq M_x$ whose elements are pairwise $\sigma$-coprime and strongly $\sigma$-coprime with $\xi$.

 \noindent
{\tt Output:} A $\Sigma$-pair of $g$ with respect to $\phi_f^{(i)}$ as given in Proposition~\ref{PROP:Generalrational+highorder}.

\smallskip \noindent
\begin{enumerate}
	\item[(1)] $(\tilde{g},\tilde{f}) \leftarrow g(ix), f(ix)$
\item[(2)] {\tt if} $\tilde{f}$ is an index of $\bT$,  {\tt then}
\begin{itemize}
  \item [] $(\xi,\eta, S) \leftarrow$ the value of $\tilde{f}$ 
\end{itemize} 
 {\tt else}
 \begin{itemize}
   \item [] $(\xi, \eta) \leftarrow$  {\tt RNF}$(\tilde{f})$ in \cite{AbPe2002b}, $S \leftarrow \{\}$, add $\tilde{f}=(\xi, \eta, S)$ to $\bT$
 \end{itemize}
 {\tt end if}
 
	\item[(3)] compute the $\Sigma$-pair $(u, r)$ of $\eta \tilde{g}$ with respect to\ $\phi_{\xi}$ by Proposition~\ref{PROP:CR+Rational}.\\
During the computation, update $S$ to  $\tilde{S}$, and update the value of $\tilde{f}$ to  $(\xi, \eta, \tilde{S})$ in $\bT$
	\item[(4)] {\tt return} $\left(\eta^{-1}\left(\frac{x}{i}\right)u\left(\frac{x}{i}\right), \eta^{-1}\left(\frac{x}{i}\right)r\left(\frac{x}{i}\right) \right)$
\end{enumerate}
\end{alg}

\begin{remark}
The detailed procedure for updating $S$ to $\tilde{S}$ in step 3 as follows:
to find a $\Sigma$-pair of $\eta\tilde g$, we start with computing a strong $\sigma$-factorization of $\den(\eta \tilde{g})$ with respect to $\xi$.
 Let the factorization be given by \eqref{Equ:SigmaFacFull}. 
Then one check if there is $q\in S$ and $m\in \bZ$ such that $\sigma^m(p_1)=q$. If this is the case, we replace $p_1$ with $q$, $k_1$, $\ell_1$ by $k_1-m$ and $\ell_1-m$, respectively. Otherwise, we add $p_1$ to $S$. Next we repeat this process for $p_2, \dots, p_s$, simultaneously completing the update of the set $S$.
Finally, we compute the desired $\Sigma$-pair by following the steps in the proof of Proposition~\ref{PROP:CR+Rational}.
\end{remark}

\begin{remark}
The reasons for taking $\bT$ as a global variable in the above algorithm are as follows. When using the algorithm to compute $\Sigma$-pairs with respect to the same complete reduction $\phi_f$ for $(C(x),\Delta_f)$ with $f \in C(x)^*$, but for different inputs $g_1, g_2 \in C(x)$, we must guarantee that the same $\phi_f^{(i)}$ is used. This requires choosing the same $\xi$, $\eta$, and the same set $S$ whenever $f$ appears in multiple computations. Since $S$ is not fixed in advance but constructed online, we reuse previously generated elements (e.g., those for $g_1$) and update $S$ when new elements are needed (e.g., for $g_2$). This is achieved by introducing a global table $\bT$, which stores, for every $w\in C(x)$, the triple $(\xi, \eta, S)$. In subsequent computations, we retrieve this information, thereby ensuring that the same complete reductions are always reused and, as a byproduct, avoiding redundant calculations.
\end{remark}

\begin{example}\label{EX:rational}
Let $f=-1$ and $g=1/(x+1)$. We are going to use Algorithm~\ref{ALG:CRForRationalFunctions} to compute a $\Sigma$-pair of $g$ with respect to $\phi_f^{(2)}$. Initially, we set the table $\bT$ to be empty. In step~$1$, $\tilde{f}=-1$ and $\tilde{g}=1/(2x+1)$. Then $\xi=-1$, $\eta=1$ and $S=\{\}$.
We add the entry $-1 \rightarrow (\xi, \eta, S)$ to $\bT$ in step 2. Let $\phi_{\xi}$ be the complete reduction for $(C(x), \Delta_{\xi})$ as
given in Proposition~\ref{PROP:CR+Rational}. In addition, the proposition delivers  a $\Sigma$-pair $(0, \tilde{g})$ of $\eta \tilde{g}$ with respect to $\phi_{\xi}$ using $S=\{\}$. The entry corresponding to $-1$ in $\bT$ is then updated to $(-1,1, [x+1/2])$. Finally, we get that $(0, 1/(x+1))$ is a  $\Sigma$-pair of $g$ with respect to $\phi_f^{(2)}$.
\end{example}

\subsection{The $R$-case} \label{SUBSECT:RCase}

\begin{convention}\label{CON:RCase}
In this section,  let $(\bA, \, \sigma)$ be a difference ring with constant field $C$, and $\gG$ be a multiplicative subgroup of $\bA^*$ that is closed under $\sigma$. Let $t$ be an $R$-monomial over $\bA$ with $\alpha:=\sigma(t)/t \in \gG$ of order $\lambda$. Fix $f=st^m$, where $s\in G$ and $m\in [\lambda-1]_0$. For $\ell \in \bZ^{+}$, set $a:=\alpha^{\sigma,\ell}$ and $\tilde \sigma:=\sigma^\ell$.
\end{convention}

\begin{hyp} \label{HYP:indR}
We assume  that there is a map that assigns to each element $z\in \gG$ and $i\in \bZ^+$ a 
complete reduction $\phi_z^{(i)}$ for $(\bA, \Delta_z^{(i)})$ as a $C$-linear map.
\end{hyp}

The goal of this section is to construct a complete reduction for $(\bA\langle t \rangle,\, \Delta_{f}^{(\ell)})$ with the assumption that Hypothesis~\ref{HYP:indR} is constructive. Recall that $\Delta_{f}^{(\ell)}=f\sigma^{\ell}-\mathbf{1}$.  In particular, we assume that for any element $g \in \bA$, a $\Sigma$-pair of $g$ with respect to $\phi_z^{(i)}$ can be computed explicitly by the function call \texttt{CRForGroundRing}($\bA,g,z,i$).

\begin{lemma}\label{LM:FormularForKshift}
    Let $u_0:=u\in \bA$ and $i\in [\lambda-1]_0$. Set 
    $ u_j= s\tilde\sigma(u_{j-1}) a^{i+(j-1)m}$ 
    for every $j\in\bZ^+$. Then, for every $k \in \bN$,
    \begin{itemize}
    \item [(i)]  $u_kt^{(i+km) \! \bmod \!\lambda}=u_0t^i+\Delta_f^{(\ell)}\left(\sum_{j=0}^{k-1}u_jt^{(i+jm) \! \bmod \!\lambda}\right);$
    \item [(ii)] $u_k=p_{i,k}\tilde\sigma^k(u_0)$ with $p_{i,k}=(sa^i)^{\tilde\sigma,k}(a^m)^{\tilde\sigma,\tilde\sigma,k}$.
    \end{itemize}
\end{lemma}

\begin{proof} 
	Recall that $a:=\alpha^{\sigma,\ell}$. By Lemma~\ref{LM:RPiSigma+term} (ii), 
	\begin{equation}\label{EQ:Rimage}
		\Delta_f^{(\ell)}(ut^i)=s\tilde\sigma(u)a^it^{(m+i) \! \bmod \!\lambda}-ut^i.
	\end{equation}
	We proceed the lemma by induction on $k$.
		
	(i) The conclusion is evident for the case $k=0$.
	Assume that $k\ge 0$  and  the conclusion holds for $k$. Set $w:=\sum_{j=0}^{k}u_jt^{(i+jm) \! \bmod \!\lambda}$. By the induction hypothesis, \eqref{EQ:Rimage} and the fact that $a^{\lambda}=1$,  we have
	\begin{align*}
		\Delta_f^{(\ell)}\left(w\right)&=\Delta_f^{(\ell)}(u_kt^{(i+km) \! \bmod \!\lambda}) + \Delta_f^{(\ell)}\biggl(\sum_{j=0}^{k-1}u_jt^{(i+jm)\! \bmod \!\lambda}\biggr) \\
		& = \Delta_f^{(\ell)}(u_kt^{(i+km) \! \bmod \!\lambda}) + u_kt^{(i+km) \! \bmod \!\lambda}-u_0t^i \\
		&=s\tilde{\sigma}(u_k)a^{i+km}t^{(i+(k+1)m)\! \bmod \!\lambda}-u_kt^{(i+km)\! \bmod \!\lambda}+ u_kt^{(i+km) \! \bmod \!\lambda}-u_0t^i\\
		&=u_{k+1}t^{(i+(k+1)m)\! \bmod \!\lambda}-u_0t^i.  
	\end{align*}

	(ii) The case $k=0$ holds trivially. Assume that $k>0$ and the conclusion holds for $k$. Then by the induction hypothesis and Remark \ref{RE:sigma+notation}, 
	\begin{align*}
		u_{k+1}&=s\tilde{\sigma}(u_k)a^{i+km}=s \tilde{\sigma}((sa^i)^{\tilde\sigma,k}(a^m)^{\tilde\sigma,\tilde\sigma,k}\tilde{\sigma}^k(u_0))a^{i+km}\\
		&=sa^i\tilde{\sigma}((sa^i)^{\tilde\sigma,k})a^{km}\tilde{\sigma}((a^m)^{\tilde\sigma,\tilde\sigma,k})\tilde\sigma^{k+1}(u_0)\\
		&=(sa^i)^{\tilde\sigma,k+1}(a^m)^{\tilde\sigma,\tilde\sigma,k+1}\tilde\sigma^{k+1}(u_0)=p_{i,k+1}\tilde\sigma^{k+1}(u_0).
	\end{align*}
\end{proof}

\begin{lemma}\label{LM:ReductionForRCase}
	Let $u\in \bA$ and $i\in [\lambda-1]_0$. Set $d:=\gcd(\lambda, m)$ and $\mu:=\lambda/d$. Then there are $q \in \bA \langle t \rangle$, $ i'\in [d-1]_0$ and $r \in \im(\phi_{p_{i',\mu}}^{(\mu\ell)})$ with $p_{i',\mu}=(sa^{i'})^{\tilde\sigma,\mu}(a^m)^{\tilde\sigma,\tilde\sigma,\mu}$ such that
	$ut^i=\Delta_f^{(\ell)}(q)+rt^{i'}$. If $(\bA,\,\sigma)$ is computable and Hypothesis~\ref{HYP:indR} is constructive, such $q$ and $r$ can be computed. 
\end{lemma}
\begin{proof}
Let  $i= i' \bmod d$. Then using the extended Euclidean algorithm, there is $k\in \bN$ such that $i'=(i+km) \bmod \lambda$.
For every $j\in \bN$, set $p_{i,j}=(sa^i)^{\tilde\sigma,j}(a^m)^{\tilde\sigma,\tilde\sigma,j}$. By Lemma~\ref{LM:FormularForKshift}, there is $q_1=-\sum_{j=0}^{k-1} p_{i,j}\tilde\sigma^j(u) t^{(i+jm) \! \bmod \!\lambda}\in \bA\langle t \rangle$ and $r'= p_{i,k}\tilde\sigma^k(u)\in \bA$ such that $ut^i=\Delta_f^{(\ell)}(q_1)+r't^{i'}$.
Since $s,a\in\gG$ and $\gG$ is a multiplicative group that is closed under $\sigma$, we have $p_{i',\mu}\in \gG$. Further, by Hypothesis~\ref{HYP:indR}, there is $v\in \bA$ such that $r'=\Delta^{(\mu\ell)}_{p_{i',\mu}}(v)+\phi_{p_{i',\mu}}^{(\mu\ell)}(r')$.
Using Lemma~\ref{LM:FormularForKshift} again, we get
$$\Delta^{(\mu\ell)}_{p_{i',\mu}}(v)t^{i'}=p_{i',\mu} \tilde\sigma^\mu(v)t^{(i'+\mu m)\! \bmod \!\lambda}-vt^{i'}=\Delta_f^{(\ell)}(q_2)$$
with $q_2=\sum_{j=0}^{\mu-1} p_{i',j}\tilde\sigma^j(v) t^{(i'+jm) \! \bmod \!\lambda} \in \bA\langle t \rangle$, 
which implies that $r't^{i'}=\Delta_f^{(\ell)}(q_2)+\phi_{p_{i',\mu}}^{(\mu\ell)} (r')t^{i'}$.
So $ut^i =\Delta_f^{(\ell)}(q_1+q_2)+\phi_{p_{i',\mu}}^{(\mu\ell)}(r')t^{i'}$. Therefore, the proof is completed by letting $q=q_1+q_2$ and $r=\phi_{p_{i',\mu}}^{(\mu\ell)}(r')$.
\end{proof}

Next, we construct a complementary space of the $C$-subspace $\im (\Delta_f^{(\ell)})$.  

\begin{lemma}\label{LM:ComplementforRExt}
Let $d:=\gcd(\lambda,m)$, $\mu:= \lambda/ d$, and $p_{i,\mu}:=(sa^i)^{\tilde\sigma,\mu}(a^m)^{\tilde\sigma,\tilde\sigma,\mu}$. Then
    $$\bA\langle t\rangle=\im (\Delta_f^{(\ell)})\oplus \biggl(\directsum_{i=0}^{d-1}\im(\phi_{p_{i,\mu}}^{(\mu\ell)})t^i\biggr).$$ 
    In particular, the projection from $\bA\langle t\rangle$ to $ \directsum_{i=0}^{d-1}\im(\phi_{p_{i,\mu}}^{(\mu\ell)})t^i$ is a complete reduction $\psi_f^{(\ell)}$ for $(\bA\langle t\rangle,\, \Delta_f^{(\ell)})$. 
    Moreover, if $(\bA,\,\sigma)$ is computable and Hypothesis~\ref{HYP:indR} is constructive, one can compute a $\Sigma$-pair of any element of $\bA\langle t\rangle$  with respect to $\psi_f^{(\ell)}$.    
\end{lemma}
\begin{proof} 
Recall that $C=\const(\bA,\,\sigma)=\const(\bA\langle t \rangle,\sigma)$ because $t$ is an $R$-monomial over $\bA$. Then the $\im(\phi_{p_{i,\mu}}^{(\mu\ell)})t^i$ are $C$-subspaces of $\bA \langle t \rangle$, which implies that	
$W:=\sum_{i=0}^{d-1}\im(\phi_{p_{i,\mu}}^{(\mu\ell)})t^i$ forms a $C$-subspace of $\bA\langle t \rangle $. 
 Every element in $ \bA\langle t \rangle$ can be uniquely written as $\sum_{i=0}^{\lambda-1}g_it^i$, where $g_i \in \bA$. 
 By Lemma~\ref{LM:ReductionForRCase}, we obtain $g_it^i=\Delta^{(\ell)}_f(q_i)+r_{i}$, where
 $q_i\in\bA\langle t\rangle$ and $r_{i}\in \im(\phi_{p_{i',\mu}}^{(\mu\ell)})t^{i'}$ for some $i'\in [d-1]_0$. Consequently, by setting $q=\sum_{i=0}^{\lambda-1}q_i\in\bA\langle t\rangle$ and $r=\sum_{i=0}^{\lambda-1}r_{i}\in W$, we conclude that $g=\Delta^{(\ell)}_f(q)+r$. 
 Hence, $\bA\langle t\rangle=\im (\Delta_f^{(\ell)})+W$. If $(\bA,\,\sigma)$ is computable and Hypothesis~\ref{HYP:indR} is constructive, such $q$ and $r$ can be computed explicitly.
 
    Assume that $g\in \im (\Delta_f^{(\ell)})\cap W$. Then $g=\Delta_f^{(\ell)}(q)$ for some $q\in \bA\langle t \rangle$. Write $q$ as $\sum_{i=0}^{\lambda-1}q_i t^i$, where $q_i\in \bA$.
    Let $i_0\in [d-1]_0$. For every $j\in\bZ^+$, set $i_j$ to be $(i_{j-1}+m) \! \bmod \!\lambda$. Note that $i_j=i_0$ if and only if $\lambda$ divides $jm$.
    So the minimal element in $\{j \in \bZ^{+} \mid i_j=i_0\}$ is equal to $\mu$.
   In other words, $i_0,\ldots,i_{\mu-1}$ are pairwise distinct.
   Moreover, $i_1,\ldots,i_{\mu-1}$ strictly lie between $d-1$ and $\lambda$. By the construction of $i_j$, for every $j \in [\mu-1]$, the coefficient  of $t^{i_j}$ in $\Delta_f^{(\ell)}(q)$ is  $s\tilde\sigma(q_{i_{j-1}}) a^{i_{j-1}} - q_{i_j}$. 
   And $s\tilde\sigma(q_{i_{\mu-1}}) a^{i_{\mu-1}} - q_{i_0}$ is the coefficient  of $t^{i_0}$ in $\Delta_{f}^{(\ell)}(q)$, denoted by $w$. Clearly, $w\in \im(\phi_{p_{i,\mu}}^{(\mu\ell)})$.
   Since $g$ is spanned by $1,t,\ldots,t^{d-1}$, $q_{i_j}=s\tilde\sigma(q_{i_{j-1}}) a^{i_{j-1}}$ for every $j \in [\mu-1]$. Replacing $i$, $u_0$ with $i_0$ and $q_{i_{0}}$  in Lemma~\ref{LM:FormularForKshift}~(ii), respectively, we get $s\tilde\sigma(q_{i_{\mu-1}}) a^{i_{\mu-1}}=p_{i_0,\mu}\tilde\sigma^\mu(q_{i_0})$. 
   It follows that 
   $w=p_{i_0,\mu}\tilde\sigma^\mu(q_{i_0}) - q_{i_0}=\Delta^{(\mu\ell)}_{p_{i_0,\mu}}(q_{i_0})$. Hence, $w=0$
   by $w \in \im(\phi_{p_{i,\mu}}^{(\mu\ell)}) \cap \im (\Delta^{(\mu\ell)}_{p_{i_0,\mu}})=\{0\}$. That means, for every $i_0\in[d-1]_0$, the coefficient of $t^{i_0}$ in $g$ must be equal to zero. Consequently, $g=0$. 
\end{proof}

 The constructive version of Lemma~\ref{LM:ComplementforRExt} can be summarized by the following algorithm.
 
 \begin{alg}{\tt CRForSimpleR-Extensions}$(\bA\langle t\rangle,g,f,\ell)$\label{ALG:Rcase}
	
	\noindent
	{\tt Input:} The data in Convention \ref{CON:RCase}, $g\in \bA\langle t\rangle$; we assume that $(\bA,\,\sigma)$ is computable and Hypothesis~\ref{HYP:indR} is constructive.
	
	\noindent
	{\tt Output:} A $\Sigma$-pair of $g$ with respect to $\psi_{f}^{(\ell)}$ as given in Lemma~\ref{LM:ComplementforRExt}.
    
    \noindent
   {\tt Remark:} Set $p_{i, j}:=(sa^i)^{\sigma,j}(a^m)^{\sigma,\sigma,j}$ for every $i\in \bZ$ and $j\in \bN$ in the following pseudo-code.
	
	\smallskip \noindent
	\begin{enumerate}
		\item[(1)] $d \leftarrow \gcd(\lambda,m)$, $\mu \leftarrow\lambda/ d$, $(q, r) \leftarrow 0,0$
		\item[(2)] {\tt for} $i$ {\tt from} $0$ {\tt to} $\lambda-1$ {\tt do} 
		\begin{itemize}
			\item []  $i'\leftarrow i \bmod  d$, \, $u \leftarrow \coeff(g,t,i)$, \, 
			\item [] compute $k\in \bN$ s.t.\ $i'= (i+km) \bmod \lambda$ 
            \item [] $(q, r') \leftarrow q-\sum_{j=0}^{k-1}p_{i,j}\tilde{\sigma}^j(u)t^{(i+jm) \! \bmod \!\lambda},\, p_{i,k}\tilde{\sigma}^k(u)$
			\item [] $(v, w) \leftarrow$ \texttt{CRForGroundRing}($\bA,r',p_{i', \mu},\ell$),
            \item [] $(q, r) \leftarrow q+\sum_{j=0}^{\mu-1}p_{i',j}\tilde{\sigma}^j(v)t^{(i'+jm) \! \bmod \!\lambda},\,  r+wt^{i'}$
		\end{itemize} 
		{\tt end do}
		\item[(3)] {\tt return} $(q, r)$	
	\end{enumerate}
\end{alg}   

 \begin{example}\label{EX:Rmonomials}
 Let $\bA:=C(x)$ with $\sigma(x)=x+1$, $\gG=C(x)^*$ and $t$ be an $R$-monomial over $\bA$ with $\sigma(t)=-t$. For $f=-t$ and $g=t/x$, we compute a $\Sigma$-pair of $g$ with respect to $\psi_f$. Note that $\lambda=2$,  $\ell=m=1$, $s=a=-1$, $\tilde{\sigma}=\sigma$ as in Convention \ref{CON:RCase}. It follows that $d=1$, $\mu=2$ and $p_{0, \mu}=-1$ as given in Lemma~\ref{LM:ComplementforRExt}. Moreover, $\bA \langle t\rangle=\im(\Delta_f)\oplus \im(\phi_{p_{0,2}}^{(2)})$. By Lemma \ref{LM:FormularForKshift}, $g=\Delta_{f}(-t/x)+r$, where $r=1/(x+1)$. Moreover, $r \in \im(\phi_{p_{0,2}}^{(2)})$ by Example~\ref{EX:rational}. So $(-t/x, 1/(x+1))$ is a $\Sigma$-pair of $g$ with respect to $\psi_f$.
 \end{example}

Lemma~\ref{LM:ComplementforRExt} leads directly to the following proposition, which improves upon~\cite[Proposition~7.12]{Schn2016} by extending the construction from telescoping to complete reductions without requiring the solution of difference equations in the ground ring.

\begin{prop}\label{Prop:NestRExt}
Let $(\bA, \, \sigma)$ be a computable difference ring with constant field $C$
and $(\bE_n,\,\sigma)$ with $\bE_n:=\bA\langle t_1\rangle \ldots \langle t_n\rangle$ be a simple $R$-tower over $(\bA, \, \sigma)$. 
Assume that there is an algorithm which, for every $z\in \bA^*$ and $i\in \bZ^{+}$, constructs a complete reduction $\phi_z^{(i)}$ for $(\bA,\,\Delta_{z}^{(i)})$ (viewed as a $C$-linear map). Then there is an algorithm that,
for every $f \in {(\bA^*)}_{\bA}^{\bE_n}$ and $\ell \in \bZ^{+}$, constructs a complete reduction $\psi_f^{(\ell)}$ as a $C$-linear map for $(\bE_n,\, \Delta_f^{(\ell)})$. 
   \end{prop}
  \begin{proof} 
We proceed by induction on $n$. It is evident for the case $n=0$.
Assume that, for every $z \in {(\bA^*)}_{\bA}^{\bE_{n-1}}$ and $i \in \bZ^{+}$, one can construct a complete reduction $\phi_{n-1,z}^{(i)}$ as a $C$-linear map for $(\bE_{n-1},\, \Delta_z^{(i)})$ and assume that the map in Hypothesis~\ref{HYP:indR} is defined by $(z,i)\rightarrow \phi_{n-1,z}^{(i)}$. Since $(\bE_n,\sigma)$ is a simple $R$-tower over $(\bA,\,\sigma)$, ${(\bA^*)}_{\bA}^{\bE_{n-1}}$ is closed under $\sigma$ and $\sigma(t_n)/t_n \in {(\bA^*)}_{\bA}^{\bE_{n-1}}$. Moreover, ${(\bA^*)}_{\bA}^{\bE_{n}}=\{st_n^m~|~s\in {(\bA^*)}_{\bA}^{\bE_{n-1}}~\text{and}~m\in [\lambda-1]_0\}$ with $\ord(t_n)=\lambda$. For every $f\in {(\bA^*)}_{\bA}^{\bE_{n}}$ and $\ell \in \bZ^{+}$,  a complete reduction $\psi_{f}^{(\ell)}$ for $(\bE_n, \Delta_f^{(\ell)})$ is  obtained from Lemma~\ref{LM:ComplementforRExt}. The map in Hypothesis~\ref{HYP:indR} is then defined by $(f, \ell) \mapsto \psi_{f}^{(\ell)}$. 
\end{proof}

For instance, by specializing the ground difference ring $\bA$ in Proposition~\ref{Prop:NestRExt} to the rational difference field, one can use Algorithm~\ref{ALG:CRForRationalFunctions} in the base case to obtain a complete algorithm for simple $R$-towers over $(C(x),\sigma^k)$ for every $k \in \bZ^{+}$. For simplicity, we only state the case $k=1$. The same simplification is also made in Algorithms \ref{ALG:CRForSimpleRPi-Towers} and \ref{ALG:CRForSimpleRPiSigma-Towers} in the sequel.

\begin{alg}\label{ALG:CRForSimpleR-Towers}
{\tt CRForSimpleR-Towers}$(\bE_n, g, f,\ell)$

 \noindent
{\tt Input:} A simple R-tower $(\bE_n,\sigma)$ with $\bE_n:=C(x)\langle t_1\rangle \ldots \langle t_n\rangle$ over the rational difference field $(C(x),\,\sigma)$, $f\in (C(x)^*)_{C(x)}^{\bE_n}$, $g \in\bE_n$ and $\ell \in \bZ^{+}$; we assume that $C$ is factorizable.

\noindent
{\tt Output:} A $\Sigma$-pair of $g$ with respect to $\psi_{f}^{(\ell)}$ as given in Proposition~\ref{Prop:NestRExt}.
\begin{enumerate}
	\item[]  {\tt if} $n=0$ {\tt then} 
\begin{itemize}
  \item [] {\tt return} {\tt CRForRationalFunctions}$(C(x),g, f, \ell)$ (*Algorithm~\ref{ALG:CRForRationalFunctions}*)
\end{itemize}
\item[] {\tt else}
 \begin{itemize}
   \item[] {\tt return}  {\tt CRForSimpleR-Extensions}$(\bE_{n-1}\langle t_n \rangle, g, f,\ell)$
   \item[]\hspace*{1.2cm}\begin{minipage}{9cm} where  the command \texttt{CRForGroundRing} is replaced by \texttt{CRForSimpleR-Towers} (*Algorithm~\ref{ALG:Rcase}*)
   \end{minipage}  
   
 \end{itemize}
\item[]  {\tt end if}
\end{enumerate}
\end{alg}

\subsection{The $\Pi$-case}\label{SUBSECT:PiCase}

\begin{convention}\label{CON:PiCase}
In this section, let $(\bA, \, \sigma)$ be a difference ring with constant field $C$, and $\gG$ be a multiplicative subgroup of $\bA^*$. Let $t$ be a $\Pi$-monomial over $\bA$ with $a:=\sigma(t)/t \in \gG$. Set $f=st^m$, where $s\in G$ and $m\in \bZ$.
\end{convention}

\begin{hyp} \label{HYP:PiInd}
We assume that there is a map that assigns to each element $z$ of $\gG$ 
a complete reduction $\phi_z$ for $(\bA,\, \Delta_z)$. 
\end{hyp}

We are going to construct a complete reduction for $(\bA\langle t \rangle, \,\Delta_{f})$ under the assumption that Hypothesis~\ref{HYP:PiInd} is constructive. In particular, we assume that, for any element $g$ in $\bA$, a $\Sigma$-pair of $g$ with respect to $\phi_z$ can be computed explicitly  by the function call \texttt{CRForGroundRing}($\bA,g,z$).

To this end,  we need some preparation.

 \begin{lemma}\label{LM:PiFormularForKshift}
    Let $i\in \bZ$ and $u_0:=u\in \bA$. Set $p_{i,j}:=(sa^i)^{\sigma,j}(a^m)^{\sigma,\sigma,j}$
    for every $j\in\bZ^+$. Then, for every $k \in \bZ^{+}$,  there is  $v_i=\sum_{j=0}^{k-1}p_{i,j}\sigma^j(u)t^{i+jm} \in \bA[t, t^{-1}]$ such that
    \begin{equation*}\label{EQ:image}
    p_{i,k}\sigma^k(u) t^{i+km}= ut^i+\Delta_f(v_i).
    \end{equation*}
\end{lemma}
\begin{proof}
The proof is analogous to that of Lemma~\ref{LM:FormularForKshift}.
\end{proof}

Recall that $C=\const(\bA,\,\sigma)=\const(\bA[t,t^{-1}],\,\sigma)$, since $t$ is a $\Pi$-monomial over $\bA$.
Note further that $\bA[t,t^{-1}]=\bA[t] \oplus t^{-1}\bA[t^{-1}]$, where $\bA[t]$ and $t^{-1}\bA[t^{-1}]$ are understood as $C$-subspaces of $\bA[t,t^{-1}]$. For $k \in \bZ^{+}$, set $$\bA[t]_{<k}:=\{p \in \bA[t] \mid \deg_t(p)<k\}.$$ 

Next, we construct the complementary space of the $C$-subspace $\im(\Delta_f)$ based on the different values of $m$.  In other words, we construct a $C$-subspace $W$ of $\bA[t,t^{-1}]$ such that
$\bA[t,t^{-1}]=\im(\Delta_f)\oplus W$.

\begin{lemma}\label{LM:m>0}
If $m>0$, then $\bA[t,t^{-1}]=\im(\Delta_f) \oplus \bA[t]_{<m}.$ In particular, for $g \in \bA[t, t^{-1}]$, one can explicitly compute $q \in \bA[t, t^{-1}]$ such that $g-\Delta_{f}(q)\in \bA[t]_{<m}$.
\end{lemma}
\begin{proof} 
    Let $i\in\bZ$ and $u\in \bA$. First we show that $ut^i \in \im (\Delta_f)+ \bA[t]_{< m}$. If $i \in [m-1]_0$, $ut^i \in \bA[t]_{<m}$. Set $p_{i,j}:=(sa^i)^{\sigma,j}(a^m)^{\sigma,\sigma,j}$
    for every $j\in\bZ^+$. If $i<0$, then there is $k \in \bZ^{+}$ and $i'\in [m-1]_0$  such that $i'=i+km$.
    By Lemma~\ref{LM:PiFormularForKshift}, we can take  $v=-\sum_{j=0}^{k-1}p_{i,j}\sigma^{j}(u)t^{i+jm}$ such that
    \[ut^i-\Delta_f(v)=p_{i,k}\sigma^k(u) t^{i'}\in \bA[t]_{< m}.\]
    Assume that $i\ge m$.
    Then there are $i'\in [m-1]_0$ and $k\in \bZ^{+}$ such that $i'+km=i$. 
    Note that $p_{i',k} \in \bA^*$ by $s, a\in G$. 
   Replacing $i$ and $u$ in Lemma~\ref{LM:PiFormularForKshift} with $i'$ and $\sigma^{-k}\left(u/p_{i',k}\right)$, respectively, we can take  $v=\sum_{j=0}^{k-1}p_{i',j}\sigma^{j-k}\left(\frac{u}{p_{i', k}}\right)t^{i'+jm}$ 
   such that
    \[ u t^{i}-\Delta_f(v)=\sigma^{-k}\left(\frac{u}{p_{i',k}}\right)t^{i'}\in \bA[t]_{< m}.\]

Let $g\in\bA[t,t^{-1}]$. We decompose $g$ as $g = g^{+}+g^{-}$, where $g^{+}\in\bA[t]$ and $g^{-}\in t^{-1}\bA[t^{-1}]$. 
Based on the above analysis, we can construct $q^{-} \in t^{-1}\bA[t^{-1}]$ and $r^{-} \in \bA[t]_{<m}$ such that  $g^{-}=\Delta_f(q^{-})+r^{-}$. 
Similarly, we can construct $q^{+} \in \bA[t]$ and $r^{+} \in \bA[t]_{<m}$ such that $g^{+}=\Delta_f(q^{+})+r^{+}$. 
Set $q:=q^{+}+q^{-}\in\bA[t,t^{-1}]$ and $r:=r^{+}+r^{-}\in\bA[t]_{<m}$. 
Then $g=\Delta_f(q)+r$, which implies that $\bA[ t,t^{-1}]=\im (\Delta_f)+ \bA[t]_{< m}$. 
 
It remains to show that $\im(\Delta_f) \cap \bA[t]_{<m}=\{0\}$. Suppose that $p \in \im(\Delta_f) \cap \bA[t]_{<m}$ is nonzero. Then there is $q \in \bA[t,t^{-1}]$ such that $p=\Delta_f(q)$.  
We arrange $q$ in descending degree, and write it as $q_{k}t^{k}+\cdots+q_{\ell}t^{\ell}$, where $q_{\ell}, q_k \in \bA$ with $q_{\ell}q_k \neq 0$.
Since $m>0$, $\tdeg_t(\Delta_f(q))=\ell$. So $\tdeg_t(q)=\tdeg_t(\Delta_f(q))$, which implies that $q\in \bA[t]$ by $\Delta_f(q)=p\in \bA[t]$.
Moreover, the leading term of $\Delta_f(q)$ is equal to $sa^k\sigma(q_k)t^{m+k}$ and $k\ge 0$. 
 Then $sa^k\sigma(q_k)=0$ by $\deg_t(p)<m$. 
 Note that $s, a \in \bA^*$ and $\sigma$ is an automorphism. 
 So $q_k=0$, which implies that $q=0$, a contradiction. 
 Consequently, $\bA[t,t^{-1}]=\im(\Delta_f) \oplus \bA[t]_{<m}.$ 
\end{proof}

The proof and the underlying construction for the case $m<0$, summarized in the next lemma, proceed analogously and are thus omitted.

\begin{lemma}\label{LM:m<0}
If $m<0$, then $\bA[t,t^{-1}]=\im(\Delta_f) \oplus t^m\bA[t]_{<-m}.$ In particular, for $g \in \bA[t, t^{-1}]$, one can explicitly compute $q \in \bA[t, t^{-1}]$  
such that $g-\Delta_{f}(q)\in t^m\bA[t]_{<-m}$.
\end{lemma}

\begin{lemma}\label{LM:m=0}
If $m=0$, then $\bA[t,t^{-1}]=\im(\Delta_f) \oplus (\directsum_{i \in \bZ}{\im(\phi_{sa^i})t^i}).$ In particular, for every $g \in \bA[t, t^{-1}]$, one can construct $q \in \bA[t, t^{-1}]$ 
such that $g-\Delta_{f}(q)\in \directsum_{i \in \bZ}{\im(\phi_{sa^i})t^i}$.
\end{lemma}
\begin{proof} 
Since $m=0$,  
\begin{equation}\label{EQ:m=0}
\Delta_f(ut^i)=(sa^i\sigma(u)-u)t^i 
\end{equation}
for every $u \in \bA$ and $i \in \bZ$.
 Let $g=\sum_{i\in \bZ}g_it^i$, where $g_i \in \bA$, finitely many nonzero. Note that $s, a \in G$ and $G$ is a multiplicative group. So $sa^i\in G$. By Hypothesis~\ref{HYP:PiInd}, $g_i=\Delta_ {sa^i}(u_i)+\phi_{sa^i}(g_i)$ for some $u_i \in \bA$, which, together with \eqref{EQ:m=0}, implies that $g_it^i=\Delta_f(u_i t^i)+\phi_{sa^i}(g_i)t^i$. Set $q:=\sum_{i \in \bZ} u_it^i$ and $r=\sum_{i \in \bZ}\phi_{sa^i}(g_i)t^i$. Then $g=\Delta_f(q)+r$.
 So $\bA[t,t^{-1}] = \im(\Delta_f) +(\sum_{i \in \bZ}{\im(\phi_{sa^i})t^i})$. 
  Finally, we show that $\im(\Delta_f) \cap (\directsum_{i \in \bZ}\im(\phi_{sa^i})t^i)=\{0\}$. Let $p\in \im(\Delta_f) \cap (\directsum_{i \in \bZ}\im(\phi_{sa^i})t^i)$. Then there is $v \in \bA[t,t^{-1}]$ such that $p=\Delta_f(v)$. Write $p, v$ as $\sum_{i\in \bZ}p_it^i, \sum_{i\in \bZ} v_it^i$, where $p_i, v_i \in \bA$, finitely many nonzero, respectively. 
 Then $p=\sum_{i\in \bZ}(\Delta_{sa^i}(v_i)t^i)$, which, together with $p \in \sum_{i \in \bZ}\im(\phi_{sa^i})t^i$, 
 implies that $p_i \in \im(\Delta_{sa^i})\cap \im(\phi_{sa^i})$ for all $i\in \bZ$. Hence, $p_i=0$. 
\end{proof}

Combining Lemmas \ref{LM:m>0}, \ref{LM:m<0} and \ref{LM:m=0}, we obtain

\begin{lemma}\label{LM:PiExt}
One can construct a $C$-subspace $W$ such that $\bA[t, t^{-1}]=\im(\Delta_f)\oplus W$. Depending on whether $m>0$, $m<0$ or $m=0$,  $W$ is given respectively by $\bA[t]_{<m}$, $t^m\bA[t]_{<-m}$ or $\directsum_{i \in \bZ}{\im(\phi_{sa^i})t^i} $. Moreover, the projection from $\bA[t, t^{-1}]$ onto $W$ is a complete reduction $\psi_f$ for $(\bA[t,t^{-1}],\, \Delta_f)$. If $(\bA,\,\sigma)$ is computable and Hypothesis~\ref{HYP:PiInd} is constructive, one can compute $\Sigma$-pairs for all elements in $\bA[t,t^{-1}]$ with respect to $\psi_f$. 
\end{lemma}

Its algorithmic version can be summarized as follows.

\begin{alg}{\tt CRForSimple$\Pi$-Extensions}$(\bA\langle t\rangle,g,f)$\label{ALG:Picase}

 \noindent
{\tt Input:} The data in Convention \ref{CON:PiCase}, $g\in \bA\langle t\rangle$; we assume that $(\bA,\,\sigma)$ is computable and Hypothesis~\ref{HYP:PiInd} is constructive.

\noindent
{\tt Output:} A $\Sigma$-pair of $g$ with respect to $\psi_{f}$ as given in Lemma~\ref{LM:PiExt}.

\noindent
{\tt Remark:} Set $p_{i, j}:=(sa^i)^{\sigma,j}(a^m)^{\sigma,\sigma,j}$ for every $i\in \bZ$ and $j\in \bN$ in the following pseudo-code.

\smallskip \noindent
\begin{enumerate}
	\item[(1)] $(d, \ell) \leftarrow \hdeg_t(g), \tdeg_t(g)$, $(q, r) \leftarrow 0, 0$

\item[(2)] {\tt if} $m>0$ {\tt then}
\begin{itemize}
  \item [] {\tt for} $i$ {\tt from} $\ell$ {\tt to} $d$ {\tt do}
  \begin{itemize}
    \item [] $u \leftarrow \coeff(g, t, i)$
    \item [] {\tt if} $0 \le i <m$ {\tt then} $r \leftarrow r+ut^i$ {\tt end if}
    \item [] {\tt if} $i<0$ {\tt then}
    \begin{itemize}
      \item []  compute $k \in \bZ^{+}$ and $i' \in [m-1]_0$ s.t. $i'=i+km$
      \item [] $(q, r) \leftarrow q-\sum_{j=0}^{k-1}p_{i,j}\sigma^{j}(u)t^{i+jm}, r+p_{i,k}\sigma^{k}(u)t^{i'}$
    \end{itemize}
     \item []{\tt end if}
    \item [] {\tt if} $i\ge m$ {\tt then}
    \begin{itemize}
      \item []  compute $k \in \bZ^{+}$ and $i' \in [m-1]_0$ s.t. $i'+km=i$
      \item [] $(q, r) \leftarrow q+\sum_{j=0}^{k-1}p_{i',j}\sigma^{j-k}\left(\frac{u}{p_{i', k}}\right)t^{i'+jm}, r+\sigma^{-k}\left(\frac{u}{p_{i', k}}\right)t^{i'}$
   \end{itemize}
    {\tt end if}
  \end{itemize}
   \item[] {\tt end do} 
\end{itemize}
{\tt end if}
\item[(3)] {\tt if} $m<0$ {\tt then}
\begin{itemize}
  \item [] {\tt for} $i$ {\tt from} $\ell$ {\tt to} $d$ {\tt do}
  \begin{itemize}
    \item [] $u \leftarrow \coeff(g, t, i)$
    \item [] {\tt if} $m < i \le 0$ {\tt then} $r \leftarrow r+ut^i$ {\tt end if}
    \item [] {\tt if} $i>0$ {\tt then}
    \begin{itemize}
      \item []  compute $k \in \bZ^{+}$ and $m <i'\le 0 $ s.t. $i'=i+km$
      \item [] $(q, r) \leftarrow q-\sum_{j=0}^{k-1}p_{i,j}\sigma^{j}(u)t^{i+jm}, r+p_{i,k}\sigma^{k}(u)t^{i'}$
    \end{itemize}
     \item []{\tt end if}
   \item [] {\tt if} $i\le m$ {\tt then}
    \begin{itemize}
      \item []  compute $k \in \bZ^{+}$ and $m <i'\le 0$ s.t. $i'+km=i$
      \item [] $(q, r) \leftarrow q+\sum_{j=0}^{k-1}p_{i',j}\sigma^{j-k}\left(\frac{u}{p_{i', k}}\right)t^{i'+jm}, r+\sigma^{-k}\left(\frac{u}{p_{i', k}}\right)t^{i'}$
    \end{itemize}
    {\tt end if}
  \end{itemize}
   \item[] {\tt end do} 
\end{itemize}
{\tt end if}
\item[(4)] {\tt if} $m=0$ {\tt then}
\begin{itemize}
  \item [] {\tt for} $i$ {\tt from}  $\ell$ {\tt to} $d$ {\tt do}
  \begin{itemize}
    \item [] $u \leftarrow \coeff(g, t, i)$
	\item [] $(v, w) \leftarrow$ \texttt{CRForGroundRing}($\bA,u,sa^i$)
    \item[] $(q, r) \leftarrow q+vt^i, r+wt^i$
    \end{itemize}
   \item[] {\tt end do} 
\end{itemize}
{\tt end if}
	\item[(5)] {\tt return} $(q, r)$
	
\end{enumerate}
 \end{alg}   
 
 \begin{example}\label{EX:RPiCase}
We will simplify the sum
$$\sum_{k=1}^{n}\frac{(1+k+2(-1)^kk)(-1)^{\lfloor \frac{k}{2} \rfloor}2^k}{k(1+k)}$$
as follows.
Let $(\bA,\,\sigma)$ with $\bA:=C(x)\langle y_1\rangle \langle y_2\rangle$ be the $R$-tower over $(C(x),\,\sigma)$, where  $\sigma(y_1)=-y_1$ and $\sigma(y_2)=-y_1y_2$. Take a $\Pi$-monomial $t$ over $\bA$ such that $\sigma(t)=2t$. Then $(\bE,\,\sigma)$ with $\bE:=\bA\langle t\rangle$ forms a simple $R\Pi$-tower over $(C(x),\,\sigma)$. Moreover,
$y_1$, $y_2$ and $t$ model $(-1)^k$, $(-1)^{\lfloor \frac{k}{2} \rfloor}$ and $2^k$, respectively. 
 The given summand is represented by $g:=\tilde{g}t\in\bE$  with $$\tilde{g}:=\frac{(1+x+2xy_1)y_2}{x(1+x)}.$$

We apply Algorithm~\ref{ALG:Picase} to compute a $\Sigma$-pair of $g$ with respect to $\psi_1$ as given in Lemma~\ref{LM:PiExt}. Note that the values in Convention \ref{CON:PiCase} are $a=2$,  $s=1$ and $m=0$.
So we enter step 4 of Algorithm~\ref{ALG:Picase} and 
it suffices to compute a $\Sigma$-pair of $\tilde{g}$ with respect to $\phi_2$ in $\bA$, where $\phi_2$ is a complete reduction for $(\bA,\, \Delta_2)$. Applying Algorithm~\ref{ALG:CRForSimpleR-Towers} to $\tilde{g}$ yields the $\Sigma$-pair $(-y_2/x, 0)$ of $\tilde{g}$  with respect to $\phi_2$. Hence, $(-y_2t/x, 0)$ is a $\Sigma$-pair of $g$ with respect to $\psi_1$. Rephrasing the arising variables to their corresponding summation objects, telescoping leads to the identity
 $$\sum_{k=1}^{n}\frac{(1+k+2(-1)^kk)(-1)^{\lfloor \frac{k}{2} \rfloor}2^k}{k(1+k)}=-\frac{(-1)^{\lfloor \frac{n+1}{2} \rfloor} 2^{n+1}}{n+1}-2.$$

\end{example}

Iterative application of Lemma~\ref{LM:PiExt} (similar to the proof of Proposition~\ref{Prop:NestRExt}) yields the following proposition, extending~\cite[Theorem~7.1~(ii)]{Schn2016} from telescoping to complete reductions.

\begin{prop}\label{PROP:NestedPiCase}
Let $(\bA, \, \sigma)$ be a computable difference ring with constant field $C$, and $G$ be a multiplicative subgroup of $\bA^*$.  Let 
 $(\bE_n,\,\sigma)$ with $\bE_n:=\bA\langle p_1 \rangle \ldots \langle p_n \rangle$ be a $G$-simple $\Pi$-tower over $(\bA,\,\sigma)$.  Assume that there is an algorithm which, for every $z\in G$, constructs a complete reduction for $(\bA,\, \Delta_z)$. Then there is an algorithm that, for every $f\in {G}_{\bA}^{\bE_n}$, constructs a complete reduction for $(\bE_n,\, \Delta_f)$. 
\end{prop}

In particular, combining Propositions~\ref{Prop:NestRExt} and~\ref{PROP:NestedPiCase}  (after reordering it accordingly to~\eqref{Equ:ReorderedSimpleExt}), we obtain the following proposition.

\begin{prop}\label{PROP:NestedRPiExt}
Let $(\bA, \, \sigma)$ be a computable difference ring with constant field $C$ and
 $(\bE,\,\sigma)$ be a simple $R\Pi$-tower over $(\bA, \, \sigma)$.
Assume that there is an algorithm which, for every $z\in \bA^*$ and $i\in \bZ^{+}$, constructs a complete reduction $\phi_z^{(i)}$  for $(\bA,\,\Delta_{z}^{(i)})$ as a $C$-linear map. Then there is an algorithm that, for every $f\in {(\bA^*)}_{\bA}^{\bE}$, constructs a complete reduction $\rho_f$ for $(\bE,\, \Delta_f)$. 
\end{prop}

Specializing the ground difference ring $(\bA,\,\sigma)$ to the rational difference field $(C(x),\,\sigma)$, we obtain a complete algorithm for simple $R\Pi$-towers over $(C(x),\,\sigma)$.

\begin{alg}\label{ALG:CRForSimpleRPi-Towers}
{\tt CRForSimpleR$\Pi$-Towers}$(\bE, g, f)$

 \noindent
{\tt Input:} A simple $R\Pi$-tower $(\bE,\,\sigma)$ with $\bE:=C(x)\langle y_1 \rangle \langle y_2 \rangle \ldots \langle y_{\ell_1} \rangle \langle p_1 \rangle \ldots \langle p_{\ell_2} \rangle$ over $(C(x),\,\sigma)$,  where the $y_i$ are R-monomials and $p_i$ are $\Pi$-monomials, $f\in {(C(x)^*)}_{C(x)}^{\bE}$ and $g \in\bE$; $C$ is factorizable.

\noindent
{\tt Output:} A $\Sigma$-pair of $g$ with respect to $\rho_{f}$ in Proposition~\ref{PROP:NestedRPiExt}.
\begin{enumerate}
	\item[]  {\tt if} $\ell_2=0$ {\tt then} 
\begin{itemize}
  \item [] {\tt return} {\tt CRForSimple$R$-Towers}$(\bE, g, f, 1)$ (*Algorithm~\ref{ALG:CRForSimpleR-Towers}*)
\end{itemize}
\item[] {\tt else}
 \begin{itemize}
   \item[] {\tt return}  {\tt CRForSimple$\Pi$-Extensions}$(\widetilde{\bE}\langle p_{\ell_2} \rangle, g, f)$(*Algorithm~\ref{ALG:Picase}*)\\
   \hspace*{1.2cm}\begin{minipage}{9cm} with $\widetilde{\bE}:=C(x)\langle y_1 \rangle \langle y_2 \rangle \ldots \langle y_{\ell_1} \rangle \langle p_1 \rangle \ldots \langle p_{\ell_2-1}\rangle$,
   	where the command \texttt{CRForGroundRing} is replaced \\ by \texttt{CRForSimple$R\Pi$-Towers} (*Algorithm~\ref{ALG:CRForSimpleRPi-Towers}*)
   \end{minipage}
 \end{itemize}
\item[]  {\tt end if}
\end{enumerate}
\end{alg}

\subsection{The $R\Pi\Sigma^*$-case}\label{SUBSECT:RPiSigma-case}

First we review how to construct a complete reduction in a $\Sigma^*$-tower as given in \cite{CGHS2025}. To this end, we need the following definition. 

\begin{define}\label{Def:effectiveBasis}
    Let $A$ be a ring with a subfield $C$, $\Theta$ be a $C$-basis of $A$ and $\theta  \in \Theta  $. Define $\theta^*$ to be the $C$-linear function on $A$ that maps $\theta$ to $1$ and all other elements to $0$. For a nonzero element $a$ of $A$, we say that $\theta$ is {\em effective} for $a$ if $\theta^*(a) \neq 0$. Furthermore,  $\Theta$ is called an {\em effective $C$-basis} if there are two algorithms available:
\begin{itemize}
\item [(i)]~given a nonzero element $a\in A$, find $\theta \in \Theta$ with $\theta^*(a)\neq0$; and 
\item [(ii)]~given $\theta \in \Theta  $ and $a \in A$, compute $\theta^*(a)$.
\end{itemize} 
\end{define}

With this additional notion (see \cite[Section~2.2]{DGLL2025} or \cite[Section~2]{CGHS2025} for further details), we can now present the algorithmic construction of complete reductions for a $\Sigma^*$-extension from \cite{CGHS2025}. While the results in \cite[Section~3.2]{CGHS2025} were originally stated for difference fields, the underlying proofs rely only on ring-closure properties and the constant set being a subfield. Consequently, these constructions extend to difference rings with a constant subfield, and the lemma follows from \cite[Proposition~4]{CGHS2025}.

\begin{lemma}\label{LM:SigmaExt}
Let $(\bA,\,\sigma)$ be a difference ring with constant field $C$ and  $(\bA[t],\, \sigma)$ be a $\Sigma^*$-extension of $\bA$. Assume that $\phi$ is a complete reduction  for $(\bA, \, \Delta)$.
 Let $\Theta$ be a $C$-basis of $\bA$ and $\theta \in \Theta$ be effective for $\phi(\Delta(t))$. Define
$$V_{\theta}=\directsum_{i\in \bN} (\im(\phi) \cap \ker(\theta^*)) t^i.$$
Then $\bA[t]=\im(\Delta)\oplus V_{\theta}$. In particular, the projection from $\bA[t]$ onto $V_{\theta}$ is a complete reduction $\psi_{\theta}$ for $(\bA[t],\, \Delta)$.
If $(\bA,\,\sigma)$ is computable, $\Theta$ is effective and one can compute a $\Sigma$-pair for any element in $\bA$ with respect to $\phi$, then one can compute $\Sigma$-pairs for all elements in $\bA[t]$ with respect to  $\psi_{\theta}$.
\end{lemma} 

In particular, combining the Algorithms \texttt{AuxiliaryReduction}, \texttt{EchelonBasis}, and \texttt{Redu-}\\\texttt{ctionForPolynomials} from \cite[Section~3.2]{CGHS2025}, which establish the algorithmic proof of Lemma~\ref{LM:SigmaExt}, yields an algorithm for computing $\Sigma^*$-pairs with respect to $\psi_{\theta}$. Below, we specify the input and output specifications of this procedure.

\begin{alg} {\tt CRFor$\Sigma^*$-Extensions}$(\bA[ t],g)$\label{ALG:Sigmacase}
	
	\noindent
	{\tt Input:} A $\Sigma^*$-monomial $t$ over $(\bA, \, \sigma)$ with constant field $C$ and $g\in \bA[ t]$; we assume that computing $\Sigma$-pairs with respect to a complete reduction $\phi$ for $(\bA, \, \Delta)$ is constructive for $(\bA, \, \Delta)$ and $\bA$ has an effective $C$-basis.
    
	\noindent
	{\tt Output:} a $\Sigma$-pair of $g$ with respect to $\psi_{\theta}$ given in Lemma~\ref{LM:SigmaExt}.
\end{alg}

\noindent By repeated applications of Lemma~\ref{LM:SigmaExt}, we get the following proposition.

\begin{prop}\label{PROP:SigmaExt}
Let $(\bA, \, \sigma)$ be a computable difference ring with constant field $C$ and $(\bE_n,\, \sigma)$ with $\bE_n:=\bA[t_1, \ldots,t_n]$ be a $\Sigma^*$-tower over $(\bA,\,\sigma)$. Assume that $(\bA,\,\sigma)$ has an effective $C$-basis and there exists an algorithm for constructing a complete reduction $\phi$ for $(\bA, \, \Delta)$. Then there is an algorithm for constructing a complete reduction $\psi$ for $(\bE_n, \, \Delta)$.
\end{prop}

We are ready to present the main result of this section which refines and improves~\cite[Theorem~2.26]{Schn2016}.

\begin{thm}\label{THM:SimpleRPiSigmaExt}
Let $(\bA, \, \sigma)$ be a computable difference ring with constant field $C$ and $(\bE,\,\sigma)$ be a simple $R\Pi\Sigma^*$-tower over $(\bA,\,\sigma)$. Assume that $\bA$ has an effective $C$-basis $\Theta_0$, and there is an algorithm that, for every $f \in \bA^*$ and $i \in \bZ^{+}$, constructs a complete reduction $\phi_f^{(i)}$ for $(\bA,\, \Delta_f^{(i)})$. Then there is an algorithm for constructing a complete reduction $\psi$ for $(\bE, \, \Delta)$.
\end{thm}
\begin{proof} 
By \cite[Lemma~4.10]{Schn2016}, $\bE$ can be reordered as~\eqref{Equ:ReorderedSimpleExt},
where the $y_i$ are $R$-monomials, the $p_i$ are $\Pi$-monomials, and the $s_i$ are $\Sigma^*$-monomials.
Set $\bE_1$ to be $\bA \langle y_1\rangle \ldots \langle y_{\ell_1} \rangle \langle  p_1\rangle \ldots \langle p_{\ell_2} \rangle$.
  Then there is an algorithm for constructing a complete reduction $\rho$ for $(\bE_1, \Delta)$ by Proposition~\ref{PROP:NestedRPiExt}.
Note that $\bE_1$ has an effective $C$-basis
$$\Theta=\{ \theta_0 y_1^{k_1}\cdots y_{\ell_1}^{k_{\ell_1}} p_1^{i_1}\cdots p_{\ell_2}^{i_{\ell_2}}~|~\theta_0 \in \Theta_0,~k_j \in[\ord(y_j)-1]_0~~\text{and}~i_j \in \bZ\}$$
and $(\bE,\,\sigma)$ is a $\Sigma^*$-tower over $(\bE_1,\sigma)$. So there is an algorithm for constructing a complete reduction $\psi$ for $(\bE, \, \Delta)$ by Proposition~\ref{PROP:SigmaExt}.
\end{proof}

Specializing the ground ring to the rational difference field, we obtain the following complete algorithmic result.

\begin{thm}\label{COR:RPiSigmaExt+Rational}
Let $(C(x), \, \sigma)$ be the rational difference field with $\sigma(x)=x+1$, where $C$ is factorizable. And let $(\bE,\, \sigma)$ be a simple $R\Pi\Sigma^*$-tower over $(C(x), \, \sigma^k)$ for some $k \in \bZ^{+}$.
Then there is an algorithm for constructing a complete reduction $\psi$ for $(\bE,\, \Delta)$.
\end{thm}
\begin{proof} 
 Note that $C(x)$ is constant-stable. So $\const(C(x), \sigma^{k})=C$. Further, there is an algorithm for constructing a complete reduction $\phi^{(ki)}_f$ for $(C(x), \Delta^{(ki)}_f)$ for every $f\in C(x)^*$ and 
 $i\in \bZ^{+}$ by Proposition~\ref{PROP:Generalrational+highorder}. Set $X = \left\{ x^i \mid i \in \bN \right\}$ and
$M_{x}$ to be the set consisting of all monic and irreducible polynomials in $C[x]$ with positive degrees. Using irreducible partial fraction decomposition, we obtain that
\begin{equation} \label{EQ:basis0}
  \Theta_0 =  X \cup \left\{ \frac{x^i}{q^j} \mid q \in M_x, 0 \le i < \deg_x(q), j \in \bZ^+ \right\}
\end{equation}
is an effective $C$-basis of $C(x)$.
Then the theorem follows from Theorem~\ref{THM:SimpleRPiSigmaExt}.
\end{proof}

The resulting algorithm of Theorem~\ref{COR:RPiSigmaExt+Rational} for $k=1$ can be summarized as follows.

\begin{alg}\label{ALG:CRForSimpleRPiSigma-Towers}
	{\tt CRForSimple$R\Pi\Sigma^*$-Towers}$(\bE, g)$
	
	\noindent
	{\tt Input:} A simple $R\Pi\Sigma^*$-tower $(\bE,\,\sigma)$ with $\bE:=\bA[s_1, \ldots, s_n]$ over $(\bA,\,\sigma)$,  where $(\bA,\,\sigma)$ is an R$\Pi$-tower over $(C(x),\,\sigma)$ and $C$ is factorizable (place the R-monomials first, followed by the $\Pi$-monomials), and the $s_i$ are $\Sigma^*$-monomials, and $g \in\bE$.
	
	\noindent
	{\tt Output:} A $\Sigma$-pair of $g$ with respect to $\psi$ as given in Theorem \ref{COR:RPiSigmaExt+Rational}.
	\begin{enumerate}
		\item[]  {\tt if} $n=0$ {\tt then} 
		\begin{itemize}
			\item [] {\tt return} {\tt CRForSimple$R\Pi$-Towers}$(\bE, g, 1)$ (*Algorithm~\ref{ALG:CRForSimpleRPi-Towers}*)
		\end{itemize}
		\item[] {\tt else}
		\begin{itemize}
			\item[] {\tt return}  {\tt CRFor$\Sigma^*$-Extensions}$(\bA[s_1,\ldots, s_{n-1}][s_n], g)$(*Algorithm~\ref{ALG:Sigmacase}*)\\  
			\hspace*{1.2cm}\begin{minipage}{9.5cm} where $\Sigma$-pairs with respect to a complete reduction $\phi$ for $(\bA[s_1,\ldots, s_{n-1}], \, \Delta)$ are computed by Algorithm {\tt CRForSimple$R\Pi\Sigma^*$-Towers} (*Algorithm~\ref{ALG:CRForSimpleRPiSigma-Towers}*)
			\end{minipage} 
		\end{itemize}
		\item[]  {\tt end if}
	\end{enumerate}
\end{alg}

\begin{example}\label{EX:RPiSigmaTowers}
We simplify the sum 
$$\sum_{k=0}^{n} (-1)^k \sum_{i=0}^{k} \binom{\nu}{i}$$
as follows. We start with $(C(x),\,\sigma)$, where $\sigma(x)=x+1$ and
$C=\bQ(\nu)$ is the rational function field in the variable $\nu$. Then we define the simple $R\Pi\Sigma^*$-tower $(\bE,\,\sigma)$ with 
$\bE:=\bQ(\nu)(x)\langle y \rangle \langle p \rangle \langle s \rangle$, and
$$\sigma(y)=-y,\quad \sigma(p)=\frac{\nu-x}{x+1}p \quad \text{and} \quad \sigma(s)=s+\frac{(\nu-x)p}{x+1}.$$
Here $y$ is an $R$-monomial representing $(-1)^k$, $p$ is a $\Pi$-monomial  denoting the binomial coefficient $\binom{\nu}{k}$, and $s$ is a $\Sigma^*$-monomial representing $\sum_{i=0}^{k} \binom{\nu}{i}$. 
So the given summand can be expressed as $y\,t$ in $\bE$. Let $\psi$ be the complete reduction for $(\bE, \, \Delta)$ as given in Theorem~\ref{COR:RPiSigmaExt+Rational}. Employing Algorithm~\ref{ALG:CRForSimpleRPiSigma-Towers} to $y\,t$ yields a $\Sigma$-pair $\left(-\frac{ys}{2}+\frac{y(\nu-x)p}{2\nu}, 0\right)$
of $y\,t$ with respect to $\psi$. Reinterpreting the result in terms of summation objects, we obtain by telescoping the simplification
$$\sum_{k=0}^{n} (-1)^k \sum_{i=0}^{k} \binom{\nu}{i}=\frac{1}{2}(-1)^n\sum_{i=0}^{n+1} \binom{\nu}{i}+\frac{(-1)^n(n-\nu+1)\binom{\nu}{n+1}}{2\nu}.$$
\end{example}

\subsection*{Implementation aspects}
In the following, we will compare it with the function {\tt RefinedTelescoping}~({\tt RT}) of the summation package {\tt Sigma} that implements the algorithms given in~\cite{Schn2007,Schn2008,Schneider:2015,Schn2016}. The function used to generate the experimental data is available at \cite{corrps}. The experiments were carried out in Mathematica~15.0.1 on a computer with Linux, CPU 2.40 GHz, Intel Core Ultra 7 265, 32GB memory. 

Each summand in the experimental data was the difference of an element in a simple $R\Pi\Sigma^*$-tower $(\bE,\,\sigma)$ over $(\bQ(x),\,\sigma)$. Consequently, both {\tt CR} and {\tt RT} are applicable and yield the same output, which is an indefinite sum of the input in the same tower. We generated three summands in the form $\Delta(f_i)$ for each $i$, where $f_i$ was a polynomial in some selected generators. Below are summaries of the average timings measured in seconds.

In the first suite, we let $\bE=\bQ(x)\langle p \rangle  \langle s \rangle$, where
$$\sigma(x)=x+1, \quad \sigma(p)=\frac{2(2x+1)}{x+1}p,\quad \text{and} \quad \sigma(s)=s+\frac{1}{x+1}.$$
Note that $p$ and $s$ represent $\binom{2n}{n}$ and $H_n$, respectively. 

Then we randomly generated a sparse polynomial $f_i \in \bQ(x)[p,s]$ of total degree $i$, with coefficients given by quotients of random polynomials in $\bQ[x]$ of degrees~$5$. 
Applying \texttt{RT} and \texttt{CR} to $\Delta(f_i)$ yields the timings listed below, highlighting the superior performance of \texttt{CR}.

\begin{center}
\begin{tabular}{|c|c|c|c|c|c|c|c|c|c|c|} \hline
		$i$ & 10 & 20 & 30 & 40 & 50 & 60 & 70 & 80 & 90 &100    \\ \hline
		{\tt RT} & 1.85 & 7.58 & 23.77 & 53.99 & 126.0 & 219.0 & 382.0 & 693.2 & 1038.7 & 1522.2 \\ \hline
		{\tt CR} & 0.13 & 0.42 & 1.06 & 1.84 & 3.4 & 5.3 & 8.0 & 12.3 & 17.0 & 24.2   \\ \hline
	\end{tabular}
\captionof{table}{Sparse polynomials in $\binom{2n}{n}$ and $H_n$} 
         \label{tab:Tele1}
\end{center}

In the second suite, we let $\bE=\bQ(x)\langle y \rangle \langle p \rangle \langle s \rangle$ with
$$\sigma(y)=-y,\quad \sigma(p)=\frac{2(2x+1)}{x+1}p \quad \text{and} \quad \sigma(s)=s-\frac{y}{x+1}.$$
Then $y$, $p$ and $s$ model $(-1)^n$, $\binom{2n}{n}$, $\sum_{i=1}^{n}\frac{(-1)^i}{i}$, respectively.
Here we generated two dense polynomials $u_i, v_i \in \bQ(x)[p,s]$ with degrees $i$ in $p$ and $s$. Their coefficients are quotients of random polynomials in $\bQ[x]$ with degrees $5$. The execution times for applying \texttt{RT} and \texttt{CR} to the difference of $f_i = u_i + y v_i$ are presented below. Again, \texttt{CR} demonstrates significantly faster runtimes across all test cases.

\begin{center}
	\begin{tabular}{|c|c|c|c|c|c|c|c|c|c|} \hline
		$i$ & 5 & 10 & 15 & 20 & 25 & 30 & 35 & 40    \\ \hline
		{\tt RT} & 3.66 & 16.19 & 49.61 & 128.20 & 301.60 & 667.21 & 1396.83 & 2821.66  \\ \hline
		{\tt CR} & 0.25 & 1.25 & 3.93 & 10.13 & 22.33 & 44.29 & 79.88 & 136.81    \\ \hline
	\end{tabular}
\captionof{table}{Dense polynomials in $(-1)^n$, $\binom{2n}{n}$ and $\sum_{i=1}^{n}\frac{(-1)^i}{i}$} 
         \label{tab:Tele2}
\end{center}

\section{Complete reductions in idempotent representations of $R\Pi\Sigma^*$-towers}\label{SECT:IdemotentRepresentation}

 This section consists of two parts. In Section \ref{SUBSECT:idem}, we introduce idempotent representation for difference rings and compute idempotent representations of $R\Pi\Sigma^*$-towers~(see Proposition~\ref{PROP:RPiSigmaIsIdempotent}). This then enables us, in Section~\ref{SUBSECT:crforgeneralcase}, to develop an algorithm for constructing complete reductions for $R\Pi\Sigma^*$-towers (see Theorem~\ref{THM:CR+RPiSigma}).

\subsection{Constructing idempotent representations}\label{SUBSECT:idem}

We begin with a class of special difference rings that is closely related to \cite[Definition 2.2]{OW2015}; for algorithmic aspects, we refer to~\cite{AS:2021}.
\begin{define}\label{DEF:idempotent}
Let $(\bA, \, \sigma)$ be a difference ring and $\lambda \in \bZ^{+}$. If there exist idempotent, pairwise orthogonal elements $e_0,e_1,\ldots,e_{\lambda-1}$ of $\bA$ (i.e., $e_u^2=e_u$ and $e_ue_v=0$ if $u\neq v$)
such that $\sigma(e_u)=e_{u+1 \! \bmod \! \lambda}$, $e_0+e_1+\cdots+e_{\lambda-1}=1$, then $(\bA, \, \sigma)$ is called a {\em weak idempotent difference ring  
with respect to $\{e_i~|~i\in [\lambda-1]_0\}$}. 
If this is the case,  we have the following decomposition
\begin{equation}\label{EQ:idempotentrepresentation}
  \bA=e_0\bA\oplus e_1\bA \oplus \cdots \oplus e_{\lambda-1}\bA
\end{equation}
as a direct sum of ideals by \cite[Chapter~1, Exercise 1.7]{Lam2001}, which is called a {\em weak idempotent representation} of $\bA$. 
In addition,
if for every $i \in [\lambda-1]_0$, $e_i\bA$ is an integral domain, we call $(\bA, \, \sigma)$ an {\em idempotent difference ring with respect to $\{e_i~|~i\in [\lambda-1]_0\}$} and \eqref{EQ:idempotentrepresentation} is called an {\em idempotent representation} of $\bA$.
\end{define}

Picard–Vessiot extensions for linear difference equations with polynomial coefficients in $(C(x), \sigma)$ form idempotent difference rings (see~\cite[Corollary~1.16]{vdSi1997} and \cite[Lemma~6.8]{HS2008}). In contrast, $R\Pi\Sigma^*$-towers are generally not representable as Picard–Vessiot extensions whenever nested $\Pi$-monomials are present (since they do not satisfy such linear difference equations).

\begin{remark}\label{RE:idemp+unit}
Let $(\bA,\, \sigma)$ be a weak idempotent difference ring with respect to $\{e_i~|~i\in [\lambda-1]_0\}$. For every $s \in [\lambda-1]_0$,
\begin{itemize}
  \item [(i)] $(e_s\bA,\, \sigma^{\lambda})$  is a difference ring with multiplicative identity $e_s$.
  \item [(ii)] The map $\sigma$ induces a difference ring isomorphism between $(e_s\bA, \sigma^{\lambda})$ and $(e_{s+1}\bA, \sigma^{\lambda})$.
\end{itemize}

\end{remark}
 
The following lemma reveals the uniqueness of the idempotent representation of an idempotent difference ring.
\begin{lemma}
Let $(\bA,\, \sigma)$ be an idempotent difference ring with respect to $\{e_i~|~i\in[\lambda-1]_0\}$ as given in \eqref{EQ:idempotentrepresentation}. Assume that $(\bA,\, \sigma)$ is also an idempotent difference ring with respect to  $\{f_i~|~i \in [\mu-1]_0\}$.
Then $\lambda=\mu$ and there is $j\in [\lambda-1]_0$  such that $e_i = f_{(i+j)\! \bmod \! \lambda}$ for all $i$.
\end{lemma}

\begin{proof}
We say that a nonzero idempotent element of $\bA$ is primitive if it can not be written as a sum of two  non-zero orthogonal idempotent elements of $\bA$.  Suppose that $e_0$ as given in \eqref{EQ:idempotentrepresentation} is not primitive. Then there are  two nonzero orthogonal idempotent elements $r_1, r_2 \in \bA$ such that $e_0=r_1+r_2$. It follows that $e_0\bA \subset r_1\bA+r_2\bA$. For arbitrary $x,y \in \bA$, $r_1x+r_2y=e_0(r_1x+r_2y)$ because  $r_1r_2=0$ and $r_i^2=r_i$ for $i=1,2$. So $e_0\bA=r_1\bA \oplus r_2 \bA$, which implies that $r_1 \in e_0\bA$. 
In particular, $r_1$ is an idempotent of $e_0\bA$. 
Since $e_0\bA$ is an integral domain, $e_0$ is the only nontrivial idempotent. In other words, $r_1=e_0$.
 Hence, $r_2=0$, a contradiction. 
 Similarly, the $e_i$, $f_i$ are primitive. 
 Note that $1=\sum_{i=0}^{\lambda-1}e_i=\sum_{i=0}^{\mu-1}f_i$. 
 So $\lambda=\mu$ and there is a permutation $\tau$ of $\{0, \ldots, \lambda-1\}$ such that $e_i = f_{\tau(i)}$ for all $i$ by \cite[Proposition~22.1]{Lam2001}. Thus, the conclusion holds by setting $j$ to be  $\tau(0)$.  
 \end{proof}

\noindent An equivalent characterization of a weak idempotent difference ring is described below.

\begin{lemma}\label{LM:idemp+equivdef}
Let $(\bA,\, \sigma)$ be a difference ring. Then it is a weak idempotent difference ring
if and only if there is an idempotent element $e_0$ of $\bA$ such that 
\begin{equation}\label{EQ:idempdecomp2}
  \bA=e_0\bA \oplus \sigma(e_0)\bA \oplus \cdots \oplus \sigma^{\lambda-1}(e_0) \bA
\end{equation}
as a direct sum of ideals.
\end{lemma}
\begin{proof} 
Let $\eqref{EQ:idempotentrepresentation}$ be a weak idempotent representation of $\bA$. Since $\sigma(e_u)=e_{u+1 \! \bmod \! \lambda}$ for every $u \in [\lambda-1]_0$, we have $e_i=\sigma^{i}(e_0)$. So \eqref{EQ:idempdecomp2} holds. Conversely,  note that $e_0$ is idempotent, so are the $\sigma^i(e_0)$. For $i, j\in [\lambda-1]_0$ with $i \neq j$, $\sigma^i(e_0) \sigma^j(e_0)\in \sigma^{i}(e_0)\bA \cap \sigma^{j}(e_0)\bA=\{0\}$. Since $1\in \bA$, there are $a_0, a_1, \ldots, a_{\lambda-1}\in \bA$ such that $1=\sum_{i=0}^{\lambda-1}\sigma^{i}(e_0)a_i$. Multiplying both sides of the above equation by $\sigma^{i}(e_0)$ yields $\sigma^{i}(e_0)=\sigma^{i}(e_0)a_i$, which implies that $1=\sum_{i=0}^{\lambda-1}\sigma^{i}(e_0)$. Then applying $\Delta$  to the above equation gives 
$0=\sigma^{\lambda}(e_0)-e_0$. Hence, $(\bA,\,\sigma)$ is a weak idempotent difference ring with respect to $\{\sigma^i(e_0)~|~i\in[\lambda-1]_0\}$ by Definition~\ref{DEF:idempotent}. 
\end{proof}

We generalize \cite[Theorem~4.3]{Schn2017} from basic $RPS$-towers\footnote{In this restricted version, the $t_i$ in Lemma~\ref{LM:anRmonomial+idem} are simple $PS$-monomials with the following extra restriction: If $t_i$ is a $P$-monomial, then $\sigma(t_i)/t_i$ is free of the $R$-monomial $y$.}~(See \cite[Definition 2.7]{Schn2017}) to $RPS$-towers. It will be used to construct idempotent representations and complete reductions for  $R\Pi\Sigma^*$-towers in the sequel.
 
\begin{lemma}\label{LM:anRmonomial+idem}
Let $(\bA,\, \sigma)$ be a difference ring with constant field $C$. Assume further that 
$(\bE,\,\sigma)$ with $\bE=\bA\langle y \rangle \langle t_1 \rangle \langle t_2 \rangle \ldots \langle t_k \rangle$ is an $APS$-tower over $(\bA,\,\sigma)$, where $y$ is an $R$-monomial over $\bA$ with $\alpha:=\sigma(y)/y \in C$ of order $\lambda$. 
Set $\widetilde{\bE}=\bA \langle t_1 \rangle \langle t_2 \rangle \ldots \langle t_k \rangle$. Then the following assertions hold.
\begin{itemize}
  \item [(i)] One can compute 
  $e_0,e_1,\ldots,e_{\lambda-1}$ of $C\langle y \rangle$
 such that $\bE$ is a weak idempotent difference ring with respect to $\{e_i~|~i\in[\lambda-1]_0\}$.
 \item [(ii)]  
  For every $s \in [\lambda-1]_0$ and $f\in\bE$,  $e_sf=e_s\cdot(f|_{y \rightarrow \alpha^{\lambda-1-s}})$.
  In particular,
  $e_s\bE=e_s\widetilde{\bE}$.
   \item[(iii)] For every $s\in [\lambda-1]_0$, $(\widetilde{\bE},\, \sigma_s)$  is an $APS$-tower over $(\bA,\, \sigma^{\lambda})$, where $\sigma_s$ is defined  by $\sigma_s(f)=\sigma^{\lambda}(f)|_{y \rightarrow \alpha^{\lambda-1-s}}$ for all $f \in \widetilde{\bE}$. 
        In particular, for every $i\in [k]$, $t_i$ is an A-monomial over $(\bA\langle t_1\rangle\langle t_2\rangle \ldots \langle t_{i-1}\rangle,\, \sigma_s)$  of order $m_i$ if $t_i$ is an A-monomial over $(\bA \langle y\rangle\langle t_1\rangle \ldots \langle t_{i-1}\rangle,\, \sigma)$ of order $m_i$, and $t_i$ is a P/S-monomial over $(\bA\langle t_1\rangle \ldots \langle t_{i-1}\rangle,\, \sigma_s)$ if $t_i$ is a P/S-monomial over $(\bA \langle y\rangle\langle t_1\rangle \ldots \langle t_{i-1}\rangle,\, \sigma)$.
       
\item[(iv)] There is a difference ring isomorphism $\tau$ from $(e_s\bE,\, \sigma^{\lambda})$ onto $(\widetilde{\bE},\, \sigma_s)$, mapping $e_sf$  to $f|_{y \rightarrow \alpha^{\lambda-1-s}}$ for every $f\in\bE$, where $\sigma_s$ is given by (iii). Moreover, $\tau^{-1}(\tilde{f})=e_0\tilde{f}$ for every $\tilde{f} \in \widetilde{\bE}$.
    
     \item[(v)]  Assume further that $(\bE,\, \sigma)$ is a $PS$-tower over $(\bA\langle y \rangle,\, \sigma)$ and $\bA$ is an integral domain. Then  $(\bE,\,\sigma)$ is an idempotent difference ring  with respect to $\{e_i~|~i\in[\lambda-1]_0\}$, where the $e_i$ are given by (i).
    
  \item [(vi)]  $(\bE,\, \sigma)$ is an $R\Pi\Sigma^*$-tower over $(\bA,\, \sigma)$ if and only if there is $k\in [\lambda-1]_0$ such that $(\tilde \bE,\,\sigma_k)$ is an $R\Pi\Sigma^*$-tower over $(\bA,\, \sigma^\lambda)$, where $\sigma_k$ is given by (iii). In this case, $\const(\tilde \bE,\, \sigma_k)=  C$ if $(\bA,\, \sigma)$ is constant-stable.
\end{itemize} 
\end{lemma}
\begin{proof}  
Set $\alpha:=\sigma(y)/y$. 
Then $\alpha$ is a $\lambda$th primitive root of unity
by Theorem~\ref{THM:testingRPiSigma-monomials}~(iii). For every $s \in [\lambda-1]_0$, let
$$\tilde{e}_s(y):= \prod_{\substack{j=0 \\ j\neq \lambda-1-s}}^{\lambda-1} (y-\alpha^j).$$
Since $C$ is a field, $\tilde{e}_s(\alpha^{\lambda-1-s}) \neq 0$, which allows us to define
\begin{equation}\label{EQ:IdempEle}
e_s=e_s(y):=\frac{\tilde{e}_s(y)}{\tilde{e}_s(\alpha^{\lambda-1-s})} \in C\langle y \rangle \backslash \{0\}.
\end{equation}

(i)~$\alpha \in C$ implies $\sigma(e_s)=e_s(\alpha y)$. Moreover, $\sigma(e_s)=e_{s+1 \! \bmod \! \lambda}$ by a straightforward computation.
It remains to show that the $e_i$ are pairwise orthogonal, idempotent, and sum to 1. Since $y^{\lambda}-1$ divides $e_i(y)e_j(y)$ for $i, j \in[\lambda-1]_0$ with $i\neq j$, the $e_i$ are pairwise orthogonal. On the other hand, write $e_s(y)^2-e_s(y)$ as $\sum_{i=0}^{\lambda-1}c_iy^i$ for $c_i \in C$. Let $C[z]$ be the ring of polynomials, where $z$ is an indeterminate over $C$. Consider $p_s(z)=\sum_{i=0}^{\lambda-1}c_iz^i \in C[z]$.  Note that
\begin{equation}\label{EQ:idempotent}
e_s(\alpha^i)  =
\left\{
\begin{array}{ll}
1, &  \text{if $i=\lambda-1-s$}, \\ 
0,  & \text{if $i \neq \lambda-1-s$}.
\end{array} \right.
\end{equation}
So $p_s(\alpha^i)=0$ for every $s\in [\lambda-1]_0$.
Since $C$ is a field, $p_s$ has at most $\lambda-1$ zeros, which implies that $p_s$ must be zero.
 Hence, $e_s^2=e_s$. Similarly, we can show that $\sum_{i=0}^{\lambda-1}e_i=1$. 
 
(ii) First we show that the evaluation $\cdot~|_{y \rightarrow \alpha^{\lambda-1-s}}$ on $\bE$ is well-defined. Let $f, g \in \bE$ be such that $f-g=u(y^{\lambda}-1)$ for some $u \in \widetilde{\bE}\langle y \rangle$. Then $f|_{y \rightarrow \alpha^{\lambda-1-s}}=g |_{y \rightarrow \alpha^{\lambda-1-s}}$ by $\alpha^{\lambda}=1$.
 Note that $f$ can be written as $(y-\alpha^{\lambda-1-s})\tilde f
+ f|_{y\to\alpha^{\lambda-1-s}}$
for some $\tilde f\in\widetilde{\bE}[y]$.
Moreover, $e_s(y-\alpha^{\lambda-1-s})$ is a multiple of
$y^\lambda-1$, and hence vanishes. Therefore,
$e_sf= e_s(f|_{y\to\alpha^{\lambda-1-s}}).$

(iii) Set $\bE_0=\bA\langle y\rangle$, $\bE_i=\bE_{i-1}\langle t_i\rangle$, $\widetilde{\bE}_0=\bA$ and $\widetilde{\bE}_i=\widetilde{\bE}_{i-1}\langle t_i\rangle$ for $i \in [k]$. We proceed the assertion by induction on $i$. For the base case $i=0$, $\sigma_s(f)=\sigma^{\lambda}(f)$ for all $f\in \bA$. Assume that $(\widetilde{\bE}_{i-1},\,\sigma_s)$ is an $APS$-tower over $(\bA, \sigma^{\lambda})$. 
Consider $(\widetilde{\bE}_{i-1}\langle t_i\rangle, \,\sigma_s)$.
Let $\sigma(t_i)=a_it_i$ for some  $a_i \in \bE_{i-1}^*$ if $t_i$ is an $AP$-monomial. 
By Lemma~\ref{LM:RPiSigma+term}~(ii), we have  $\sigma^{\lambda}(t_i)= a_i^{\sigma, \lambda}t_i$, which implies that
 $\sigma_s(t_i)=\tilde{a}_it_i$ for  $\tilde{a}_i=a_i^{\sigma, \lambda}~|_{y\rightarrow \alpha^{\lambda-1-s}}$. 
Since $a_i \in \bE_{i-1}^*$ and $\sigma$ is an automorphism, $\tilde{a}_i \in \widetilde{\bE}_{i-1}^*$. 
So $t_i$ is an $AP$-monomial over $(\widetilde{\bE}_{i-1},\, \sigma_s)$.
Furthermore, if $a_i$ is a $\lambda_i$th root of unity in $\bE_{i-1}$, 
then $\tilde{a}_i$ is also a $\lambda_i$th root of unity in $\widetilde{\bE}_{i-1}$.
That is, $t_i$ is an $A$-extension of $(\widetilde{\bE}_{i-1},\, \sigma_s)$ of order $\lambda_i$ if $t_i$ is an $A$-extension of $(\bE_{i-1},\, \sigma)$ of order $\lambda_i$. 
On the other hand,  $\sigma(t_i)=t_i+b_i$ for some $b_i \in \bE_{i-1}$ if $t_i$ is a $\Sigma^*$-monomial. 
It follows from Lemma~\ref{LM:RPiSigma+term}~(i) that $\sigma_s(t_i)=t_i+\tilde{b}_i$ for some $\tilde{b}_i \in  \bE_{i-1}$. 
So $t_i$ is an $S$-monomial over $\widetilde{\bE}_{i-1}$. 
Consequently, $(\widetilde{\bE}_i,\, \sigma_s)$ is an $APS$-tower over $(\bA,\, \sigma^{\lambda})$.

(iv)~By Remark \ref{RE:idemp+unit}~(i), $(e_sE, \, \sigma^{\lambda})$ is a difference ring. 
Next, we verify that $\tau$ is well-defined. Let $f, g \in\bE$ be such that $e_sf=e_sg$. Then $e_s(f-g)=u(y^{\lambda}-1)$ for some $u \in \widetilde{\bE}\langle y \rangle$. By \eqref{EQ:IdempEle}, $y-\alpha^{\lambda-1-s}$ divides $f-g$. So $f|_{y \rightarrow \alpha^{\lambda-1-s}}=g |_{y \rightarrow \alpha^{\lambda-1-s}}$, which implies that $\tau$ is well defined. Furthermore, $\tau$ is a ring homomorphism from $e_s\bE$ to $\widetilde{\bE}$. Now consider the ring homomorphism

\[
 \begin{array}{cccc}
\rho: & \widetilde{\bE} & \rightarrow & e_s\bE \\
       &   \tilde{f}                & \mapsto     & e_s\tilde{f}.
\end{array}
\]
Then, for every $f\in\bE$, $\rho(\tau(e_sf))=e_s(f|_{y \rightarrow \alpha^{\lambda-1-s}})=e_sf$ by (ii). 
On the other hand, $\tau(\rho(\tilde{f}))=\tilde{f}$ for every $\tilde{f} \in \widetilde{\bE}$. 
So $\rho \tau=\tau \rho=\bf{1}$. Thus, $\tau$ is a ring isomorphism and $\tau^{-1}=\rho$. 
By (ii) again, every element in $e_s\bE$ can be expressed as $e_s \tilde{f}$ for some $\tilde{f} \in \widetilde{\bE}$. 
Then
\[\tau(\sigma^{\lambda}(e_s\tilde f))= \tau(e_s\sigma^{\lambda}(\tilde{f}))=\sigma^{\lambda}(\tilde f)|_{y \rightarrow \alpha^{\lambda-1-s}}=\sigma_s(\tilde f)=\sigma_s(\tau(e_s\tilde f)).\]
So $\tau \circ \sigma^{\lambda}=\sigma_s \circ \tau$. Consequently, $\tau$ is a difference ring isomorphism from  $(e_s\bE,\, \sigma^{\lambda})$ onto $(\widetilde{\bE},\, \sigma_s)$.

(v) By (iii) and the assumption as given in (v), $(\widetilde{\bE},\, \sigma_s)$ is a $PS$-tower over $(\bA,\, \sigma^{\lambda})$. So $\widetilde{\bE}$ is an integral domain. Together with (iv), $e_s\bE$ is also an integral domain. Therefore, by (i), $\bE$ is an idempotent difference ring with respect to $\{e_i~|~i \in [\lambda-1]_0\}$.

(vi) By (iii), $(\widetilde{\bE},\, \sigma_s)$ is an $APS$-tower over $(\bA,\, \sigma^{\lambda})$ for every $s\in [\lambda-1]_0$. By (iv), it suffices to show that $\const(\bE,\,\sigma)=\const(\bA,\,\sigma)$ if and only if there is $k\in[\lambda-1]_0$ such that $\const(e_k \bE,\, \sigma^{\lambda})=\const(e_k\bA, \, \sigma^{\lambda})$.  
Assume that $\const(\bE,\, \sigma)=\const(\bA,\,\sigma)$ and fix any $k\in[\lambda-1]_0$. Let $f \in \bE$ be such that $e_kf \in \const(e_k \bE,\, \sigma^{\lambda})$. Set $h$ to be $\sum_{j=0}^{\lambda-1} \sigma^j(e_kf)$. Then $\sigma(h)-h=\sigma^{\lambda}(e_kf)-e_kf=0$ by $e_kf \in \const(e_k \bE,\, \sigma^{\lambda})$. So $h \in \const(\bE,\, \sigma)$, which implies that $h \in \const(\bA, \, \sigma)$. 
Since $\sigma(e_k)=e_{k+1 \! \bmod \! \lambda}$, $e_{k}e_{k+1}=0$ and $e_k^2=e_k$, we have $e_kf=e_kh \in e_k\bA$. So $e_k f \in \const(e_k\bA, \,\sigma^{\lambda})$. 
Thus, $\const(e_k \bE,\, \sigma^{\lambda})=\const(e_k\bA,\, \sigma^{\lambda})$ by $\const(e_k\bA,\, \sigma^{\lambda}) \subset \const(e_k \bE,\, \sigma^{\lambda})$. 
Conversely, suppose that we can choose $k\in[\lambda-1]_0$ such that $\const(e_k \bE,\, \sigma^{\lambda})=\const(e_k\bA,\, \sigma^{\lambda})$.
Let $f \in \const(\bE,\,\sigma)$ and $s \in [\lambda-1]_0$.
Then $e_sf \in \const(e_s\bE,\,\sigma^{\lambda})$. 
Since $\const(e_k \bE,\, \sigma^{\lambda})=\const(e_k\bA,\, \sigma^{\lambda})$, we have $\const(e_s \bE,\, \sigma^{\lambda})=\const(e_s\bA,\, \sigma^{\lambda})$ by Remark \ref{RE:idemp+unit}. 
So $e_s f \in \const(e_s\bA, \,\sigma^{\lambda})$. In particular, $e_sf \in \bA\langle y \rangle$. This leads to
$$f=1\cdot f=e_0f+e_1f+\cdots+e_{\lambda-1}f \in \bA\langle y \rangle.$$
Hence, $f \in \const(\bA\langle y \rangle,\, \sigma)$. Consequently, $f \in \const(\bA,\, \sigma)$ because $y$ is an $R$-monomial. In this case, if $(\bA,\,\sigma)$ is constant-stable, then $\const(\widetilde{\bE},\, \sigma_s)= \const(\bA,\, \sigma^{\lambda})=C$.
\end{proof}

A direct benefit of the above lemma is that, one can determine whether an $APS$-tower
is an $R\Pi\Sigma^*$-tower. This generalizes the method given in \cite{Schn2016}, which applies only to the special case of simple $APS$-towers.

\begin{prop}\label{PROP:checkingRPiSigmaTower}
Let $(\bA,\,\sigma)$ be a computable difference ring which is integral and constant-stable. Assume that, for every $k\in\bZ^{+}$, one can determine whether a $PS$-tower over $(\bA,\, \sigma^k)$ is a $\Pi\Sigma^*$-tower.   Then one can determine whether an $APS$-tower over $(\bA,\,\sigma)$ is an $R\Pi\Sigma^*$-tower.
\end{prop}
\begin{proof} 
Let $(\bE,\,\sigma)$ with $\bE=\bA\langle t_1 \rangle \langle t_2 \rangle \ldots \langle t_e \rangle$ be an $APS$-tower over $(\bA, \, \sigma)$. We proceed the conclusion by induction on the number $\ell$ of $A$-monomials in $\bE$. The base case $\ell=0$ follows from the assumption of the proposition. Assume that $\ell >0$ and the conclusion holds for $\ell-1$. Let $j$ be the minimal such that $t_j$ is an $A$-monomial. Then $(\tilde{\bA},\,\sigma)$ with $\tilde{\bA}=\bA \langle t_1 \rangle \langle t_2 \rangle \ldots \langle t_{j-1} \rangle$ is a $PS$-tower over  $(\bA, \, \sigma)$. 
By the assumption of the proposition, one can determine whether $(\tilde{\bA},\, \sigma)$ is a $\Pi\Sigma^*$-tower over $(\bA,\,\sigma)$. If the answer is negative, $(\bE,\,\sigma)$ is not an $R\Pi\Sigma^*$-tower over $(\bA,\,\sigma)$. Otherwise, $(\tilde{\bA},\, \sigma)$ is constant-stable by Lemma~\ref{Lemma:ConstantStable}.  Set $\lambda=\ord(t_j)$. Then $t_j$ is an $R$-monomial over $(\tilde{\bA},\, \sigma)$ if and only if $\sigma(t_j)/t_j$ is a primitive $\lambda$th root of unity in $C$ by Proposition~\ref{PROP:R-monomial+primitive}. If the answer is negative, $(\bE,\,\sigma)$ is not an $R\Pi\Sigma^*$-tower over $(\bA,\,\sigma)$. Otherwise, using Lemma~\ref{LM:anRmonomial+idem} (i), we obtain a weak idempotent representation of $\bE$:
  $$\bE=e_0\bE \oplus e_1\bE \oplus \cdots \oplus e_{\lambda-1}\bE.$$
  Set $\widetilde{\bE}:=\tilde{\bA} \langle t_{j+1} \rangle \ldots \langle t_e \rangle$. By Lemma~\ref{LM:anRmonomial+idem} (vi) and Remark \ref{RE:idemp+unit}~(ii), $(\bE,\,\sigma)$ is an $R\Pi\Sigma^*$-tower over $(\tilde{\bA},\, \sigma)$ if and only if $(\widetilde{\bE},\, \sigma_0)$ with $\sigma_0(\tilde{f})=\sigma^{\lambda}(\tilde{f})|_{t_j\rightarrow \alpha^{\lambda-1}}$ for $\tilde{f} \in \widetilde{\bE}$ is an $R\Pi\Sigma^*$-tower over $(\tilde{\bA},\, \sigma^{\lambda})$. Recall that  $(\tilde{\bA},\, \sigma^{\lambda})$ is constant-stable by Lemma~\ref{Lemma:ConstantStable}. Thus $(\tilde{\bA}, \, \sigma^{\lambda})$ is a $\Pi\Sigma^*$-tower over $(\bA,\, \sigma^{\lambda})$, which implies that $(\bE, \, \sigma)$ is an $R\Pi\Sigma^*$-tower over $(\bA,\, \sigma)$ if and only if $(\widetilde{\bE}, \, \sigma_0)$ is an $R\Pi\Sigma^*$-tower over $(\bA,\, \sigma^{\lambda})$.\\ 
  Again we can determine whether a $PS$-tower over $(\bA,\, \sigma^{\lambda})$ is a $\Pi\Sigma^*$-tower due to the assumption of the proposition. Since in the $APS$-tower
  $(\widetilde{\bE}, \sigma_0)$ over $(\bA,\, \sigma^{\lambda})$ there are only $\ell-1$ $A$-monomials left, we can apply the induction hypothesis to determine whether it is an $R\Pi\Sigma^*$-tower. If so, $(\bE,\, \sigma)$ is also an $R\Pi\Sigma^*$-tower over $(\bA,\,\sigma)$.
\end{proof}

Recall that any $\Pi\Sigma^*$-field $(\bF,\, \sigma)$ over $C$ is a constant-stable integral domain. Provided the constant field $C$ is orbit-computable, the methods in~\cite{Karr1981,Schn2016} enable us to decide algorithmically whether a $PS$-tower over $(\bF, \, \sigma^k)$ is a $\Pi\Sigma^*$-tower for all $k \in \mathbb{Z}^+$\footnote{For instance, testing whether an $S$-monomial is a $\Sigma^*$-monomial relies on Theorem~\ref{THM:testingRPiSigma-monomials}~(i) combined with Theorem~\ref{Thm:PiSiOverSimpleRPiSiDESolver}.}. Consequently, we obtain the following algorithmic special case of Proposition~\ref{PROP:checkingRPiSigmaTower}.

\begin{cor}\label{Cor:CheckInPiSigma}
Let $(\bF,\,\sigma)$ be a $\Pi\Sigma^*$-field over an orbit-computable constant field $C$. Then one can determine whether an $APS$-tower over $(\bF,\,\sigma)$ is an $R\Pi\Sigma^*$-tower algorithmically.
\end{cor}

We illustrate the above decision procedure with a concrete example for the rational difference field.

\begin{example}\label{EX:CheckRmonomials}
Let $(\bF,\,\sigma)$ be the rational difference field with $\bF=C(x)$ and $\sigma(x)=x+1$, where $C:=\bQ(\sqrt{-1})$. 
By Remark~\ref{Remark:AlgRatOribt}, $C$ is orbit-computable. Take the $R$-monomial $y_1$ over $\bF$ of order $2$ with $\sigma(y_1)=-y_1$. Let $e_0 :=\frac{1-y_1}{2}$ and $e_1 :=\frac{1+y_1}{2}$. Then $e_0, e_1$ are idempotent, pairwise orthogonal and $e_0+e_1=1$. Moreover, $e_0\bF \oplus e_1\bF$ is an idempotent representation of $\bF \langle y_1 \rangle$ by Lemma~\ref{LM:anRmonomial+idem}~(i). Set $w:=e_0+\sqrt{-1}e_1\in\bF\langle y_1\rangle$. Then $w^4=1$. Let $y_2$ be an A-monomial of order $4$ with $\sigma(y_2)=wy_2$. Clearly, $y_2$ is not simple. 
Thus a check if it is an $R$-monomial cannot be carried out by the algorithms given in~\cite{Schn2016}. 
But it is possible by our new technique summarized in Corollary~\ref{Cor:CheckInPiSigma}. 
Consider $(\bF\langle y_2 \rangle,\, \sigma_0)$, where $\sigma_0(\tilde{f})=\sigma^2(\tilde{f})|_{y_1\rightarrow -1}$ for every $\tilde{f} \in \bF\langle y_2 \rangle$. By a direct computation, $\sigma_0(y_2)=\sigma^2(y_2)|_{y_1\rightarrow -1}=\sqrt{-1}y_2$. Since $\sqrt{-1}$ is a $4$th primitive root of unity, $y_2$ is an $R$-monomial over $(\bF,\,\sigma_0)$ and thus an $R$-extension over $(\bF\langle y_1\rangle,\sigma)$.
\end{example}

The following proposition constructs an idempotent representation of an $R\Pi\Sigma^*$-tower, in which every component is a $\Pi\Sigma^*$-tower. 
\begin{prop}\label{PROP:RPiSigmaIsIdempotent}
Suppose that the difference ring $(\bA,\, \sigma)$ with constant field $C$ is integral and constant-stable. Let $(\bE,\,\sigma)$ with
$\bE:=\bA\langle t_1 \rangle \langle t_2 \rangle \ldots \langle t_k \rangle$ be an $R\Pi\Sigma^*$-tower over $(\bA,\,\sigma)$, and set $$\bI:=\{ i \in [k]~|~t_i~\text{is an R-monomial}\}.$$ Assume further that $\bI$ is not empty. Set $\lambda=\prod_{i\in \bI} \ord(t_i)$. 
Then the following assertions hold.

\begin{itemize}
  \item [(i)]  One can compute $e_0,\ldots,e_{\lambda-1} \in \bE$ such that $(\bE,\,\sigma)$ is an idempotent difference ring with respect to $\{e_i \mid i\in[\lambda-1]_0\}$.
  \item [(ii)] Set $\bH :=\bA\langle t_i~|~i\in [k] \backslash \bI \rangle$. One can construct $\{\beta_i\in C \mid i\in \bI \}$ such that $(\bH,\, \sigma_0)$, where,  for every $h \in \bH$, $\sigma_0(h)=\sigma^{\lambda}(h)|_{\{t_i \rightarrow \beta_i\mid i\in\bI\}}$,  is a $\Pi\Sigma^*$-tower over $(\bA,\, \sigma^{\lambda})$. Moreover, there is a difference ring isomorphism
      \[
\begin{array}{cccc}
\delta: & (e_0\bE,\, \sigma^\lambda) & \rightarrow & (\bH, \, \sigma_0) \\
      & e_0f& \mapsto     & f|_{\{t_i \rightarrow \beta_i \mid i\in\bI\}}
\end{array}
\]
for every $f \in \bE$. Furthermore, $\delta^{-1}(h)=e_0h$ for every $h \in \bH$.     
\end{itemize}

\end{prop}
\begin{proof} 
By Proposition~\ref{PROP:orderRPiSigma}, $\bE$ can be reordered as~\eqref{EQ:reoder+RPiSigma}
where the $y_i$ are $R$-monomials, the $p_i$ are $\Pi$-monomials, and the $s_i$ are $\Sigma^*$-monomials. Since $\bI\neq\emptyset$, it follows that $\ell>0$.
Then $\lambda$ and $\bH$ as given in the proposition are equal to $\prod_{i=1}^{\ell} \ord(y_i)$ and $\bA\langle p_1 \rangle \ldots \langle p_{m} \rangle  \langle s_1 \rangle \ldots \langle s_n \rangle$, respectively.
We proceed the conclusion by induction on $\ell$. 
First we have $\sigma(y_1)/y_1 \in C$ by Proposition~\ref{PROP:R-monomial+primitive}. Set $\ord(y_{1})$  to be $\lambda_1$. 
According to Lemma~\ref{LM:anRmonomial+idem} (i), there are 
$u_0, u_1, \ldots, u_{\lambda_1-1} \in C \langle y \rangle$ such that $(\bE,\,\sigma)$ is a weak idempotent difference ring with respect to 
$\{u_i\mid i\in[\lambda_1-1]_0\}$.

Let $\widetilde{\bE}:=\bA\langle y_2 \rangle \ldots \langle y_{\ell} \rangle \langle p_1 \rangle \ldots \langle p_{m} \rangle  \langle s_1 \rangle \ldots \langle s_n \rangle$ and $\beta_1:=(\sigma(y_1)/y_1)^{\lambda_1-1}$.
By application of Lemma~\ref{LM:anRmonomial+idem}~(iv), there is a difference ring isomorphism
 \[\begin{array}{cccc}
\tau: & (u_0\bE,\,\sigma^{\lambda_1}) & \rightarrow & (\widetilde{\bE}, \tilde{\sigma}_0) \\
      & u_0f& \mapsto     & f|_{y_1 \rightarrow \beta_1}
\end{array}
\] 
for every $f \in \bE$, where the automorphism $\tilde{\sigma}_0$ on $\widetilde{\bE}$ is given by $$\tilde{\sigma}_0(\tilde{f}):=\sigma^{\lambda_1}(\tilde{f})|_{y_1 \rightarrow \beta_1}$$ for every $\tilde{f} \in \widetilde{\bE}$. Moreover, $\tau^{-1}(\tilde{f})=u_0\tilde{f}$ for every $\tilde{f} \in \widetilde{\bE}$. 
On the other hand, by Lemma~\ref{LM:anRmonomial+idem} (vi), $(\widetilde{\bE}, \tilde{\sigma}_0)$ is an $R\Pi\Sigma^*$-tower over $(\bA, \sigma^{\lambda_1})$.

If $\ell=1$, then the proposition holds  by Lemma~\ref{LM:anRmonomial+idem} (v). 

Assume that $\ell>1$  and the conclusion holds for $\ell-1$. 
By Remark \ref{RE:constant-stable}, $\const(\bA,\, \sigma^{\lambda_1})=C$ and $(\bA,\, \sigma^{\lambda_1})$ is a constant-stable integral domain. Note that the number of generators in $\widetilde{\bE}$ is equal to $\ell-1$ and the product of the orders of $R$-monomials in $\widetilde{\bE}$ is equal to $\lambda/\lambda_1$,  denoted by $\mu$.
So by the induction hypothesis, 
\begin{itemize}
  \item  one can compute $v_0, v_1, \ldots, v_{\mu-1} \in \widetilde{\bE}$ such that $(\widetilde{\bE},\, \tilde{\sigma}_0)$ is an idempotent difference ring with respect to $\{v_i~|~i\in [\mu-1]_0\}$;
  \item there is  $\{\beta_i\in C~|~2 \le i \le k\}$ such that $(\bH, \sigma_0)$, where for every $h\in \bH$, $\sigma_0(h):=\tilde{\sigma}_0^{\mu}(h)|_{\{y_i \rightarrow \beta_i | 2\le i \le \ell\}}$,  is a $\Pi\Sigma^*$-tower over $(\bA,\, \sigma^{\lambda})$. 
Moreover, we obtain a difference ring isomorphism
\[\begin{array}{cccc}
\tilde{\delta}: & (v_0\tilde \bE,\tilde \sigma_0^\mu) & \rightarrow & (\bH, \sigma_0) \\
      & v_0f& \mapsto     & f|_{\{y_i \rightarrow \beta_i\mid 2\le i\le \ell\}},
\end{array}
\] 
for every $f\in \widetilde{\bE}$, and $\tilde{\delta}^{-1}(h)=v_0h$ for every $h\in \bH$.
\end{itemize}

The proof will be completed by the following three claims.

\noindent{\bf Claim~1.} For every $f \in \bE$ and $i \in \bZ^{+}$, $\tilde{\sigma}_0^i(f|_{y_1 \rightarrow \beta_1})=\sigma^{\lambda_1i}(f)|_{y_1 \rightarrow \beta_1}$.

{\em Proof of Claim~1.} Write $f$ as $\sum_{i=0}^{\lambda_1-1}f_iy_1^i$, where $f_i \in \widetilde{\bE}$. Then
$$\tilde{\sigma}_0(f|_{y_1 \rightarrow \beta_1})=\tilde{\sigma}_0\left(\sum_{i=0}^{\lambda_1-1}f_i \beta_1^i\right)=\sum_{i=0}^{\lambda_1-1}\beta_1^i \sigma^{\lambda_1}(f_i)|_{y_1 \rightarrow \beta_1}$$
by $\beta_1\in C$, and
$$\sigma^{\lambda_1}(f)|_{y_1 \rightarrow \beta_1}=\sum_{i=0}^{\lambda_1-1}(\sigma^{\lambda_1}(f_i)y_1^i)|_{y_1 \rightarrow \beta_1}=\sum_{i=0}^{\lambda_1-1}\beta_1^i \sigma^{\lambda_1}(f_i)|_{y_1 \rightarrow \beta_1}$$
by $\sigma^{\lambda_1}(y_1)=y_1$. So $\tilde{\sigma}_0(f|_{y_1 \rightarrow \beta_1})=\sigma^{\lambda_1}(f)|_{y_1 \rightarrow \beta_1}$. Hence, the claim is completed by induction on $i$.
 
In the rest of this proof, set $e_0:=u_0v_0$. 

\noindent{\bf Claim~2.} For every $h \in \bH$, $\sigma_0(h)=\sigma^{\lambda}(h)|_{\{y_i \rightarrow \beta_i \mid i \in [\ell]\}}$. Furthermore, $\delta:=\tilde{\delta} \circ \tau|_{e_0\bE}$ is a difference ring isomorphism from $(e_0\bE,\,\sigma^{\lambda})$ onto $(\bH, \sigma_0)$ by sending $e_0f$ to $f|_{\{y_i \rightarrow \beta_i\mid i\in [\ell]\}}$. Moreover, $\delta^{-1}(h)=e_0h$ for every $h\in \bH$.

{\em Proof of Claim~2.} Since $h$ is free of $y_1$ and $\lambda_1\mu=\lambda$, we have 
$\tilde{\sigma}_0^{\mu}(h)=\tilde{\sigma}_0^{\mu}(h|_{y_1 \rightarrow \beta_1})=\sigma^{\lambda}(h)|_{y_1 \rightarrow \beta_1}$ by Claim 1. It follows that $\sigma_0(h)=\sigma^{\lambda}(h)|_{\{y_i \rightarrow \beta_i \mid i \in [\ell]\}}$ by the definition of $\sigma_0$. 
Set $\tilde{\tau}:=\tau|_{e_0\bE}$. Then $\tilde{\tau}$ is a ring isomorphism from $e_0\bE$ onto $v_0\widetilde{\bE}$ by mapping $e_0f$ to $(v_0 f)|_{y_1\rightarrow \beta_1}$ for every $f\in \bE$. On the other hand, $(e_0\bE,\,\sigma^{\lambda})$ and $(v_0\widetilde{\bE}, \tilde{\sigma}_0^{\mu})$ are two difference rings. Moreover,
$$\tilde{\sigma}_0^{\mu}(\tilde{\tau}(e_0f))=\tilde{\sigma}_0^{\mu}((v_0f)|_{y_1\rightarrow \beta_1})=\sigma^{\lambda}(v_0f)|_{y_1\rightarrow \beta_1}=\tau(u_0\sigma^{\lambda}(v_0f))=\tilde{\tau}(e_0\sigma^{\lambda}(f))$$
by Claim~1. So $\tilde{\tau}$ is a difference ring isomorphism from $(e_0\bE,\,\sigma^{\lambda})$ onto $(v_0\widetilde{\bE}, \tilde{\sigma}_0^{\mu})$. Hence, 
$\delta:=\tilde{\delta} \circ \tilde{\tau}$ is a difference ring isomorphism from $(e_0\bE,\,\sigma^{\lambda})$ onto $(\bH, \sigma_0)$. Note that $v_0$ is free of $y_1$. So $\delta(e_0f)=\tilde{\delta}(v_0\cdot (f|_{y_1\rightarrow \beta_1}))=f|_{\{y_i \rightarrow \beta_i\mid i\in [\ell]\}}$. Furthermore, for every $h\in \bH$, $\delta^{-1}(h)=\tilde{\tau}^{-1}(\tilde{\delta}^{-1}(h))=\tilde{\tau}^{-1}(v_0h)=e_0h$.

\noindent{\bf Claim~3.} $\bE= e_0\bE \oplus\sigma( e_0)\bE\oplus \cdots \oplus\sigma^{\lambda-1} ( e_{0})\bE.$

{\em Proof of Claim~3.} Recall that $\tau^{-1}$ is a difference ring isomorphism from $(\widetilde{\bE},\, \tilde{\sigma}_0)$ onto $(u_0\bE,\, \sigma^{\lambda_1})$ given by $\tau^{-1}(f)=u_0f$ for every $f\in \widetilde{\bE}$. Then, for every $j \in [\mu-1]_0$, 
 $$\tau^{-1}(\tilde{\sigma}_0^j (v_0))=\tau^{-1}(\sigma^{\lambda_1j}(v_0)|_{y_1 \rightarrow \beta_1})=u_0(\sigma^{\lambda_1j}(v_0)|_{y_1 \rightarrow \beta_1})=u_0\sigma^{\lambda_1j}(v_0)=\sigma^{\lambda_1j}(e_0).$$
  The first and third equalities are due to Claim 1 and Lemma~\ref{LM:anRmonomial+idem}~(ii), respectively. Furthermore, it follows from $u_0^2=u_0$ and $\sigma^{\lambda_1}(u_0)=u_0$ that $\tau^{-1}(\tilde{\sigma}_0^j (v_0)\widetilde{\bE})=\sigma^{\lambda_1j}(e_0)u_0\bE=\sigma^{\lambda_1j}(u_0e_0)\bE=\sigma^{\lambda_1j}(e_0)\bE$. By Lemma~\ref{LM:idemp+equivdef}, we have $\tilde \bE=v_0\tilde \bE\oplus \tilde\sigma_0(v_0)\tilde \bE\oplus\cdots \oplus  \tilde\sigma_0^{\mu-1}(v_0)\tilde \bE$. Applying $\tau^{-1}$ to the above direct sum yields 

 $$u_0\bE=e_0\bE \oplus \sigma^{\lambda_1}(e_0)\bE \oplus \cdots \oplus \sigma^{\lambda_1(\mu-1)}(e_0)\bE.$$
Recall that $\sigma(u_i)=u_{i+1  \! \bmod \! \lambda_1}$. Then applying $\sigma, \sigma^2, \ldots, \sigma^{\lambda_1-1}$ to the above equation  successively leads to
\[
  \left\{
\begin{aligned}
u_0\bE& =e_0\bE \oplus \sigma^{\lambda_1}(e_0)\bE \oplus \cdots \oplus \sigma^{\lambda_1(\mu-1)}(e_0)\bE\\
u_1\bE  &= \sigma(e_0)\bE \oplus \sigma^{\lambda_1+1}(e_0)\bE\oplus \cdots \oplus \sigma^{\lambda_1(\mu-1)+1}(e_0)\bE \\
\vdots \\
  u_i\bE   &= \sigma^i(e_0)\bE \oplus \sigma^{\lambda_1+i}(e_0)\bE\oplus \cdots \oplus \sigma^{\lambda_1(\mu-1)+i}(e_0)\bE \\
   \vdots \\
   u_{\lambda_1-1}\bE  &= \sigma^{\lambda_1-1}(e_0)\bE \oplus \sigma^{2\lambda_1-1}(e_0)\bE\oplus \cdots \oplus \sigma^{\lambda_1\mu-1}(e_0)\bE.
\end{aligned}
\right.
\]
Recall that $\lambda=\lambda_1\mu$. Arranging the right-hand side of the above system in an $S$-shaped order, we obtain
 $$\bE=u_0\bE \oplus u_1\bE\oplus \cdots \oplus u_{\lambda_1-1}\bE=\directsumideals_{i=0}^{\lambda-1} \sigma^i(e_0)\bE,$$
 which proves the third claim.\\
Finally, we are ready to show that the conclusion holds for the case $\ell$. The second assertion follows from Claim 2 and the induction hypothesis. For every $i \in [\lambda-1]$, let $e_i:=\sigma^{i}(e_0)$. Then $\bE$ is a weak idempotent difference ring with respect to $\{e_i~|~i \in [\lambda-1]_0\}$ by Lemma~\ref{LM:idemp+equivdef} and Claim 3. Moreover, we note that $\bH$ is an integral domain, so is $e_0\bE$ by the second assertion. Hence, the first assertion holds.
 \end{proof}

The constructive proof leads to the following algorithm.
 
 \begin{alg}\label{ALG:IdemDecomp}
 {\tt IdempotentDecompositionsFor$R\Pi\Sigma^*$-Towers}$(\bE,\, \sigma)$
 
  \noindent
 {\tt Input:} 
An $R\Pi\Sigma^*$-tower $(\bE,\, \sigma)$ over a computable, constant-stable, integral difference ring $(\bA,\,\sigma)$ with constant field $C$, of
the form
$\bE =
\bA\langle y_1\rangle\ldots\langle y_\ell\rangle
\langle p_1\rangle\ldots\langle p_m\rangle
\langle s_1\rangle\ldots\langle s_n\rangle,
$
where $\ell>0$, and the $y_i$, $p_i$, and $s_i$ are $R$-, $\Pi$-, and $\Sigma^*$-monomials, respectively.

 \noindent
 {\tt Output:} $\{(e_0, \lambda), (\bH,\, \sigma_0), \delta\}$ such that
 $\bE=\directsumideals_{i=0}^{\lambda-1} \sigma^i(e_0)\bE$ with the idempotent $e_0 \in \bE$ and $\lambda\in \bZ^{+}$; $(\bH,\, \sigma_0)$ is a $\Pi\Sigma^*$-tower over $(\bA,\, \sigma^{\lambda})$ with constant field $C$  and $\delta$ is a difference ring isomorphism from $(e_0\bE,\,\sigma^{\lambda})$ onto  $(\bH,\, \sigma_0)$.
 
\begin{enumerate}
	\item[(1)] Set $a_1:=\frac{\sigma(y_1)}{y_1}$, $\lambda_1:=\ord(y_1)$ and $u_0:=\prod_{j=0}^{\lambda_1-2}\frac{y_1-a_1^j}{a_1^{\lambda_1-1}-a_1^j}$.
	
	\item[(2)] Set $\beta_1:=a_1^{\lambda_1-1}$ and $\widetilde{\bE}:=C(x)\langle y_2 \rangle \ldots \langle y_{\ell} \rangle \langle p_1 \rangle \ldots \langle p_{m} \rangle  \langle s_1 \rangle \ldots \langle s_{n} \rangle$ together with $\tilde{\sigma}_0(f):=\sigma^{\lambda_1}(f)|_{y_1 \rightarrow \beta_1}$ for  $f\in \widetilde{\bE}$, and  
\[\begin{array}{cccc}
		\tau: & (u_0\bE,\,\sigma^{\lambda_1}) & \rightarrow & (\widetilde{\bE}, \,\tilde{\sigma}_0) \\
		& u_0f& \mapsto     & f|_{y_1 \rightarrow \beta_1}.
	\end{array}
	\] 
	\item [(3)] $(^*${\sl Base case}$^*)$ If $\ell=1$ then return $\{(u_0, \lambda_1), (\widetilde{\bE}, \tilde{\sigma}_0), \tau\}$. 	
	\item [(4)] $(^*${\sl Recursion}$^*)$ Note that the number of $R$-monomials in $\widetilde{\bE}$ is $\ell-1$. 
	So using the algorithm, the output for $(\widetilde{\bE}, \tilde{\sigma}_0)$ is $\{(v_0, \mu), (\bH, \sigma_0), \tilde{\delta}\}$.
	\item [(5)] Return $\{(u_0v_0, \lambda_1\mu), (\bH, \sigma_0),\tilde{\delta} \circ \tau|_{u_0v_0\bE}\}$.
\end{enumerate}
 \end{alg}

\begin{example}\label{EX:IdemDecomp}
Let $C=\bQ(\sqrt{-1})$ and $(C(x),\,\sigma)$ be the rational difference field with $\sigma(x)=x+1$. Consider the $R$-tower $(\bE,\,\sigma)$ over $(C(x),\,\sigma)$ with
 $\bE=C(x)\langle y_1 \rangle \langle y_2 \rangle$ of $C(x)$ and  $\sigma(y_1)=-y_1$ and $\sigma(y_2)=(\frac{1-y_1}{2}+\frac{1+y_1}{2} \sqrt{-1})y_2$ ; compare Example~\ref{EX:CheckRmonomials}. We will compute an idempotent representation of $(\bE,\, \sigma)$ by using Algorithm~\ref{ALG:IdemDecomp}. 

In step 1, $a_1=-1$, $\lambda_1=2$, $u_0=\frac{1-y_1}{2}$. Further, we get the difference ring isomorphism 
\[\begin{array}{cccc}
\tau: & (u_0\bE,\,\sigma^{\lambda_1}) & \rightarrow & (C(x)\langle y_2 \rangle, \tilde{\sigma}_0) \\
      & u_0f& \mapsto     & f|_{y_1 \rightarrow -1}
\end{array}
\] 
 with $\tilde{\sigma}_0(f)=\sigma^2(f)|_{y_1\rightarrow -1}$ for every $f \in C(x)\langle y_2 \rangle$ by Example~\ref{EX:CheckRmonomials} again. In particular, $\tilde{\sigma}_0(y_2)= \sqrt{-1}y_2$. Next, we go directly to step 4. By recursion, $\mu=4$,
 $$v_0=\frac{y_2^3-\sqrt{-1}y_2^2-y_2+\sqrt{-1}}{4\sqrt{-1}},$$
 and there is a difference ring isomorphism
 \[\begin{array}{cccc}
\tilde{\delta}: & (v_0C(x)\langle y_2 \rangle, \sigma^8) & \rightarrow & (C(x), \sigma^8) \\
      & v_0f& \mapsto     & f|_{y_2 \rightarrow -\sqrt{-1}}.
\end{array}
\] 
Therefore, $\directsumideals_{i=0}^{7}\sigma^i(u_0v_0)\bE$ is an idempotent representation of $(\bE,\, \sigma)$. Furthermore, $(u_0v_0\bE,\, \sigma)$ and $(C(x), \sigma^8)$ are isomorphic by sending $u_0v_0f$ to $f|_{\{y_1 \rightarrow -1, y_2 \rightarrow -\sqrt{-1}\}}$ for every $f\in \bE$.
\end{example}

\subsection{Complete reductions for $R\Pi\Sigma^*$-towers}\label{SUBSECT:crforgeneralcase}

First we show the following three lemmas. It sets the stage for constructing a complete reduction for an $R\Pi\Sigma^*$-tower.

\begin{lemma}\label{LM:projecttothefirstcompotent}
Let $(\bA,\, \sigma)$ be an idempotent difference ring with respect to $\{e_i~|~i\in[\lambda-1]_0\}$ and $\lambda>1$.  Then, for every $f\in \bA$, 
  $f=\Delta(v)+e_0\sum_{i=0}^{\lambda-1}\sigma^{-i}(f)$
for $v=\sum_{i=1}^{\lambda-1}\sum_{j=1}^i e_{i-j} \sigma^{-j}(f)$.
\end{lemma}
\begin{proof} 
First we have $f=\sum_{i=0}^{\lambda-1}e_if$ by \eqref{EQ:idempotentrepresentation}. For $i\in [\lambda-1]$, set $u_i$ to be $\sum_{j=1}^i e_{i-j} \sigma^{-j}(f)$. It follows from telescoping that $e_if=\Delta(u_i)+e_0\sigma^{-i}(f)$. So $f=\Delta(\sum_{i=1}^{\lambda-1}u_i)+e_0\sum_{i=1}^{\lambda-1}\sigma^{-i}(f).$ 
\end{proof}

\begin{lemma}\label{LM:lambdatozero}
Let $(\bA,\, \sigma)$ be a difference ring, $\lambda > 1$, and $f\in \bA$. If $f=\Delta^{(\lambda)}(g)$ for some $g \in \bA$, then there is $\tilde{g}=\sum_{i=0}^{\lambda-1}\sigma^{i}(g)$ such that $f=\Delta(\tilde{g})$. In particular, $ \Delta^{(\lambda)}(\bA) \subset \Delta(\bA)$.
\end{lemma}
\begin{proof} 
This is an immediate consequence of telescoping cancellation.
\end{proof}

We reduce the problem of constructing a complete reduction of an idempotent difference ring to constructing it on one of its idempotent components.

\begin{lemma}\label{LM:convert}
Let $(\bA,\, \sigma)$ be an idempotent difference ring with respect to $\{e_i~|~i\in[\lambda-1]_0\}$ and $\lambda>1$. Assume that the set $C$ of constants in $\bA$ is a field. Define the map
\[
\begin{array}{cccc}
\psi:& \bA & \rightarrow & e_0\bA \\
      & f & \mapsto     & e_0\sum_{i=0}^{\lambda-1}\sigma^{-i}(f)
\end{array}
\]
and let $\phi$ be a complete reduction for $(e_0\bA, \, \Delta^{(\lambda)})$ that is $e_0C$-linear. 
Then
\begin{itemize}
\item [(i)]  $\psi$ is $C$-linear with $\psi(e_jf)=e_0\sigma^{-j}(f)$ for every $f \in \bA$ and $j\in[\lambda-1]_0$.
 \item [(ii)] $\phi \circ \psi$ is a complete reduction for $(\bA, \, \Delta)$ as a $C$-linear map.
 \item [(iii)] Let $f\in \bA$. Then $(\sum_{i=1}^{\lambda-1}\sum_{j=1}^i e_{i-j} \sigma^{-j}(f)+\sum_{i=0}^{\lambda-1}\sigma^{i}(g), r)$ is a $\Sigma$-pair with respect to $\phi \circ \psi$ if $(g, r)$ is a $\Sigma$-pair of $\psi(f)$ with respect to $\phi$.
 \item[(iv)] Let $f\in\bA$. Then $f\in \Delta(\bA)$ if and only if $\psi(f) \in \Delta^{(\lambda)}(e_0\bA)$. In particular, if there is  $\tilde{g} \in e_0\bA$ such that $\psi(f)=\Delta^{(\lambda)}(\tilde{g})$, then $f=\Delta(g)$ with $g=\sum_{i=1}^{\lambda-1}\sum_{j=1}^i e_{i-j} \sigma^{-j}(f)+\sum_{i=0}^{\lambda-1}\sigma^{i}(\tilde{g})$.
\end{itemize} 
\end{lemma}

\begin{proof} 
(i) Let $f\in \bA$ and $i,j\in [\lambda-1]$. Then $\psi$ is $C$-linear by $\sigma^{-i}(cf)=c\sigma^{-i}(f)$ for every $c\in C$. Consider $e_0 \sigma^{-i}(e_j)$. It is equal to $e_0$ if $i=j$, and 0 otherwise. So $\psi(e_jf)=e_0\sum_{i=0}^{\lambda-1}\sigma^{-i}(e_jf)=e_0\sigma^{-j}(e_jf)=e_0\sigma^{-j}(f)$.

(ii) Since $C$ is a field and $e_0$ is the multiplicative identity of $e_0C$, $e_0C$ is a subfield of $e_0\bA$. Moreover, $e_0C \subset \const(e_0\bA,\, \sigma^{\lambda})$. So $\phi$ is well-defined. 
Let $f\in \bA$ and $c\in C$. 
Then $$\phi(\psi(cf))=\phi(c\psi(f))=\phi(ce_0\psi(f))=ce_0(\phi(\psi(f)))=c\phi(\psi(f)).$$
The first and third equalities follow from the fact that $\psi$ and $\phi$ are $C$-linear and $e_0C$-linear, respectively. 
The second and fourth equalities hold because $\im(\psi) \subset e_0\bA$ and $r=e_0r$ for all $r \in e_0\bA$,  as noted in Remark \ref{RE:idemp+unit}~(i).
Hence, $\phi \circ \psi$ is $C$-linear. Next, we are going to show that $\phi \circ \psi$ is a complete reduction for $(\bA, \, \Delta)$. 
Since $\phi(\psi(f)) \in e_0\bA$, we have $\psi(\phi(\psi(f)))=\phi(\psi(f))$ by (i), which, together with $\phi^2=\phi$, implies that $\phi \circ \psi \circ \phi \circ \psi(f)=\phi \circ \psi(f).$ So $\phi \circ \psi$ is an idempotent map. 

By Remark \ref{RE:cr}, it remains to show $\Delta(\bA)=\ker(\phi \circ \psi)$. Assume that $f\in \Delta(\bA)$. Then there is $g\in \bA$ such that $f=\Delta(g)$. According to \eqref{EQ:idempotentrepresentation}, we expand $\Delta(g)$ as $\sum_{i=0}^{\lambda-1}e_i\Delta(g)$. Then, by (i) again,
$\psi(e_i\Delta(g))=e_0\sigma^{-i}(\sigma(g)-g)=e_0(\sigma^{1-i}(g)-\sigma^{-i}(g)).$ This yields
$$\psi(f)=e_0\sum_{i=0}^{\lambda-1}(\sigma^{1-i}(g)-\sigma^{-i}(g))=e_0(\sigma(g)-\sigma^{1-\lambda}(g))=\sigma^{\lambda}(\tilde{g})-\tilde{g} $$
for $\tilde{g}=e_0\sigma^{1-\lambda}(g)$. So $\psi(f)\in \Delta^{(\lambda)}(e_0\bA)$. Hence, $f \in \ker(\phi \circ \psi)$ by $\ker(\phi)=\Delta^{(\lambda)}(e_0\bA)$. Conversely, let $f \in \ker(\phi \circ \psi)$. Then $\psi(f) \in \Delta^{(\lambda)}(\bA)$ by $e_0\bA \subset \bA$ and $\ker(\phi)=\Delta^{(\lambda)}(e_0\bA)$ again.  Lemma~\ref{LM:lambdatozero} ensures that $\psi(f) \in \Delta(\bA)$. Consequently, $f \in \Delta(\bA)$ by Lemma~\ref{LM:projecttothefirstcompotent}.

(iii) Set $v$ to be $\sum_{i=1}^{\lambda-1}\sum_{j=1}^i e_{i-j} \sigma^{-j}(f)$.  By Lemma~\ref{LM:projecttothefirstcompotent}, $f=\Delta(v)+\psi(f)$.  Let $(g, r)$ be a $\Sigma$-pair of $\psi(f)$ with respect to $\phi$.
Then $f=\Delta(v)+\Delta^{(\lambda)}(g)+r=\Delta(v+\sum_{i=0}^{\lambda-1}\sigma^{i}(g))+r$ by Lemma~\ref{LM:lambdatozero}. Note that $r=\phi \circ \psi(f) \in \im(\phi \circ \psi)$. So $(v+\sum_{i=0}^{\lambda-1}\sigma^{i}(g), r)$ is a $\Sigma$-pair of $f$ with respect to  $\phi \circ \psi$.

(iv) By (ii), $\phi \circ \psi$ is a complete reduction for $(\bA, \, \Delta)$. So $\Delta(\bA)=\ker(\phi \circ \psi)$, which implies that $f\in \Delta(\bA)$ if and only if $f\in \ker(\phi \circ \psi)$, i.e., $\psi(f) \in \ker(\phi)$. Moreover, $\ker(\phi)=\Delta^{(\lambda)}(e_0\bA)$ by the definition of $\phi$.
It follows that  $\psi(f)\in \ker(\phi)$ if and only if $\psi(f)\in \Delta^{(\lambda)}(e_0\bA)$. Thus, $f\in \Delta(\bA)$ if and only if $\psi(f) \in \Delta^{(\lambda)}(e_0\bA)$. The remaining conclusion follows directly from (iii).
\end{proof}

We are now ready to state the main result in this section.

\begin{thm}\label{THM:CR+RPiSigma}
	Let $(\bA,\,\sigma)$ be a computable difference ring which is constant-stable and an integral domain, and  let $(\bE,\,\sigma)$ be an $R\Pi\Sigma^*$-tower over $(\bA, \, \sigma)$. Assume that the set of constants $C$ forms a field. Then the following assertions hold.
	\begin{itemize}
    \item [(i)] If one can solve the telescoping problem for every $ i\in\bZ^{+}$ and every $\Pi\Sigma^*$-tower over $(\bA,\, \sigma^{i})$,
    then it can also be  solved for $(\bE,\,\sigma)$.
    \item [(ii)] If  
    $\bA$ has an effective $C$-basis, and there exists an algorithm which, for every $i\in \bZ^{+}$ and $f \in \bA^*$, constructs a complete reduction $\phi_f^{(i)}$ for $(\bA,\, \Delta^{(i)}_{f})$. Then there is an algorithm for constructing a complete reduction $\omega$ for $(\bE, \, \Delta)$. 
    \end{itemize}
\end{thm}
\begin{proof}
 If $\bE$ is free of $R$-monomials, then (i) is trivial and (ii) follows by Theorem~\ref{THM:SimpleRPiSigmaExt}. 
Otherwise,  
let $\bE=\bA\langle t_1 \rangle \langle t_2 \rangle \ldots \langle t_k \rangle$. Define 
$$\bI:=\{1\le i \le k~|~t_i~\text{is an $R$-monomial}\}\neq\emptyset$$ 
and $\lambda:=\prod_{i\in \bI} \ord(t_i)$. Moreover, set $\bH:=\bA\langle t_i\mid i\in [k]\setminus \bI\rangle$.
By Proposition~\ref{PROP:RPiSigmaIsIdempotent}~(i) (for details see Algorithm~\ref{ALG:IdemDecomp}), one can compute $e_0,\ldots,e_{\lambda-1} \in \bE$ such that $(\bE,\,\sigma)$ is an idempotent difference ring  with respect to $\{e_i~|~i\in[\lambda-1]_0\}$. And
there is $\{\beta_i\in C~|~i\in \bI \}$ such that $(\bH,\, \sigma_0)$ with $\sigma_0(h)=\sigma^{\lambda}(h)|_{\{t_i \rightarrow \beta_i|i \in \bI\}}$ for every $h \in \bH$ is a $\Pi\Sigma^*$-tower over $(\bA,\, \sigma^{\lambda})$. Moreover, we get a difference ring isomorphism
\[
\begin{array}{cccc}
	\delta: & (e_0\bE,\,\sigma^\lambda) & \rightarrow & (\bH, \sigma_0) \\
	& e_0f& \mapsto     & f|_{\{t_i \rightarrow \beta_i\mid i\in\bI\}}
\end{array}
\]
for every $f \in \bE$ and $\delta^{-1}(h)=e_0h$ for every $h \in \bH$. Note that $\const(\bH,\, \sigma_0)=C$. 
Note further that a complete reduction for $(e_0\bA, \,\Delta^{(\lambda)})$ that is $e_0C$-linear always exists, as assumed in Lemma~\ref{LM:convert};
 we will use the existence in statement (i), but rely on an explicit construction in statement (ii).

Let $\psi$ be a $C$-linear function from $\bE$ to $e_0\bE$ by sending $f$ to $e_0\sum_{i=0}^{\lambda-1}\sigma^{-i}(f)$. And set $\rho:=\sigma_0-\bf{1}$. 

(i) Let $f\in\bE$. Then $f\in \Delta(\bE)$ if and only if $\psi(f) \in \Delta^{(\lambda)}(e_0\bE)$ by Lemma~\ref{LM:convert} (iv). Moreover, $\psi(f) \in \Delta^{(\lambda)}(e_0\bE)$ if and only $\delta(\psi(f)) \in \rho(\bH)$. By the assumption, the latter statement can be decided.  
If there is $g' \in \bH$ such that $\rho(g')=\delta(\psi(f))$, then $\delta^{-1} \circ \sigma^{\lambda}=\sigma_0 \circ \delta$ implies that $\psi(f)=\Delta^{(\lambda)}(\delta^{-1}(g'))$. Thus, $f=\Delta(g)$ with $g=\sum_{i=1}^{\lambda-1}\sum_{j=1}^i e_{i-j} \sigma^{-j}(f)+\sum_{i=0}^{\lambda-1}\sigma^{i}(\delta^{-1}(g'))$ by Lemma~\ref{LM:convert} (iv) again.

(ii) Under the assumption of statement (ii), there is an algorithm for constructing a complete reduction $\Phi$ as a $C$-linear map for $(\bH,\, \rho)$ by Theorem~\ref{THM:SimpleRPiSigmaExt}. 
By Lemma~\ref{LM:diffiso+CR}, $\phi:=\delta^{-1} \circ \Phi \circ \delta$ is a complete reduction for $(e_0\bE,\, \Delta^{(\lambda)})$.
Moreover, $\omega:=\phi\circ \psi$ is a complete reduction for $(\bE,\,\Delta)$ by Lemma~\ref{LM:convert}~(ii). It remains to describe how to compute $\Sigma$-pairs with respect to $\omega$. Let $f\in\bE$ and denote $\psi(f)$ by $f'$. Then by assumption, we can compute a $\Sigma$-pair $(g',r')$ of $\delta(f')$ with respect to $\Phi$. 
 Then $(g, r):=(\delta^{-1}(g'),\delta^{-1}(r'))$ is a $\Sigma$-pair of $f'$ with respect to $\phi$ by Lemma~\ref{LM:diffiso+CR}. Finally, $(\sum_{i=1}^{\lambda-1}\sum_{j=1}^i e_{i-j} \sigma^{-j}(f)+\sum_{i=0}^{\lambda-1}e_i\sigma^{i}(g), r)$ is a $\Sigma$-pair of $f$ with respect to $\omega$ by Lemma~\ref{LM:convert}~(iii). 
\end{proof}

Combining Theorem~\ref{THM:CR+RPiSigma}~(i) with Theorem~\ref{Thm:PiSiOverSimpleRPiSiDESolver} leads to the following general telescoping solver for general $R\Pi\Sigma^*$-towers.

\begin{cor}\label{Cor:PiSiOverRPiSiTelescoper}
	Let $(\bE,\,\sigma)$ be an $R\Pi\Sigma^*$-tower over a $\Pi\Sigma^*$-field whose constant field $C$ is orbit-computable. Then, for every $g\in\bE$, one can decide algorithmically if there is $q\in\bE$ with $\Delta(q)=g$.
\end{cor}

More generally, one may use the telescoping tools developed and implemented within the summation package \texttt{Sigma}, yielding telescoping algorithms  for cases where the ground ring can be, e.g., also a $\Pi\Sigma^*$-field  over unspecified sequences~\cite{KS:06,PS:19}.\\ 
We remark that combining \cite[Proposition~6.4]{Schn2017} with the idempotent representation from Proposition~\ref{PROP:RPiSigmaIsIdempotent} provides an alternative algorithmic approach to Corollary~\ref{Cor:PiSiOverRPiSiTelescoper}.
However, the proposed reduction strategy therein requires solving the underlying telescoping problem for each component separately; in contrast, the new algorithm derived in Theorem~\ref{THM:CR+RPiSigma} reduces this requirement to solving the telescoping problem in only one of the components. This new feature is not only more efficient, but provides also more flexibility. One can apply all of the refined telescoping algorithms given in~~\cite{Schn2007,Schn2008,Schneider:2015,Schn2016} straightforwardly (without synchronizing refined extensions coming from different components).

 The statement (ii) of Theorem~\ref{THM:CR+RPiSigma} reduces the calculation of a complete reduction for a general $R\Pi\Sigma^*$-tower to that for a $\Pi\Sigma^*$-tower over 
 the ground ring.
As a consequence, we obtain the following algorithmic version for the rational difference field.

\begin{cor}\label{COR:CR+RPiSigma+Rational}
Let $\bF=C(x)$ be the rational difference field with $\sigma(x)=x+1$, where $C$ is factorizable, and let $(\bE,\, \sigma)$ be an $R\Pi\Sigma^*$-tower over $(\bF,\, \sigma)$. Then there is an algorithm for constructing a complete reduction $\omega$ for $(\bE,\, \Delta)$.
\end{cor}
\begin{proof} 
Note that $(\bF,\, \sigma)$ is constant-stable by \cite[Proposition~32]{Schn2020}. Then the result follows from  Proposition~\ref{PROP:Generalrational+highorder} and Theorem~\ref{THM:CR+RPiSigma}~(ii). \end{proof}

We summarize the algorithmic procedure as follows.

 \begin{alg}\label{ALG:SigmaPairinGeneralRPiSigma}
{\tt CRFor$R\Pi\Sigma^*$-Towers}$(\bE, g)$

	\noindent
	{\tt Input:} an  $R\Pi\Sigma^*$-tower $(\bE,\,\sigma)$ over $(C(x),\,\sigma)$ with $\sigma(x)=x+1$, where $C$ is factorizable; $g\in\bE$.
	
	\noindent
	{\tt Output:} a $\Sigma$-pair of $g$ with respect to $\omega$ as given in Corollary \ref{COR:CR+RPiSigma+Rational}.
\begin{enumerate}

\item[(1)] If $(\bE,\,\sigma)$ is a $\Pi\Sigma^*$-tower, use Algorithm~\ref{ALG:CRForSimpleRPiSigma-Towers} to derive the desired result. Otherwise, we proceed as follows.

\item[(2)] Applying Algorithm~\ref{ALG:IdemDecomp} to $(\bE,\,\sigma)$ yields $\{(e_0, \lambda), (\bH,\, \sigma_0),\delta\}$.
Here $e_0\in\bE$ is an idempotent and $\bE=\directsumideals_{i=0}^{\lambda-1}\sigma^i(e_0)\bE$. Moreover, $(\bH,\,\sigma_0)$ is a $\Pi\Sigma^*$-tower over $(C(x),\sigma^\lambda)$, and $\delta$ is a difference ring isomorphism from $(e_0\bE,\,\sigma^\lambda)$ onto $(\bH,\, \sigma_0)$.
    
\item[(3)] Set $\psi$ to be a $C$-linear function from $\bE$ to $e_0\bE$ by sending $f$ to $e_0\sum_{i=0}^{\lambda-1}\sigma^{-i}(f)$. 
Using Algorithm~\ref{ALG:CRForSimpleRPiSigma-Towers}, we get a $\Sigma$-pair $(u,v)$ of $\delta(\psi(f))$ with respect to a complete reduction $\Phi$ for $(\bH,\, \sigma_0-\bf{1})$.
 
 \item [(4)] 
Return $$(\sum_{i=1}^{\lambda-1}\sum_{j=1}^i e_{i-j} \sigma^{-j}(f)+\sum_{i=0}^{\lambda-1}e_i\sigma^{i}(u), e_0v).$$
\end{enumerate}
\end{alg} 
 We note that the proposed algorithm can be improved further to compute $\Sigma^*$-pairs in an $R\Pi\Sigma^*$-tower $(\bE,\,\sigma)$. Let $\bE=\bA\langle s_1 \rangle \ldots \langle s_n \rangle$, where $(\bA,\,\sigma)$ is an $R\Pi$-tower and the $s_i$ are $\Sigma^*$-monomials.  Instead of computing an idempotent representation of the full tower $\bE$, one can apply Algorithm~\ref{ALG:SigmaPairinGeneralRPiSigma} only to compute $\Sigma$-pairs in the $R\Pi$-tower $\bA$ and then handle the $\Sigma^*$-monomials using Algorithm~\ref{ALG:Sigmacase}. We denote the resulting algorithm by \texttt{CRhalfID}. In particular, if $(\bE, \,\sigma)$ is a simple $R\Pi\Sigma^*$-tower, Algorithm~\ref{ALG:CRForSimpleRPiSigma-Towers} can be applied directly. When applicable, these strategies are expected to outperform Algorithm~\ref{ALG:SigmaPairinGeneralRPiSigma} (\texttt{CRfullID}). We illustrate this by recomputing the benchmarks from Table~\ref{tab:Tele2} with \texttt{CRfullID} and \texttt{CRhalfID}, yielding the timings shown below.

\begin{center}
	\begin{tabular}{|c|c|c|c|c|c|c|c|c|c|} \hline
		$i$ & 5&10&15&20&25&30&35&40    \\ \hline
		{\tt CRfullID} &2.30 &35.9&292.1&1484.5&5317.8&14950.9&$>36000$& $>36000$  \\ \hline 
        {\tt CRhalfID} &0.65&4.3&17.3&53.1&132.5&283.1&555.4&984.8   \\ \hline 
	\end{tabular}
\end{center}

Compared to Table \ref{tab:Tele2}, the direct approach to compute complete reductions outperforms all other approaches for the given simple $R\Pi\Sigma^*$-tower. {\tt CRhalfID} still outperforms the function {\tt RefinedTelescoping} in the {\tt Sigma} summation package, and {\tt CRfullID} exhibits relatively poor performance.

\begin{example}\label{EX:newidentity}
Simplify    
	\[\sum_{j=0}^n \frac{((-1)^j(2-j)+j)j!}{2^j \lfloor \frac j2\rfloor!}.\]
 We start with the rational difference field $(\bQ(x), \, \sigma)$ with $\sigma(x)=x+1$ and construct the $R\Pi\Sigma^*$-tower $(\bE,\,\sigma)$ over $(\bQ(x),\sigma)$ with $\bE=\bQ(x)\langle y\rangle \langle p_1\rangle \langle p_2\rangle \langle p_3\rangle$ and
   \[\sigma(y)=-y, \, \sigma(p_1)=2p_1, \, \sigma(p_2)=(x+1)p_2 \,\, \text{and} \,\, \sigma(p_3)= ap_3,\]
   where $a=\frac {x+3-y(x-1)}{4}$. By a direct computation, $\left(\frac{1-y}{x+1}+\frac{1+y}{2}\right)a=1$. So $a\in\bQ(x)\langle y\rangle \langle p_1\rangle \langle p_2\rangle^*$. Since $a$ is not a product of $y, p_1$ and $p_2$, $\bE$ is not simple.  
   On the other hand, $y, p_1$ and $p_2$ model $(-1)^n$, $2^j$ and $j!$, respectively. Note that $\lfloor \frac j2\rfloor!$ can be expressed as $\prod_{k=1}^{j}\frac{1}{4}(k+2+(-1)^k(k-2))$. So $p_3$ models $\lfloor \frac j2\rfloor!$. It follows that the given summand can be represented by 
   $$f=\frac{(y (2 - x) + x) p_2}{p_1 p_3} \in \bE.$$ 
   In this setting, we use Algorithm~\ref{ALG:SigmaPairinGeneralRPiSigma} to compute a $\Sigma$-pair of $f$ with respect to $\omega$ given in Corollary~\ref{COR:CR+RPiSigma+Rational}. 
   In step~2, we get the idempotent representation $e_0\bE \oplus \sigma(e_0)\bE$ of $(\bE,\, \sigma)$ with $e_0=\frac{1-y}2$. And we obtain the $\Pi$-tower $(\bH,\sigma_0)$ over $(\bQ(x),\sigma^2)$ with $\bH=\bQ(x) \langle p_1\rangle \langle p_2\rangle \langle p_3\rangle$ and
       \[\sigma_0(p_1)=4p_1, \, \sigma_0(p_2)=(x+1)(x+2)p_2, \, \text{and} \, \sigma_0(p_3)= \frac{1+x}2 p_3.\]
    Moreover, there is a difference ring isomorphism $\delta\colon (e_0 \bE,\, \sigma^2) \to (\bH,\, \sigma_0)$ by sending $g$ to $g|_{y\to -1}$ for every $g\in e_0\bE$. In step 3,  $\delta(\psi(f))=\frac{1-y}2(f+\sigma^{-1}(f))|_{y\to -1}=2 x^{-1}(2 - x + x^2)p_1^{-1}p_2p_3^{-1}$. Applying Algorithm~\ref{ALG:CRForSimpleRPiSigma-Towers} to
     $(\bH,\sigma_0)$ and $\delta(\psi(f))$, we obtain a $\Sigma$-pair $(u, 0)$ of $\delta(\psi(f))$, where $u=4  (x-1)p_2x^{-1}p_1^{-1} p_3^{-1}$. Finally, we conclude that
     $$\left((2+x-2y+xy)p_2p_1^{-1}p_3^{-1}, 0\right)$$
     is a $\Sigma$-pair of $f$ with respect to $\omega$. Rephrasing this result to our summation objects yields the identity
      \[\sum_{j=0}^n \frac{((-1)^j(2-j)+j)j!}{2^j \lfloor \frac j2\rfloor!}=\frac{2(n+1)!}{2^n \lfloor \frac n2\rfloor!}.\]
\end{example}

\section{Creative telescoping}\label{SECT:ParaTele}

A classical problem in symbolic summation is creative telescoping introduced in~\cite{Zeil1991} for hypergeometric terms that can be stated more generally as follows. Given a bivariate function $F(n,k)$, determine whether there are $m\in \bN$ and polynomials $c_0(n),\ldots, c_m(n)$, free of $k$, not all zero, and $G(n,k)$ such that
\begin{equation}\label{EQ:creativetelescoping}
c_0(n)F(n,k)+\cdots+c_m(n)F(n+m,k)=G(n, k+1)-G(n,k).
\end{equation}
Let $S(n)=\sum_{k=a}^{b} F(n, k)$. If one finds such a solution for \eqref{EQ:creativetelescoping} that holds for $k$ with $a \le k \le b$, then summing \eqref{EQ:creativetelescoping} over $k$, we get a recurrence for $S(n)$:
$$c_0(n)S(n)+\cdots+c_m(n)S(n+m)=q(n).$$
The creative telescoping problem can be reformulated in the difference ring setting as follows. Let $\bK(n, k)$ be a rational function field in the variables $n,k$ over a field $\bK$ of characteristic $0$, and let $(\bK(n, k),\, \{\sigma_n, \sigma_k\})$ be the partial difference field with two commuting  shift operators $\sigma_n$ and $\sigma_k$ defined by $\sigma_k(k)=k+1$ and $\sigma_n(n)=n+1$. Note that $(C(k),\,\sigma_k)$ with $C=\bK(n)$ forms the rational difference field with constant field $C$. 
Let $(\bE,\, \sigma_k)$ be an $R\Pi\Sigma^*$-tower over $(C(k),\, \sigma_k)$ and extend $\sigma_n$ to an appropriate ring automorphism $\sigma_n:\bE\to\bE$ (modeling the shift in $n$ on the summation objects accordingly) such that $\sigma_n\sigma_k(f)=\sigma_k\sigma_n(f)$ for every $f\in \bE$. Given $f\in \bE$, determine whether there are $m\in \bN$, 
$c_0,\ldots,c_m \in C$ not all zero, and $g \in \bE$ such that 
\begin{equation}\label{EQ:paratele}
c_0f+\cdots+c_m\sigma_n^{m}(f)=\sigma_k(g)-g
\end{equation}
holds.  If such a solution exists, then we call $\sum_{i=0}^{m}c_i \sigma_n^{i}$ a {\em telescoper for $f$}. Moreover, if $m$ is minimal such that \eqref{EQ:paratele} holds, then $\sum_{i=0}^{m}c_i \sigma_n^{i}$ is called a {\em minimal telescoper for $f$}; note that it is unique up to multiplication by a constant. 

We note that creative telescoping can be rephrased in the setting of parameterized telescoping by letting $f_i=\sigma_n^{i}(f)$ with $i=0,\dots,m$ in~\eqref{Equ:ParaTele} for a fixed $m$. 
As stated in Property (iii) on page~\pageref{prop:complete_reduction} of the introduction, complete reductions can be used to deal with the above problem as follows.
\begin{prop}\label{PROP:telescoper}
With the notation introduced above, let $\omega$ be a complete reduction for $(\bE,\,\sigma_k-\bf{1})$. 
For $f\in \bE$ and $m\in \bN$, $f$ has a telescoper of order  $m$  if and only if there exist $c_0,\ldots, c_m\in C$ with $c_m \neq 0$ such that
\begin{equation}\label{EQ:telescopingequation}
c_0\omega(f)+c_1 \omega(\sigma_n(f))+\cdots+c_m\omega(\sigma_n^m(f))=0.
\end{equation}
\end{prop}
\begin{proof} 
Let $L=\sum_{i=0}^{m} c_i \sigma_n^i$ for $c_0, \ldots, c_m$ with $c_m \neq 0$ and $(g,r)$ be a $\Sigma$-pair of $L(f)$ with respect to $\omega$. Then $\omega(L(f))=\sum_{i=0}^{m} c_i \omega(\sigma_n^i(f))$ because $\omega$ is $\bK(n)$-linear. Assume that \eqref{EQ:telescopingequation} holds, i.e., $r=0$. Then~\eqref{EQ:paratele} holds and 
$L$ is a telescoper for $f$ with order $m$. Conversely, assume that $L$ is a telescoper for $f$ of order $m$.  Since $\omega$ is a complete reduction, then $r=\omega(L(f))=0$. So \eqref{EQ:telescopingequation} holds. 
\end{proof}

In particular, we can extract the following algorithm.

\begin{alg}{\tt CreativeTelescopingViaCR}$(f, \bE)$ \label{ALG:ParametricTelescoping}

 \noindent
{\tt Input:} An R$\Pi\Sigma^*$-tower $(\bE,\,\sigma_k)$ over $(\bK(n, k),\, \{\sigma_n, \sigma_k\})$, where $\bK$ is factorizable. Furthermore, an element $f\in \bE$ where we assume that $f$ has a telescoper\footnote{This assertion is intended to guarantee that the loop below terminates. However, from a practical standpoint, it is often relaxed in the hope of finding a solution.}.

 \noindent
{\tt Output:} A minimal telescoper $\sum_{i=0}^{m}c_i\sigma_n^i$ for $f$ and $g\in \bE$ such that \eqref{EQ:paratele} holds.

\smallskip \noindent
\begin{enumerate}
	\item[(1)]  Use Algorithm~\ref{ALG:SigmaPairinGeneralRPiSigma} to compute a $\Sigma$-pair $(g_0, r_0)$ of $f$ with respect to a complete reduction $\omega$ for $(\bE,\,\sigma_k-\mathbf{1})$
\item[(2)] If $r_0=0$, then return $(\mathbf{1}, g_0)$. Otherwise, set $g:=g_0$, $r:=\ell_0r_0$, where $\ell_0$ is a new constant indeterminate.
\item [(3)] For $m=1,2,3, \ldots$ do
\begin{itemize}
  \item [(3.1)] Use Algorithm~\ref{ALG:SigmaPairinGeneralRPiSigma} to compute  a $\Sigma$-pair $(g_m, r_m)$ of $\sigma_n(r_{m-1})$ with respect to $\omega$, and update $g, r$ to $\sigma_n(g)+\ell_mg_m$ and $r+\ell_mr_m$, where $\ell_m$ is a new constant indeterminate. 
 
  \item [(3.2)] Find $c_0,\ldots, c_{m} \in \bK(n)$ such that $r=0$ by solving a linear system in $\ell_0, \ldots, \ell_m$ over $\bK(n)$ (note that the coefficients of $r$ with respect to $\ell_i$ can be represented explicitly in a $\bK(n)$-basis of $\bE$; see the construction in the proof of Theorem \ref{THM:SimpleRPiSigmaExt}).\\ 
  If there is a nontrivial solution, return  $(\sum_{i=0}^{m}c_i\sigma_n^i,$  $g|_{\{ \ell_i \rightarrow c_j \mid i\in[m]_0\}})$
  \end{itemize}
 \end{enumerate}
\end{alg}

We present some empirical results about computing telescopers obtained
 by  Algorithm~\ref{ALG:ParametricTelescoping}
 ({\tt CTviaCR}) and the function {\tt CreativeTelescoping}~({\tt CT}) of the summation package {\tt Sigma} that implements the algorithms given in~\cite{Schn2008,Schneider:2015}. 

Here we study in detail the family of sums given by \eqref{EQ:recurrence}.
First we construct the $R\Pi\Sigma^*$-tower $(\bE,\sigma_k)$ of $(\bQ(n)(k),\sigma_k)$ with constant field $C=\bQ(n)$, $\bE:=\bQ(n)( k)[p,p^{-1}][s]$ and
$$\sigma_k(p)=\frac{n-k}{k+1}p \quad \text{and} \quad \sigma_k(s)=s+\frac{1}{k+1}.$$
To model the shift in $n$, we define the automorphism $\sigma_n:\bE\to\bE$ by $\sigma_n|_{\bQ(k)}=\mathbf{1}$, $\sigma_n(n)=n+1$, $\sigma_n(p)=(n+1)p/(n-k+1)$ and $\sigma_n(s)=s$.
Then, for every $\ell\in\bN$ and $i\in\bN$, the summand of $S_{n+i,\ell}$ can be expressed as $\sigma_n^i(f_{\ell})\in\bE$ with
$f_{\ell}=(1-\ell ks+\ell(n-k)s)p^{\ell}$. Note that $S_{n,\ell}$ satisfies a linear recurrence with polynomials in $n$ for every $\ell\in\bN$ by WZ-summation theory~\cite{PWZ1996}. So $f_{\ell}$ has a telescoper, which implies that both Algorithm~\ref{ALG:ParametricTelescoping} and \texttt{CreativeTelescoping} terminate when applied to $f_{\ell}$.
Our experiments up to $\ell = 16$ verified that the telescopers obtained by both algorithms agree up to a multiplicative constant. The timings for $8\leq\ell \leq 16$ are listed in the following table, in seconds. The second row of the table records orders of telescopers for the corresponding summands, as computed by the algorithms. 

\begin{center}
	\begin{tabular}{|c|c|c|c|c|c|c|c|c|c|c|c|} \hline
		$\ell$ & 8 & 9 & 10 & 11 & 12 & 13 & 14 & 15  & 16   \\ \hline
        {\tt order}& 7&9& 9 & 11 & 11 & 13 & 13 & 15 & 15      \\ \hline
		{\tt CT}& 20.19 & 52.70 & 136.1 & 250.5 & 1928.7 & 1242.3 & 14644.4 & 8258.0  & $>36000$ \\ \hline
		{\tt CTviaCR} &  4.08 & 7.18 & 19.3 & 29.3 & 107.1 & 102.2 & 322.0 & 320.7 & 1104.8  \\ \hline
	\end{tabular}
\captionof{table}{Telescoper computation via the classical method}
    \label{tab:CT1}
\end{center}

In~\cite{PS2003}, it has been observed that one may refine the conditions on the telescoping solutions to find telescopers of smaller orders as follows.  
Find the minimal  $m \in \bN$ and a $\Sigma^*$-extension $(\bE [s'],\, \sigma)$ of $(\bE,\,\sigma)$  with  $\sigma_k(s')=s'+a$ for some  $a\in \bQ(n)(k)[p,p^{-1}]$ such that 
\begin{equation}\label{EQ:PT2} 
c_0f_{\ell}+\cdots+c_m\sigma_n^{m}(f_{\ell})=\sigma_k(g_{\ell})-g_{\ell}
\end{equation}
 for some $g_{\ell} \in \bE[s']$. 
Here the summation package \texttt{Sigma} employs \cite[Algorithm~1]{Schn2004} to solve this problem. Similarly, we can carry out this refined version with complete reductions by using the fact that such a solution $g_{\ell}$ of~\eqref{EQ:PT2} has the form $g_{\ell}=c'\,s'+\gamma$ for some $c'\in\bQ(n)$ and $\gamma\in\bE$. Namely, let $\omega$ be the complete reduction for $(\bE,\, \sigma_k-\mathbf{1})$. Then the above problem is equivalent to finding the minimal $m$ such that $a=c_0\omega (f_{\ell})+\cdots+c_m \omega (\sigma_n^{m}(f_{\ell})) \in\bQ(n)(k)[p, p^{-1}]$. This leads to a refined linear system that enables one to search for such $\Sigma$-pairs. We performed additional experimental tests for all values of $\ell$ listed above, confirming that both algorithms produce the same minimal telescopers

\begin{center}
	\begin{tabular}{|c|c|c|c|c|c|c|c|c|c|c|c|} \hline
		$\ell$& 8&9& 10 & $11$ & $12$ & $13$ &  $14$ & $15$ & $16$   \\ \hline
        {\tt Order}& 4 & 5 & 5 & 6 & 6 & 7 & 7&8&8     \\ \hline
		{\tt CT}& 1.48 & 2.76 & 4.08 & 7.70 & 13.39 & 23.98 & 44.1 & 64.9 & 136.8 \\ \hline
		{\tt CTviaCR} & 1.19 & 2.18 & 4.19 & 7.24 & 13.44 & 20.38 & 41.6 & 57.8 & 105.9  \\ \hline
	 \end{tabular}
\captionof{table}{Telescoper computation  via the refined method}
    \label{tab:CT2}
\end{center}

Also here the timings improve slightly, but we have to consider higher $\ell$ to enter into the situation that the underlying brute-force linear system to be solved in \texttt{Sigma} gets substantially more complicated.

\begin{center}
	\begin{tabular}{|c|c|c|c|c|c|c|c|c|c|} \hline
		$\ell$ & 17 & $18$ & $19$ & $20$ &  $21$ & $22$ & $23$   \\ \hline
        {\tt Order} &  9 & 9 & 10 & 10 & 11  &11&12  \\ \hline
		{\tt CT}& 176.3 & 418.4 & 423.3 & 1239.4 & 1468.3 & 3621.2 & 7353.8 \\ \hline
		{\tt CTviaCR} & 144.1 & 290.5 & 345.8 & 680.2 & 729.6 & 1475.2 & 1639.9  \\ \hline
	\end{tabular}
\captionof{table}{Telescoper computation at higher orders via the refined method}
    \label{tab:CT3}
\end{center}

These timings illustrate that complete reductions in the setting of difference rings are significantly more efficient for large-scale problems. We therefore expect that calculations in particle physics, for instance, will open up new avenues for simplifying Feynman integrals once they have been transformed into multiple sums.
In~\cite{QCD20}, these sums, which reached summands requiring up to 1 GB of memory, were evaluated using the classical difference ring tools in \texttt{Sigma}. As future calculations are expected to involve even larger expressions, the complete reduction techniques presented here will render such challenging computations feasible for the first time.
\section{Concluding remarks}\label{SECT:Conclusion}

In this paper, we introduced a general framework for constructing complete reductions in $R\Pi\Sigma^*$-towers, assuming that the ground difference ring $(\bA,\,\sigma)$ satisfies certain algorithmic properties. Specifically, we provided explicit constructions of complete reductions for the rational difference field $\bA=C(x)$ with $\sigma(x)=x+1$. Extending this framework to the mixed case~\cite{BaPe1999}—including the $q$-rational difference field with $\sigma(x) = qx$ where $q \in C^*$ is not a root of unity—and to general $R\Pi\Sigma^*$-towers over $\Pi\Sigma^*$-fields \cite{Karr1981} will significantly broaden the scope and applicability of this machinery.

Based on additive decomposition, criteria for the existence of telescopers have been established for the bivariate hypergeometric case~\cite{Abra2003}, the $q$-hypergeometric case~\cite{CHM2005}, and the mixed case~\cite{CCFFL2015}. A natural question is how to leverage complete reductions in $R\Pi\Sigma^*$-towers to derive an existence criterion for telescopers of elements in such extensions.

In Example~\ref{EX:newidentity}, we showed that $\lfloor \frac{n}{2} \rfloor !$ can be modeled in a non-simple $R\Pi\Sigma^*$-tower. This opens the way to representing compositions of summation expressions and floor or ceiling functions via $R\Pi\Sigma^*$-tower constructions. In particular, the double factorial 
\[
n!! = \left\{\begin{array}{ll}
	2^k k!\ & n=2k,\\
	\frac{(2k+1)!}{2^k k!} & n=2k+1,
\end{array}
\right. \]
which is defined by interlacing two hypergeometric terms, can be rewritten as 
$$\frac{1+(-1)^n}{2} 2^{\lfloor \frac{n}{2} \rfloor} \lfloor \frac{n}{2} \rfloor ! +\frac{1-(-1)^n}{2} \frac{n!}{2^{\lfloor \frac{n}{2} \rfloor} \lfloor \frac{n}{2} \rfloor !}.$$
Together with Example~\ref{EX:newidentity} and the identity
$$2^{\lfloor \frac{n}{2} \rfloor} = \frac{1+(-1)^n}{2} (\sqrt{2})^n + \frac{1-(-1)^n}{2}\frac{(\sqrt{2})^n}{\sqrt{2}},$$ 
we obtain that $n!!$ can be modeled in the $R\Pi\Sigma^*$-tower  $(\bQ(\sqrt2)(x)\langle y \rangle \langle p_1 \rangle \langle p_2 \rangle \langle p_3 \rangle,\,\sigma)$ over $(\bQ(\sqrt2)(x),\,\sigma)$ equipped with
\[\sigma(y)=-y, \, \sigma(p_1)=\sqrt{2}p_1, \, \sigma(p_2)=(x+1)p_2 \,\, \text{and} \,\, \sigma(p_3)= ap_3,\]
where $a=\frac {x+3-y(x-1)}{4}$.
A future task will be the automatic representation
of such interlaced products and sums within the general framework of $R\Pi\Sigma^*$-extensions.
One prominent application of this framework including interlacing is the representation of Liouvillian sequences. These extend d'Alembertian sequences~\cite{Abramov:94} --which are constructed via nested indefinite sums and hypergeometric products -- by additionally incorporating the interlacing operator. Introduced by Hendriks and Singer~\cite{HeSi1999}, this solution class is fundamental: the solution space of a linear recurrence with polynomial coefficients admits a basis of Liouvillian sequences if and only if its difference Galois group is solvable. Similarly to the optimal representation of d'Alembertian sequences in the setting of difference rings~\cite{Schn2021,Schneider:2023,Schneider:26}, it will be an important task to focus on the algebraic simplification of Liouvillian solutions. 
For formalizing their representation within difference rings, one may utilize the result of~\cite{Reut2012,Petkov:2013} that every Liouvillian sequence can be expressed as an interlacing of d'Alembertian sequences and the algorithms from~\cite{Schn2020,OS2024,Schn2021} that any d'Alembertian sequence can be represented in a simple $R\Pi\Sigma^*$-tower. 

\subsection*{Acknowledgment} 
This research was funded in whole or in part by the
Austrian Science Fund (FWF) grants DOI 10.55776/P33530, DOI 10.55776/PAT1332123, and DOI 10.55776/I6130.


\newcommand{\Gathen}{\relax}\newcommand{\Hoeij}{\relax}\newcommand{\Hoeven}{\relax}\def\cprime{$'$}
\def\cprime{$'$} \def\cprime{$'$} \def\cprime{$'$} \def\cprime{$'$}
\def\cprime{$'$} \def\cprime{$'$} \def\cprime{$'$} \def\cprime{$'$}
\def\polhk#1{\setbox0=\hbox{#1}{\ooalign{\hidewidth
			\lower1.5ex\hbox{`}\hidewidth\crcr\unhbox0}}} \def\cprime{$'$}

\end{document}